\documentclass[a4paper]{article}

\usepackage[left=3cm,right=3cm,top=3cm,bottom=3cm]{geometry}
\usepackage{authblk}

\usepackage{amsmath}
\usepackage{amssymb}
\usepackage{amsfonts}
\usepackage{amsthm}
\usepackage{adjustbox}
\usepackage{caption}
\usepackage{graphicx}
\usepackage{hyperref}
\usepackage{cleveref}

\usepackage{enumerate}
\usepackage[shortlabels]{enumitem}

\DeclareMathOperator{\alphabet}{\textsf{alph}}

\DeclareMathOperator{\width}{\textsf{w}}
\DeclareMathOperator{\pathwidth}{\textsf{pw}}
\DeclareMathOperator{\cutwidth}{\textsf{cw}}
\DeclareMathOperator{\socutwidth}{\textsf{cw}_2}

\DeclareMathOperator{\cutedges}{\mathcal{C}}

\DeclareMathOperator{\pset}{\textsf{ps}}
\DeclareMathOperator{\loc}{\textsf{loc}}
\DeclareMathOperator{\condensed}{\textsf{cond}}
\DeclareMathOperator{\bigo}{O}

\DeclareMathOperator{\open}{\texttt{open}}
\DeclareMathOperator{\act}{\texttt{active}}
\DeclareMathOperator{\closed}{\texttt{closed}}

\DeclareMathOperator{\opensym}{\texttt{o}}
\DeclareMathOperator{\actsym}{\texttt{a}}
\DeclareMathOperator{\closesym}{\texttt{c}}

\DeclareMathOperator{\fpt}{\textsf{fpt}}

\newcommand{\alphagraph}{G_{\alpha}}
\newcommand{\betagraph}{G_{\beta}}

\DeclareMathOperator{\greedy}{\texttt{GREEDY}}

\DeclareMathOperator{\SOstrategy}{\textsf{FewOcc}}
\DeclareMathOperator{\MOstrategy}{\textsf{ManyOcc}}
\DeclareMathOperator{\SNMstrategy}{\textsf{FewBlocks}}
\DeclareMathOperator{\LRstrategy}{\textsf{LeftRight}}
\DeclareMathOperator{\BEstrategy}{\textsf{BE}}
\DeclareMathOperator{\BEstrategyAltOne}{\textsf{BE}-1}
\DeclareMathOperator{\BEstrategyAltTwo}{\textsf{BE}-2}

\DeclareMathOperator{\opt}{\textsf{opt}}

\newcommand{\MN}{\mathbb{N}}

\newcommand{\NP}{\textsf{NP}}
\newcommand{\pclass}{\textsf{P}}

\newcommand{\locProb}{\textsc{Loc}}
\newcommand{\minlocProb}{\textsc{MinLoc}}

\newcommand{\cutwidthProb}{\textsc{Cutwidth}}
\newcommand{\pathwidthProb}{\textsc{Pathwidth}}
\newcommand{\minPathwidthProb}{\textsc{MinPathwidth}}
\newcommand{\minCutwidthProb}{\textsc{MinCutwidth}}

\newcommand{\ta}{\mathtt{a}}
\newcommand{\tb}{\mathtt{b}}
\newcommand{\tc}{\mathtt{c}}
\newcommand{\td}{\mathtt{d}}
\newcommand{\te}{\mathtt{e}}

\newcommand{\tx}{\mathtt{x}}
\newcommand{\ty}{\mathtt{y}}
\newcommand{\tz}{\mathtt{z}}

\newcommand{\maxocc}[1]{|#1|_{\mathsf{maxocc}}}

\newcommand{\N}{\mathbb{N}}

\newtheorem{lemma}{Lemma}[section]
\newtheorem{theorem}[lemma]{Theorem}
\newtheorem{proposition}[lemma]{Proposition}
\newtheorem{corollary}[lemma]{Corollary}

\theoremstyle{definition}

\newtheorem{remark}[lemma]{Remark}

\newtheorem{observation}[lemma]{Observation}
\newtheorem{openproblem}[lemma]{Open Problem}

\AtBeginDocument{%
  \providecommand\BibTeX{{%
    \normalfont B\kern-0.5em{\scshape i\kern-0.25em b}\kern-0.8em\TeX}}}

\begin{document}

\title{Graph and String Parameters: Connections Between Pathwidth, Cutwidth and the Locality Number}

\author[1]{Katrin Casel}
\author[2]{Joel D. Day}
\author[3]{Pamela Fleischmann}
\author[4]{Tomasz Kociumaka}
\author[5]{Florin Manea}
\author[1]{Markus L.\ Schmid}

\affil[1]{Humboldt-Universit\"at zu Berlin, Berlin, Germany, \texttt{Katrin.Casel@hu-berlin.de}, \texttt{MLSchmid@MLSchmid.de}}
\affil[2]{Department of Computer Science, Loughborough University, Loughborough, United Kingdom, \texttt{J.Day@lboro.ac.uk}}
\affil[3]{Department of Computer Science, Kiel University, Kiel, Germany, \texttt{fpa@informatik.uni-kiel.de}}
\affil[4]{Max Planck Institute for Informatics, Saarland Informatics Campus, Saarbr\"ucken, Germany, \texttt{tomasz.kociumaka@mpi-inf.mpg.de}}
\affil[5]{Computer Science Department, Universit\"at G\"ottingen, G\"ottingen, Germany, \texttt{florin.manea@informatik.uni-goettingen.de}}

\date{}

\maketitle

\begin{abstract}
We investigate the locality number, a recently introduced structural parameter for strings (with applications in pattern matching with variables), and its connection to two important graph-parameters, cutwidth and pathwidth. These connections allow us to show that computing the locality number is $\NP$-hard, but fixed-parameter tractable, if parameterised by the locality number or by the alphabet size, which has been formulated as open problems in the literature. Moreover, the locality number can be approximated with ratio $\bigo(\sqrt{\log(\opt)} \log(n))$. \par

An important aspect of our work -- that is relevant in its own right and of independent interest -- is that we identify connections between the string parameter of the locality number on the one hand, and the famous graph parameters of cutwidth and pathwidth, on the other hand. These two parameters have been jointly investigated in the literature and are arguably among the most central graph parameters that are based on ``linearisations'' of graphs. In this way, we also identify a direct approximation preserving reduction from cutwidth to pathwidth, which shows that any polynomial $f(\opt,|V|)$-approximation algorithm for pathwidth yields a polynomial $2f(2\opt,h)$-approximation algorithm for cutwidth on multigraphs (where $h$ is the number of edges). In particular, this translates known approximation ratios for pathwidth into new approximation ratios for cutwidth, namely $\bigo(\sqrt{\log(\opt)} \log(h))$ and $\bigo(\sqrt{\log(\opt)} \opt)$ for (multi) graphs  with $h$ edges.
\end{abstract}

\maketitle

\section{Introduction}\label{sec:intro}

Graphs, on the one hand, and strings (we also use the term \emph{word}), on the other, are two different types of data objects and they have certain particularities. Graphs seem to be more popular in fields like classical and parameterised algorithms and complexity (due to the fact that many natural graph problems are intractable), while fields like formal languages, pattern matching, verification or compression are more concerned with strings. Moreover, both the field of graph algorithms as well as string algorithms are well established and provide rich toolboxes of algorithmic techniques, but they differ in that the former is tailored to computationally hard problems (e.g., the approach of treewidth and related parameters), while the latter focuses on providing efficient data-structures for near-linear-time algorithms. Nevertheless, it is sometimes possible to bridge this divide, i.e., by ``flattening'' a graph into a sequential form, or by ``inflating'' a string into a graph, to make use of respective algorithmic techniques otherwise not applicable. This paradigm shift may provide the necessary leverage for new algorithmic approaches. 

In this paper, we are concerned with certain structural parameters (and the problems of computing them) for graphs and strings: the \emph{cutwidth} $\cutwidth(G)$ of a graph $G$ (i.e., the maximum number of ``stacked'' edges if the vertices of a graph are drawn on a straight line), the \emph{pathwidth} $\pathwidth(G)$ of a graph $G$ (i.e., the minimum width of a tree decomposition the tree structure of which is a path), and the \emph{locality number} $\loc(\alpha)$ of a string $\alpha$ (explained in more detail in the next paragraph; formal definitions follow in Section~\ref{sec:prim}). By $\cutwidthProb$, $\pathwidthProb$ and $\locProb$, we denote the corresponding decision problems (i.\,e., checking whether $\cutwidth(G) \leq k$, $\pathwidth(G) \leq k$, or $\loc(\alpha) \leq k$, respectively) and with the prefix $\textsc{Min}$, we refer to the minimisation variants (for which we are mainly interested in approximation algorithms). The two former graph-parameters are very classical. Pathwidth is a simple (yet still hard to compute) subvariant of treewidth, which measures how much a graph resembles a path. The problems $\pathwidthProb$ and $\minPathwidthProb$ are intensively studied (in terms of exact, parameterised and approximation algorithms) and have numerous applications (see the surveys and textbook~\cite{Bodlaender1998, Kloks1994, Bodlaender1993}). $\cutwidthProb$ is the best-known example of a whole class of so-called \emph{graph layout problems} (see the survey~\cite{surveyDiaz, Petit2011} for detailed information), which are studied since the 1970s and were originally motivated by questions of circuit layouts. 

The locality number is rather new and we shall discuss it in more detail. A word is $k$-local if there exists an order of its symbols such that, if we {\em mark} the symbols in the respective order (which is called a \emph{marking sequence}), at each stage there are at most $k$ contiguous blocks of marked symbols in the word. This $k$ is called the {\em marking number} of that marking sequence. The \emph{locality number} of a word is the smallest $k$ for which that word is $k$-local, or, in other words, the minimum marking number over all marking sequences. For example, the marking sequence $\sigma = (\tx, \ty, \tz)$ marks $\alpha = \tx \ty \tx \ty \tz \tx \tz$ as follows (marked blocks are illustrated by overlines): 
\begin{equation*}
\tx \ty \tx \ty \tz \tx \tz \leadsto \overline{\tx} \ty \overline{\tx} \ty \tz \overline{\tx} \tz \leadsto\overline{\tx \ty \tx \ty} \tz \overline{\tx} \tz \leadsto \overline{\tx \ty \tx \ty \tz \tx \tz}\,;
\end{equation*}
thus, the marking number of $\sigma$ is $3$. In fact, all marking sequences for $\alpha$ have a marking number of $3$, except $(\ty, \tx, \tz)$, for which it is $2$: 
\begin{equation*}
\tx \overline{\ty} \tx \overline{\ty} \tz \tx \tz \leadsto \overline{\tx \ty \tx \ty} \tz \overline{\tx} \tz \leadsto \overline{\tx \ty \tx \ty \tz \tx \tz}\,. 
\end{equation*}
Thus, the locality number of $\alpha$ is $\loc(\alpha) = 2$. For the slightly more complicated word $\beta = \ta \td \ta \tb \ta \td \tb \td \ta \te \tc \tb \tc \tb$, it can be verified that $\loc(\beta) = 3$ (see Figure~\ref{fig:markingSequence} for an illustration of two marking sequences for $\beta$ with marking numbers $4$ and $3$, respectively).

\begin{figure}
\begin{tabular}{l|l|l}
Marked word & \begin{minipage}{1.5cm}Marking\\ sequence\end{minipage} & \begin{minipage}{1.5cm}Marking\\ number\end{minipage}\\\hline
$\ta \td \ta \tb \ta \td \tb \td \ta \te \tc \tb \tc \tb$ & & $0$\\
$\ta \td \ta \overline{\tb} \ta \td \overline{\tb} \td \ta \te \tc \overline{\tb} \tc \overline{\tb}$ & $\tb$ & $4$\\
$\ta \td \ta \overline{\tb} \ta \td \overline{\tb} \td \ta \te \overline{\tc \tb \tc \tb}$ & $\tc$ & $3$\\
$\ta \td \ta \overline{\tb} \ta \td \overline{\tb} \td \ta \overline{\te \tc \tb \tc \tb}$ & $\te$ & $3$\\
$\ta \overline{\td} \ta \overline{\tb} \ta \overline{\td \tb \td} \ta \overline{\te \tc \tb \tc \tb}$ & $\td$ & $4$\\
$\overline{\ta \td \ta \tb \ta \td \tb \td \ta \te \tc \tb \tc \tb}$ & $\ta$ & $1$
\end{tabular}
\hspace{20pt}
\begin{tabular}{l|l|l}
Marked word & \begin{minipage}{1.5cm}Marking\\ sequence\end{minipage} & \begin{minipage}{1.5cm}Marking\\ number\end{minipage}\\\hline
$\ta \td \ta \tb \ta \td \tb \td \ta \te \tc \tb \tc \tb$ & & $0$\\
$\ta \overline{\td} \ta \tb \ta \overline{\td} \tb \overline{\td} \ta \te \tc \tb \tc \tb$ & $\td$ & $3$\\
$\overline{\ta \td \ta} \tb \overline{\ta \td} \tb \overline{\td \ta} \te \tc \tb \tc \tb$ & $\ta$ & $3$\\
$\overline{\ta \td \ta \tb \ta \td \tb \td \ta} \te \tc \overline{\tb} \tc \overline{\tb}$ & $\tb$ & $3$\\
$\overline{\ta \td \ta \tb \ta \td \tb \td \ta} \te \overline{\tc \tb \tc \tb}$ & $\tc$ & $2$\\
$\overline{\ta \td \ta \tb \ta \td \tb \td \ta \te \tc \tb \tc \tb}$ & $\te$ & $1$
\end{tabular}
\caption{An illustration of the marking sequence $(\tb, \tc, \te, \td, \ta)$ with marking number of $4$ for the word $\beta = \ta \td \ta \tb \ta \td \tb \td \ta \te \tc \tb \tc \tb$ (left side), and the marking sequence $(\td, \ta, \tb, \tc, \te)$ with marking number of $3$ for the word $\beta$ (right side).}
\label{fig:markingSequence}
\end{figure}

The locality number has applications in pattern matching with variables~\cite{FSTTCS}. A \emph{pattern} is a word that consists of \emph{terminal symbols}  (e.g., $\ta, \tb, \tc$), treated as constants, and \emph{variables} (e.g., $x_1, x_2, x_3, \ldots $). A pattern is mapped to a word by substituting the variables by strings of terminals. For example, $x_1 x_1 \tb \ta \tb x_2 x_2$ can be mapped to $\ta \tc \ta \tc \tb \ta \tb \tc \tc$ by the substitution $(x_1 \to \ta \tc, x_2 \to \tc)$. Deciding whether a given pattern matches (i.e., can be mapped to) a given word is one of the most important problems that arise in the study of patterns with variables (note that the concept of patterns with variables arises in several different domains like combinatorics on words (word equations~\cite{kar:the}, unavoidable patterns~\cite{Loth02}), pattern matching~\cite{ami:gen}, language theory~\cite{ang:fin2}, learning theory~\cite{ang:fin2,erl:lea,ng:dev,rei:dis,kea:apo,FeMaMeSc16_TCS}, database theory~\cite{bar:exp}, as well as in practice, e.g., extended regular expressions with backreferences~\cite{Fre2013,FreydenbergerSchmid2017,Schmid2016,fri:mas}, used in programming languages like Perl, Java, Python, etc.). Unfortunately, the \emph{matching problem} is $\NP$-complete \cite{ang:fin2} in general (it is also $\NP$-complete for strongly restricted variants~\cite{FerSch2015, FeMaMeSc14_stacs} and also intractable in the parameterised setting~\cite{FerSchVil2015}); see also~\cite{ManeaSchmid2019} for a survey. As demonstrated in~\cite{rei:patIaC}, for the matching problem a paradigm shift as sketched in the first paragraph above yields a very promising algorithmic approach. 
More precisely, any class of patterns with bounded treewidth (for suitable graph representations) can be matched in polynomial-time. However, computing (and therefore algorithmically exploiting) the treewidth of a pattern is difficult (see the discussion in~\cite{FeMaMeSc14_stacs,rei:patIaC}), which motivates more direct string-parameters that bound the treewidth and are simple to compute (virtually all known structural parameters that lead to tractability~\cite{FSTTCS, FeMaMeSc14_stacs,rei:patIaC, shi:pol2} are of this kind (the efficiently matchable classes investigated in~\cite{DayEtAl2018ta} are one of the rare exceptions)). This also establishes an interesting connection between ad-hoc string parameters and the more general (and much better studied) graph parameter treewidth. The locality number is a simple parameter directly defined on strings, it bounds the treewidth and the corresponding marking sequences can be seen as instructions for a dynamic programming algorithm. However, compared to other ``tractability-parameters'', it seems to cover best the treewidth of a string, but whether it can be efficiently computed is unclear.

In this paper, we investigate the problem of computing the locality number (in the exact sense as well as fixed-parameter algorithms and approximations) and, by doing so, we establish an interesting connection to the graph parameters cutwidth and pathwidth with algorithmic implications for approximating cutwidth. In the following, we first discuss related results in more detail and then outline our respective contributions.

Note that a conference version of this paper has been published in ICALP~2019~\cite{CaselEtAl2019}.

\subsection{Known Results and Open Questions} \label{sec:knownResults}

For $\locProb$, only exact exponential-time algorithms are known and whether it can be solved in polynomial-time, or whether it is at least fixed-parameter tractable is mentioned as open problems in~\cite{FSTTCS}. Approximation algorithms have not yet been considered. Addressing these questions is the main purpose of this paper. 

$\pathwidthProb$ and $\cutwidthProb$ are $\NP$-complete, but fixed-parameter tractable with respect to the standard parameters $\pathwidth(G)$ or $\cutwidth(G)$, respectively (even with ``linear'' fpt-time $g(k) \bigo(n)$~\cite{Bodlaender1996, Bodlaender2012, ThilikosSB05}). 
With respect to approximation, their minimisation variants have received a lot of attention, mainly because they yield (like many other graph parameters) general algorithmic approaches for numerous graph problems, i.e., a good linear arrangement or path-decomposition can often be used for a dynamic programming (or even divide and conquer) algorithm. More generally speaking, pathwidth and cutwidth are related to the more fundamental concepts of small balanced vertex or edge separators for graphs (i.e., a small set of vertices (or edges, respectively) that, if removed, divides the graph into two parts of roughly the same size. More precisely, $\pathwidth(G)$ and $\cutwidth(G)$ are upper bounds for the smallest balanced \emph{vertex} separator of $G$ and the smallest balanced \emph{edge} separator of $G$, respectively (see~\cite{FeigeEtAl2008} for further details and explanations of the algorithmic relevance of balanced separators). 

With respect to $\minPathwidthProb$, there is an approximation algorithm with ratio $\bigo(\log n\sqrt{\log opt})$ (see~\cite[Corollary 6.5]{FeigeEtAl2008}) and an approximation algorithm with ratio $\bigo(\mathsf{tw} \sqrt{\log \mathsf{tw}})$, where $\mathsf{tw}$ is the treewidth of the graph (see~\cite{GroenlandEtAl2023}).

For $\minCutwidthProb$, there is an $\bigo(\sqrt{\log(n)} \log(n))$ approximation algorithm. This follows from using the $\bigo(\sqrt{\log n})$-approximation for sparsest cut of ~\cite{AroraEtAl2009} as described in the cutwidth approximation of~\cite[Section 3.8]{Leighton1999}, which adds a $\log(n)$-factor.

\subsection{Our Contributions} 

There are two natural approaches to represent a word $\alpha$ over alphabet $\Sigma$ as a graph $G_{\alpha} = (V_{\alpha}, E_{\alpha})$. The first option is to represent $\alpha$'s positions as vertices, i.\,e., $V_{\alpha} = \{1, 2, \ldots, |\alpha|\}$, and then somehow use the edges to represent the actual symbols on these position. We present such a reduction to relate the locality number of words with the pathwidth of graphs. More precisely, we transform a word $\alpha$ into a graph such that $|E_{\alpha}| = \bigo(|\alpha|^2)$ and $\loc(\alpha) \leq \pathwidth(G_{\alpha}) \leq 2\loc(\alpha)$. \par
The second option is to use a vertex per symbol that occurs in $\alpha$, i.\,e., $V_{\alpha} = \Sigma$, and somehow use the edges to encode where these symbols occur in the word. By such a reduction with $|E_{\alpha}| = \bigo(|\alpha|)$, we can relate the locality number of words with the cutwidth of graphs in the sense that $\cutwidth(G_{\alpha}) = 2\loc(\alpha)$. \par

Since these reductions are parameterised reductions and also approximation preserving, known upper bounds for the problems of computing or approximating the pathwidth or cutwidth of graphs carry over to the problem of computing or approximating the locality number of words. More precisely, we can conclude that $\locProb$ is fixed-parameter tractable if parameterised by $|\Sigma|$ or by the locality number (answering the respective open problem from~\cite{FSTTCS}), and also that there is a polynomial-time $\bigo(\sqrt{\log(\opt)} \log(n))$-approximation algorithm for $\minlocProb$. \par
In addition to these reductions, we can also show how an arbitrary multi-graph $G = (V, E)$ and an edge $e \in E$ can be represented by a word $\alpha_{G, e}$ over alphabet $V$, of length $|E|$ and with $\cutwidth(G) \leq \loc(\alpha_{G, e}) \leq \cutwidth(G)+1$. Moreover, there must be an edge $e \in E$ such that $\loc(\alpha_{G, e}) = \cutwidth(G)$. This describes an (approximation preserving) Turing-reduction from $\cutwidthProb$ to $\locProb$ which allows us to conclude that $\locProb$ is $\NP$-complete (which solves the other open problem from~\cite{FSTTCS}).\par
Even though the reduction from $\minlocProb$ to $\minPathwidthProb$ yields an $\bigo(\sqrt{\log(\opt)} \log(n))$-approximation algorithm for $\minlocProb$, it is also important to directly investigate whether obvious greedy strategies for constructing marking sequences (e.\,g., always marking a symbol next that leads to the smallest number of marked blocks) yield good approximation ratios. On the one hand, if such strategies fail, we can rule them out as possible approximation algorithms for computing the locality number, and, on the other hand, if such simple strategies work, then, due to the reduction from $\minCutwidthProb$ to $\minlocProb$, this might open a new angle to the approximation of cutwidth. Unfortunately, we can formally show that many natural candidates for greedy strategies fail to yield promising approximation algorithms (and are therefore also not helpful for cutwidth approximation). \par

Expecting an improvement of cutwidth approximation -- a heavily researched area -- by translating the problem into a string problem and then investigating the approximability of this string problem seems naive. This makes it even more surprising that linking cutwidth with pathwidth via the locality number is in fact helpful for cutwidth approximation. More precisely, by plugging together our reductions from $\minCutwidthProb$ to $\minlocProb$ and from $\minlocProb$ to $\minPathwidthProb$, we obtain a reduction which directly transfers approximation results from $\minPathwidthProb$ (e.\,g., the ones of~\cite{FeigeEtAl2008, GroenlandEtAl2023}; see the discussion of Section~\ref{sec:knownResults}) to $\minCutwidthProb$. On the one hand, this reduction yields new concrete approximation ratios for cutwidth (mentioned in more detail below), but, on the other hand, it also shows that any future improvement of pathwidth approximation directly carries over to cutwidth approximation (although there is a constant factor involved for constant factor approximations of pathwidth).\footnote{Note that both the pathwidth and the cutwidth is NP-hard to approximate to within a constant factor~\cite{WuEtAl2014}.}\par
A reason why this direct reduction from cutwidth to pathwidth has been overlooked might be that the literature on cutwidth and pathwidth approximation is focussed on more general approximation techniques (i.\,e., vertex and edge separators), which then yield approximation algorithms for these graph parameters. Another reason might be that this relation is less obvious on the graph level and becomes more apparent if linked via the string parameter of locality, as in our considerations. Nevertheless, since pathwidth and cutwidth are such crucial parameters for graph algorithms, we also translate our locality based reduction into one from graphs to graphs directly.\par

We conclude this subsection by summarising the main results of this work:
\begin{itemize}
\item We present approximation preserving reductions from $\locProb$ to $\cutwidthProb$ and $\pathwidthProb$, and an approximation preserving reduction from $\cutwidthProb$ to $\locProb$.
\item $\locProb$ is $\NP$-complete, but fixed-parameter tractable if parameterised by $|\Sigma|$ or by the locality number.
\item There is a polynomial-time $\bigo(\sqrt{\log(\opt)} \log(n))$-approximation algorithm for $\minlocProb$. 
\item Many obvious greedy strategies for $\minlocProb$ do not yield good approximation algorithms. 
\item We present an approximation preserving reduction from $\cutwidthProb$ to $\pathwidthProb$.
\item There is a polynomial-time $\bigo(\sqrt{\log(\opt)} \log(n))$-approximation algorithm and a polynomial-time $\bigo(\sqrt{\log(\opt)} \opt)$-approximation algorithm for $\minCutwidthProb$.
\end{itemize}

\subsection{Organisation of the Paper}

In Section~\ref{sec:prim}, we give basic definitions (including the central parameters of the locality number, the cutwidth and the pathwidth). In the next Section~\ref{sec:ExamplesWordComb}, we discuss the concept of the locality number with some examples and some word combinatorial considerations. The purpose of this section is to develop a better understanding of this parameter for readers less familiar with string parameters and combinatorics on words (the technical statements of this section are formally proven in the appendix). \par
The main results are presented in Sections~\ref{sec:cutwidth},~\ref{sec:approx}~and~\ref{sec:direct}. First, in Section~\ref{sec:cutwidth}, we present the reductions from $\locProb$ to $\cutwidthProb$ and vice versa, and we discuss the consequences of these reductions. Then, in Section~\ref{sec:approx}, we show how $\locProb$ can be reduced to $\pathwidthProb$, which yields an approximation algorithm for computing the locality number; furthermore, we investigate the performance of direct greedy strategies for approximating the locality number. Finally, since we consider this of high importance independent of the locality number, we provide a direct reduction from cutwidth to pathwidth in Section~\ref{sec:direct}.\par
In Section~\ref{sec:conc}, we conclude the paper by discussing some related topics and possible further research questions.

\section{Preliminaries}\label{sec:prim}

We now define basic concepts of complexity theory, some basics about string algorithms and the locality number, the central graph parameters cutwidth and pathwidth, and the formal definitions of the decision and minimisation problems to be investigated.

\subsection{(Parameterised) Complexity Theory and Approximation Algorithms}

We briefly summarise the fundamentals of parameterised complexity~\cite{FlumGrohe2006, DowneyFellows2013} and approximation algorithms~\cite{Ausiello1999}.

A \emph{parameterised problem} is a decision problem with instances $(x, k)$, where $x$ is the actual input and $k \in \mathbb{N}$ is the \emph{parameter}.
A parameterised problem $P$ is \emph{fixed-parameter tractable} if there is an \emph{$\fpt$-algorithm} for it, i.e., one that solves $P$ on input $(x, k)$ in time $f(k) \cdot p(|x|)$ for a recursive function $f$ and a polynomial $p$. In this case, we also say that the parameterised problem $P$ can be solved with \emph{fpt-running-time} $f(k) \cdot p(|x|)$.

We use the $\bigo^*(\cdot)$ notation which hides multiplicative factors polynomial in $|x|$.

A minimisation problem $P$ is a triple $(I, S, m)$, where $I$ is the \emph{set of instances}, $S$ is a function that 
maps instances $x \in I$ to the \emph{set of feasible solutions} for $x$, and $m$ is the 
\emph{objective function} that maps pairs $(x,y)$ with $x \in I$ and $y \in S(x)$ to a positive 
rational number. For every $x\in I$, we denote $m^*(x) = \min\{m(x,y)\colon y\in S(x)\}$. 
For $x\in I$ and $y\in S(x)$, the value $R(x,y)=\frac{m(x, y)}{m^*(x)}$ is the \emph{performance ratio} of $y$ with respect to $x$.
An algorithm $\mathcal{A}$ is an \emph{approximation algorithm} for $P$ with ratio $r : 
\mathbb{N} \to \mathbb{Q}$ (or an $r$-approximation algorithm, for short) if, for every $x \in I$, 
$\mathcal{A}(x) = y \in S(x)$, and $R(x,y)\leq r(|x|)$.
 We also let $r$ be of 
the form $\mathbb{Q}\times \mathbb{N} \to \mathbb{Q}$ when the ratio $r$ depends on $m^*(x)$ and $|x|$; 
in this case, we write $r(\opt,|x|)$. We further assume that the function $r$ is monotonically non-decreasing. Unless stated otherwise, all approximation algorithms run in polynomial time with respect to $|x|$.

\subsection{Basic String Definitions and Locality}\label{sec:localityDefinition}

The set of strings (or words) over an alphabet $X$ is denoted by $X^*$, by $|\alpha|$ we denote the length of a word $\alpha$, $\alphabet(\alpha)$ is the smallest alphabet $X$ with $\alpha \in X^*$, and $\varepsilon$ denotes the empty word with $|\varepsilon| = 0$. A string $\beta$ is called a \emph{factor} of $\alpha$ if $\alpha = \alpha' \beta \alpha''$; if $\alpha' = \varepsilon$ or $\alpha'' = \varepsilon$, then $\beta$ is a \emph{prefix} or a \emph{suffix} of $\alpha$, respectively. For a position $j$, $1 \leq j \leq |\alpha|$, we refer to the symbol at position $j$ of $\alpha$ by the expression $\alpha[j]$, and 
$\alpha[j..j'] = \alpha[j] \alpha[j + 1] \ldots \alpha[j']$, $1\leq j \leq j' \leq |\alpha|$. For a word $\alpha$ and $x \in \alphabet(\alpha)$, let $\pset_x(\alpha) = \{i \mid 1 \leq i \leq |\alpha|, \alpha[i] = x\}$ be the set of all positions where $x$ occurs in $\alpha$. For a word $\alpha$, let $\alpha^0=\varepsilon$ and $\alpha^{i+1}=\alpha\alpha^i$ for $i\geq 0$. 

Let $\alpha$ be a word and let $X = \alphabet(\alpha) = \{x_1,x_2,\ldots, x_n\}$. A \emph{marking sequence} (\emph{over $X$}) of $\alpha$ is an enumeration, or ordering on the letters from $X$, and hence may be represented either as an ordered list of the letters or, equivalently, as a bijection $\sigma : \{1,2,\ldots,|X|\} \to \{1,2,\ldots,|X|\}$. For a marking sequence $\sigma = (x_{\sigma(1)}, x_{\sigma(2)}, \ldots, x_{\sigma(m)})$, a word $\alpha$ and every $i$ with $1 \leq i \leq m$, by \emph{stage $i$ of $\sigma$} we denote the word $\alpha$ with exactly positions $\bigcup^i_{j = 1} \pset_{x_{\sigma(j)}}(\alpha)$ marked. 
The \emph{marking number} $\pi_\sigma(\alpha)$ (of $\sigma$ with respect to $\alpha$) is the maximum number of marked blocks in any stage $i$ of $\sigma$. We say that $\alpha$ is $k$-local if and only if, for some marking sequence $\sigma$, we have $\pi_\sigma(\alpha) \leq k$, and the smallest $k$ such that $\alpha$ is $k$-local is the \emph{locality number} of $\alpha$, denoted by $\loc(\alpha)$. We say that a word $w$ is \emph{strictly} $k$-local, if $\loc(\alpha) = k$. A marking sequence $\sigma$ with $\pi_{\sigma}(\alpha) = \loc(\alpha)$ is \emph{optimal} (for $\alpha$).
For illustration, see also the examples given in Section~\ref{sec:intro}.
 
For a word $\alpha$, the \emph{condensed form of $\alpha$}, denoted by $\condensed(\alpha)$, is obtained by replacing every maximal factor $x^k$ with $x \in \alphabet(\alpha)$ by $x$. For example, $\condensed(x_1x_1x_2x_2x_2x_1x_2x_2) =x_1x_2x_1x_2$. A word $\alpha$ is \emph{condensed} if $\alpha = \condensed(\alpha)$.

\begin{observation}\label{condensedWordsObs}
For a word $\alpha \in X^*$ and any marking sequence $\sigma$ over $X$, we have $\pi_{\sigma}(\condensed(\alpha)) = \pi_{\sigma}(\alpha)$. Moreover, if $\alpha$ is condensed, then the maximum number of occurrences of any symbol in $\alpha$ is bounded by $2\loc(\alpha)$ (see~\cite{FSTTCS} for details). In particular, this means that for condensed words $\alpha \in X^*$, we have that $|\alpha| = \bigo(|X|\loc(\alpha))$.
\end{observation}

Observation~\ref{condensedWordsObs} justifies that in the following, we are only concerned with condensed words (and therefore words with at most $2\loc(\alpha)$ occurrences per symbol and total length of at most $|X|2\loc(\alpha)$). In particular, for any word $\alpha$ we can compute $\condensed(\alpha)$ in time $\bigo(|\alpha|)$; thus algorithms for computing the locality number (and the respective marking sequences) for \emph{condensed} words extend to algorithms for general words. For the sake of convenience, in the following we shall only use the term \emph{word} and keep in mind that we always talk about condensed words.

\subsection{Basic Graph Definitions and Graph Parameters}\label{sec:graphTerm}

Let $G = (V,E)$ be a (multi)graph with the vertices $V = \{v_1, \ldots, v_n\}$. A {\em cut} of $G$ is a partition $(V_1,V_2)$ of $V$ into two disjoint subsets $V_1,V_2$, $V_1 \cup V_2 = V$; the (multi)set of edges $\cutedges(V_1,V_2)=\{\{x,y\}\in E \mid x\in V_1, y\in V_2\}$ is called the \emph{cut-set} or the \emph{(multi)set of edges crossing the cut}, while $V_1$ and $V_2$ are called the \emph{sides} of the cut. The {\em size} of this cut is the number of crossing edges, i.e., $|\cutedges(V_1,V_2)|$. A \emph{linear arrangement} of the (multi)graph $G$ is a sequence $(v_{j_1}, v_{j_2}, \ldots, v_{j_n})$, where $(j_1, j_2, \ldots, j_n)$ is a permutation of $(1, 2, \ldots, n)$. For a linear arrangement $L = (v_{j_1}, v_{j_2}, \ldots, v_{j_n})$, let $L(i)=\{v_{j_1}, v_{j_2}, \ldots, v_{j_i}\}$. For every $i$, $1 \leq i < n$, we consider the cut $(L(i), V\setminus L(i))$ of $G$, and denote the cut-set $\cutedges_L(i) = \cutedges(L(i), V\setminus L(i))$
(for technical reasons, we also set $\cutedges_L(0) = \cutedges_L(n) = \emptyset$). We define the \emph{cutwidth} of $L$ by $\cutwidth(L) = \max\{|\cutedges_L(i)| \mid 0 \leq i \leq n\}$. Finally, the cutwidth of $G$ is the minimum over all cutwidths of linear arrangements of $G$, i.e., $\cutwidth(G) = \min\{\cutwidth(L) \mid L \text{ is a linear arrangement for $G$}\}$. \par
Let us discuss an example. To this end, let $H = (V, E)$ with $V = \{u, v, w, x, y, z\}$ and the edges $E$ are as illustrated in Figure~\ref{fig:exampleCutwidth}. A possible linear arrangement for $H$ is $L = (u, v, w, x, y, z)$ with $|\cutedges_L(1)| = 3$, $|\cutedges_L(2)| = 5$, $|\cutedges_L(3)| = 5$, $|\cutedges_L(4)| = 2$ and $|\cutedges_L(5)| = 2$; thus, $\cutwidth(L) = 5$ (a cut with maximum size is $(L(3), V\setminus L(3))$, as illustrated by a vertical line in Figure~\ref{fig:exampleCutwidth}). Another linear arrangement is $L' = (w, u, x, v, y, z)$ with $\cutwidth(L') = 3$ (see Figure~\ref{fig:exampleCutwidth}). Moreover, it can be verified that $\cutwidth(H) = 3$.

\begin{figure}[tb]
\begin{center}

\includegraphics{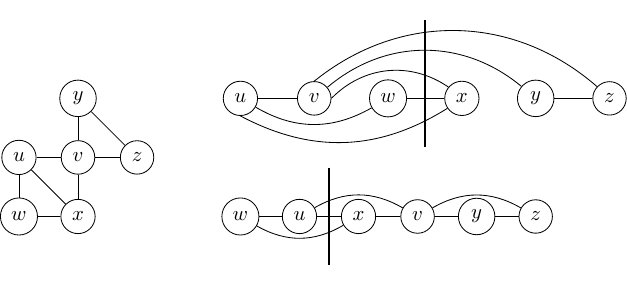}

\end{center}
\caption{A graph $H$ and two possible linear arrangements with cuts of maximum size illustrated by vertical lines.}
\label{fig:exampleCutwidth}
\end{figure}

A path decomposition (see~\cite{Bodlaender2012}) of a connected graph $G = (V, E)$ is a tree decomposition whose underlying tree is a path, i.e., a sequence $Q = (B_0, B_1, \ldots, B_{m})$ (of \emph{bags}) with $B_i \subseteq V$, $0 \leq i \leq m$, satisfying the following two properties:
\begin{itemize}
\item \emph{Cover property}: for every $\{u, v\} \in E$, there is an index $i$, $0 \leq i \leq m$, with $\{u, v\} \subseteq B_i$.
\item \emph{Connectivity property}: for every $v \in V$, there exist indices $i_v$ and $j_v$, $0\leq i_v\leq j_v\leq m$, such that $\{j \mid v \in B_j\}=\{i\mid i_v\leq i\leq j_v\}$. In other words, the bags that contain $v$ occur on consecutive positions in $(B_0, \ldots, B_{m})$.
\end{itemize}
The \emph{width} of a path decomposition $Q$ is $\width(Q) = \max\{|B_i| \mid 0 \leq i \leq m\} - 1$, and the \emph{pathwidth} of a graph ${G}$ is $\pathwidth({G}) = \min\{\width(Q) \mid Q \text{ is a path decomposition of } {G}\}$. A path decomposition is \emph{nice} if $B_0 = B_{m} = \emptyset$ and, for every $i$, $1 \leq i \leq m$, either $B_i = B_{i-1} \cup \{v\}$ or $B_i = B_{i-1} \setminus \{v\}$, for some $v \in V$. We further use $|Q|=\sum_{i=1}^n |B_i|$.\par
For example, $(\{u, w, x\}, \{u, v, x\}, \{v, y, z\})$ is a width-$2$ path decomposition for the graph $H$ defined above (see also Figure~\ref{fig:exampleCutwidth}); as can be easily seen, $\pathwidth({H}) = 2$. \par
In this paper, it shall be convenient to interpret path decompositions as \emph{marking schemes} of $V$ in the following way. Every vertex from $V$ can be marked as $\open$, as $\act$ or as $\closed$. Initially, every vertex is $\open$. Only $\open$ vertices can be set to $\act$, only $\act$ vertices can be set to $\closed$, and in the end of the marking scheme, all vertices must be $\closed$. In each step of the marking scheme, we allow an arbitrary number of vertices to be set from $\open$ to $\act$, and an arbitrary number of vertices to be set from $\act$ to $\closed$. Any such marking scheme translates into a sequence $Q = (B_0, B_1, \ldots, B_{m})$ with $B_i \subseteq V$, $0 \leq i \leq m$, by letting $B_i$ contain exactly the $\act$ vertices of step $i$ of the marking scheme. Obviously, $Q$ satisfies the connectivity property (this is a direct consequence from the fact that every vertex is marked $\act$ at some point and as soon as it is marked $\closed$, it is never marked as $\act$ again). If $Q$ also satisfies the cover property, then $Q$ is a path decomposition, and in this case we call the corresponding marking scheme a \emph{pd-marking scheme}. The width of a pd-marking scheme is then the maximum number of vertices which are marked $\act$ at the same time minus one. \par
For example, the path decomposition $(\{u, w, x\}, \{u, v, x\}, \{v, y, z\})$ for graph $H$ can be represented as a pd-marking scheme as illustrated in Figure~\ref{fig:examplePathwidth} (for convenience, we omit the vertex labels; see also Figure~\ref{fig:exampleCutwidth} for an illustration of $H$).\par
Both the locality number of a word and the pathwidth of a graph is defined via markings. In order to avoid confusion, we therefore use different terminology to distinguish between these two concepts (see also the terminology defined in Section~\ref{sec:localityDefinition}): The markings for words are called \emph{marking sequences}, while the markings for graphs are called \emph{pd-marking schemes}; the versions of a word during a marking sequence are called the \emph{stages} (of the marking sequence), while the different marked version of a graph during a pd-marking scheme are called the \emph{steps} (of the pd-marking scheme).

\begin{figure}[tb]
\begin{center}

\includegraphics{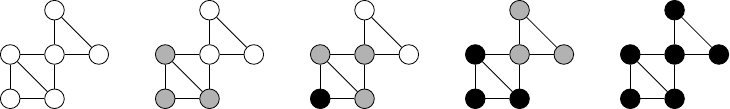}

\end{center}
\caption{The path decomposition $(\{u, w, x\}, \{u, v, x\}, \{v, y, z\})$ for graph $H$ (see Figure~\ref{fig:exampleCutwidth}) as a pd-marking scheme. White vertices are $\open$, grey vertices are $\act$, and black vertices are $\closed$. In order to see that this is a pd-marking scheme, it is sufficient to observe that for every edge there is a step in the pd-marking scheme where both incident vertices are grey (i.\,e., $\act$).}
\label{fig:examplePathwidth}
\end{figure}

\subsection{Problem Definitions}

We next formally define the computational problems of computing the parameters defined above. By $\locProb$, $\cutwidthProb$ and $\pathwidthProb$, we denote the problems to check for a given word $\alpha$ or graph $G$ and integer $k \in \mathbb{N}$, whether $\loc(\alpha) \leq k$, $\cutwidth(G) \leq k$, and $\pathwidth(G) \leq k$, respectively. Note that since we can assume that $k \leq |\alpha|$ and $k \leq |G|$, whether $k$ is given in binary or unary has no impact on the complexity. With the prefix $\textsc{Min}$, we refer to the minimisation variants. More precisely, $\minlocProb = (I, S, m)$, where $I$ is the set of words, $S(\alpha)$ is the set of all marking sequences for $\alpha$ and $m(\alpha, \sigma) = \pi_{\sigma}(\alpha)$ (note that $m^*(\alpha) = \loc(\alpha)$); $\minCutwidthProb=(I,S,m)$, where $I$ are all multigraphs, $S(G)$ is the set of linear arrangements of $G$, and $m(G,L)=\cutwidth(L)$ (note that $m^*(G)=\cutwidth(G)$); finally, $\minPathwidthProb=(I,S,m)$, where $I$ are all graphs, $S(G)$ is the set of path decompositions of $G$, and $m(G,Q)=\width(Q)$ (note that $m^*(G)=\pathwidth(G)$). 

\section{Examples and Word Combinatorial Considerations}\label{sec:ExamplesWordComb}

In this section, we discuss some examples that illustrate the concepts of marking sequences and the locality number, and we also discuss some word combinatorial properties related to the locality number. Note that for illustration purposes, the example words considered in this section are not necessarily condensed. \par
It is easy to see that $1$-locality implies some sort of \emph{palindromic} structure of a word. For example, palindromes like the English words {\em radar}, {\em refer} and {\em rotator} are obviously $1$-local, while the palindrome $\ta \tb \ta \tb \ta \tb \ta$ is obviously not $1$-local. Moreover, also $1$-local non-palindromes, like the word {\em blender}, have some palindromic structure. More precisely, it can be shown that a word $w$ is $1$-local if and only if $w \in \{a_1\}^* \{a_2\}^* \ldots \{a_n\}^* \{a_{n-1}\}^* \ldots \{a_1\}^*$, such that $\{a_1, a_2, \ldots, a_i\} \cap \{a_{i + 1}, a_{i + 2}, \ldots, a_n\} = \emptyset$ for every $1 \leq i \leq n$. An alternative equivalent point of view is that $1$-local words are necessarily of the form $y^\ell \alpha y^r$, where $\alpha$ is $1$-local with $y \notin \alphabet(\alpha)$. For further details, we refer to~\cite{FSTTCS}, where the structure of $1$-local and $2$-local words is characterised.\par

Determining structural properties that lead to high locality is more challenging. The Finnish word {\em tutustuttu} (perfect passive of \emph{tutustua}---to meet) is $4$-local, while 
\begin{quote}{\em pneumonoultramicroscopicsilicovolcanoconiosis}\end{quote} 
is an (English) $8$-local word,
and 
\begin{quote}{\em lentokonesuihkuturbiinimoottoriapumekaanikkoaliupseerioppilas}\end{quote}
 is a
$10$-local (Finnish) word. In general, in order to have a high locality number, a word needs to contain many alternating occurrences of (at least) two letters. For instance, $(x_1x_2)^n$ is $n$-local. However, the number of occurrences of a letter alone is not always a good indicator of the locality of a word. The German word {\em Einzelelement} (a basic component of a construction) has 5 occurrences of {\em e}, but is only $3$-local, as witnessed by marking sequence ({\em{l,m,e,i,n,z,t}}):
\begin{align*}
&\mathtt{E} \mathtt{i} \mathtt{n} \mathtt{z} \mathtt{e} \mathtt{l} \mathtt{e} \mathtt{l} \mathtt{e} \mathtt{m} \mathtt{e} \mathtt{n} \mathtt{t}& &\leadsto&
&\mathtt{E} \mathtt{i} \mathtt{n} \mathtt{z} \mathtt{e} \overline{\mathtt{l}} \mathtt{e} \overline{\mathtt{l}} \mathtt{e} \mathtt{m} \mathtt{e} \mathtt{n} \mathtt{t}& &\leadsto&
&\mathtt{E} \mathtt{i} \mathtt{n} \mathtt{z} \mathtt{e} \overline{\mathtt{l}} \mathtt{e} \overline{\mathtt{l}} \mathtt{e} \overline{\mathtt{m}} \mathtt{e} \mathtt{n} \mathtt{t}& &\leadsto&
&\overline{\mathtt{E}} \mathtt{i} \mathtt{n} \mathtt{z} \overline{\mathtt{e} \mathtt{l} \mathtt{e} \mathtt{l} \mathtt{e} \mathtt{m} \mathtt{e}} \mathtt{n} \mathtt{t}& &\leadsto& \\
&\overline{\mathtt{E} \mathtt{i}} \mathtt{n} \mathtt{z} \overline{\mathtt{e} \mathtt{l} \mathtt{e} \mathtt{l} \mathtt{e} \mathtt{m} \mathtt{e}} \mathtt{n} \mathtt{t}& &\leadsto&
&\overline{\mathtt{E} \mathtt{i} \mathtt{n}} \mathtt{z} \overline{\mathtt{e} \mathtt{l} \mathtt{e} \mathtt{l} \mathtt{e} \mathtt{m} \mathtt{e} \mathtt{n}} \mathtt{t}& &\leadsto&
&\overline{\mathtt{E} \mathtt{i} \mathtt{n} \mathtt{z} \mathtt{e} \mathtt{l} \mathtt{e} \mathtt{l} \mathtt{e} \mathtt{m} \mathtt{e} \mathtt{n}} \mathtt{t}& &\leadsto&
&\overline{\mathtt{E} \mathtt{i} \mathtt{n} \mathtt{z} \mathtt{e} \mathtt{l} \mathtt{e} \mathtt{l} \mathtt{e} \mathtt{m} \mathtt{e} \mathtt{n} \mathtt{t}}& &&
\end{align*}

For this example marking sequence, it is worth noting that marking the many occurrences of $e$ joins several individual marked blocks into one marked block. This also intuitively explains the correspondence between the locality number and the maximum number of occurrences per symbol (in condensed words): if there are $2k$ occurrences of a symbol, then, by marking this symbol either at least $k$ new marked blocks are created, or at least $k$ marked blocks must already exist before marking this symbol (see Observation~\ref{condensedWordsObs}). \par
A repetitive structure often leads to high locality. For example, note that {\em tutustuttu} from above is nearly a repetition. Regarding the question of how repetitions of a word affect its locality number, we can show the following result (see the Appendix for a proof).

\begin{lemma}
Let $w=u^i$ be the $i$-times repetition of $u\in X^{\ast}$ and $i\in\N$. If $u$ is strictly $k$-local then $\loc(w) \in \{ik-i+1, ik\}$.
\end{lemma}

The well-known \emph{Zimin words} \cite{Loth02} also have high locality numbers compared to their lengths. These words are important in the domain of avoidability, 
as it was shown that a terminal-free pattern is unavoidable (i.e., it occurs in every infinite word over a large enough finite alphabet) if and only if it occurs in a Zimin word. 
The  Zimin words $Z_i$, for $i\in\MN$, are inductively defined by $Z_1=x_1$ and $Z_{i+1}=Z_ix_{i+1}Z_i$. 
Clearly, $|Z_i| = 2^i-1$ for all $i\in\MN$. 
Regarding the locality of $Z_i$, note that marking $x_2$ leads to $2^{i-2}$ marked blocks; further, marking $x_1$ first and then the remaining symbols in an arbitrary order only extends or joins marked blocks. Thus, we obtain a sequence with marking number $2^{i-2}$. In fact, with respect to the locality of Zimin words, we can show the following result (see the Appendix for a proof).

\begin{lemma}
$\loc(Z_i)=\frac{|Z_i|+1}{4}=2^{i-2}$ for $i\in\MN_{\geq 2}$.
\end{lemma}

Notice that both Zimin words and $1$-local words have an obvious palindromic structure. However, in the Zimin words, the letters occur multiple times, but not in large blocks, while in $1$-local words there are at most $2$ blocks of each letter. With respect to palindromes, we can show the following general result (see the Appendix for a proof).

\begin{lemma}
If $w$ is a palindrome, with $w=uau^R$ or $w=uu^R$ ($u^R$ denotes the reversal of~$u$), and $\loc(u)=k$, then $\loc(w)\in \{2k-1,2k,2k+1\}$.
\end{lemma}

\section{Locality and Cutwidth}\label{sec:cutwidth}

In this section, we introduce polynomial-time reductions from the problem of computing the locality number of a word to the problem of computing the cutwidth of a graph, and vice versa. This establishes a close relationship between these two problems (and their corresponding parameters), which lets us derive several upper and lower complexity bounds for $\locProb$. We also discuss the approximation-preserving properties of our reductions, which shall be important later on.

\subsection{Reducing Locality Number to Cutwidth}\label{sec:locToCutwidth}

First, we show a reduction from $\locProb$ to $\cutwidthProb$. For a word $\alpha$ and an integer $k\in \mathbb{N}$, we build a multigraph $H_{\alpha, k} = (V,E)$ whose set of nodes $V = \alphabet(\alpha) \cup \{\$, \#\}$ consists of symbols occurring in $\alpha$ and two additional characters $\$,\#\notin\alphabet(\alpha)$. The multiset of edges $E$ is constructed as follows. We scan through the word from left to right, and for every individual occurrence of a size-$2$ factor $x y$ (often called \emph{digrams} in the literature on strings), we add an edge between vertices $x$ and $y$. Since the edges are undirected, this means that both an occurrence of $x y$ and an occurrence of $y x$ will cause the addition of and edge $\{x, y\}$. Moreover, independent of $\alpha$'s structure, we add $2k$ edges between $\$$ and $\#$, one edge between $\$$ and the first letter of $\alpha$, and one edge between $\$$ and the last letter of $\alpha$.\par

Let us clarify this reduction with an example. Let $\alpha = \ta \tb \tc \tb \tc \td \tb \ta \td \ta$ and $k = 2$. This means that $H_{\alpha,k} = (V, E)$ with $V = \{\ta, \tb, \tc, \td, \$, \#\}$. There are $2k = 4$ edges between vertices $\$$ and $\#$. Moreover, two edges connect $\$$ with the first and last letter of $\alpha$, respectively, which, in this case, is both the letter $\ta$, which means there are two edges between $\$$ and $\ta$. The edges that actually depend on $\alpha$'s structure are obtained by scanning through $\alpha$ from left to right with a window of size $2$, and adding the respective edge for each size-$2$ factor that we see in this window. For our example, we thus add an edge between $\ta$ and $\tb$ due to factor $\ta \tb$, then an edge between $\tb$ and $\tc$ due to the factor $\tb \tc$, then another edge between $\tb$ and $\tc$ due to the factor $\tc \tb$, and so on. This results in the multigraph shown on the left of Figure~\ref{fig:cut-loc}.

\begin{figure}
\begin{center}

\includegraphics{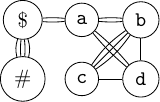}
\hspace{1cm}
\includegraphics{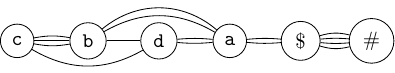}

\end{center}
\caption{The graph $H_{\alpha, k}$ for $\alpha = \ta \tb \tc \tb \tc \td \tb \ta \td \ta$ and $k = 2$; an optimal linear arrangement of $H_{\alpha,k}$ with cutwidth $4$ induces the optimal marking sequence $(\tc, \tb, \td, \ta)$ for $\alpha$ with marking number $2$.}
\label{fig:cut-loc}
\end{figure}

We claim that $\loc(\alpha) \leq k$ if and only if the smallest cutwidth of any linear arrangement of the graph $H_{\alpha,k}$ is exactly $2k$. Before formally proving this, let us discuss this on an intuitive level with our example. It is not hard to see that any linear arrangement of $H_{\alpha, k}$ with cutwidth $2k$ must start with vertices $\#, \$$ or end with vertices $\$, \#$ (this is due to the edges between $\#$ and $\$$). This means that any linear arrangement of $H_{\alpha, k}$ with cutwidth $2k$ induces a marking sequence for $\alpha$. The optimal linear arrangement of $H_{\alpha, k}$ with cutwidth $4$ shown on the right side of Figure~\ref{fig:cut-loc} thus induces the marking sequence $(\tc, \tb, \td, \ta)$. If we now carry out this marking sequence on $\alpha$ and in parallel move through $H_{\alpha, k}$'s linear arrangement from left to right, we can see that the vertices to the left of our current position correspond to marked symbols, the vertices to the right of our current position correspond to unmarked symbols, and there is exactly one crossing edge per boundary between a marked and an unmarked block in $\alpha$ (in particular, all non-crossing edges to the left and all non-crossing edges to the right correspond to size-$2$ factors that are completely contained in a marked block or an unmarked block, respectively). This means that the cuts of the linear arrangement have size twice the current number of marked blocks in the marked version of $\alpha$; thus, the size of any cut is at most $2k$, where $k$ is the marking number of the corresponding marking sequence. Note, however, that we need the edges between $\$$ and the first and last symbol to also enforce two crossing edges per marked prefix or suffix, and the $2k$ edges between $\$$ and $\#$ force the cutwidth to be at least $2k$. In particular, for our example, the linear arrangement has a cutwidth of $4$, while the corresponding marking sequence has a marking number of $2$ (which is optimal for $\alpha$). We shall now provide a formal proof for our claim.

\begin{lemma}\label{loc-cut}
The graph $H_{\alpha,k}$ satisfies $\cutwidth(H_{\alpha, k})=2k$ if and only if $\loc(\alpha) \leq k$.
\end{lemma}

\begin{proof}
    Suppose firstly that $\alpha$ is $k$-local, and let $\sigma = 
    (x_1,x_2,\ldots,x_n)$ be an optimal marking sequence of $\alpha$. Consider the 
    linear arrangement $L=(x_1,x_2,\ldots, x_n, \$,\#)$. We observe that $|\cutedges(\{x_1,x_2,\ldots,x_n,\$\},\{\#\})| = 2k$ and $|\cutedges(\{x_1,x_2,\ldots,x_n\},\{\$,\#\})| = 2$. Now consider a cut $(K_1, K_2)$ with $K_1 = \{x_1,x_2,\ldots,x_i\}$ and $K_2 = \{x_{i+1},\ldots, x_n, \$,\#\})$ for $1\leq i <n$. 
    Every edge $e \in \cutedges(K_1, K_2)$ is of the form $\{x_j, x_h\}$ with $j\leq i <h$, or of the form $\{\alpha[1], \$\}$ or $\{\$, \alpha[|\alpha|]\}$. Consequently, every edge $e \in \cutedges(K_1, K_2)$ corresponds to a unique factor $x_j x_h$ or $x_h x_j$ of $\alpha$ with $j\leq i <h$ and, after exactly the symbols $x_1,x_2,\ldots, x_i$ are marked, $x_j$ is marked and $x_h$ is not, or to a unique factor $\alpha[1]$ or $\alpha[|\alpha|]$ and, after exactly the symbols $x_1,x_2,\ldots, x_i$ are marked, $\alpha[1]$ or $\alpha[|\alpha|]$ is marked. Since there can be at most $k$ marked blocks in $\alpha$ after marking the symbols $x_1,\ldots,x_i$, there are at most $2k$ such factors, which means that $|\cutedges(K_1,K_2)| \leq 2k$.
    Thus $\cutwidth(H_{\alpha, k})\leq 2k$. Note that any linear arrangement must at some 
    point separate the nodes $\$$ and $\#$, meaning $\cutwidth(H_{\alpha, k})\geq 
    2k$, so we get that $\cutwidth(H_{\alpha, k})=2k$.

    Now suppose that the cutwidth of $H_{\alpha, k}$ is $2k$ and let $L$ be an optimal 
    linear arrangement witnessing this fact. Firstly, we note that $L$ must either 
    start with $\#$ followed by $\$$ (i.e., have the form $(\#,\$, \ldots)$)  or end 
    with $\#$ preceded by $\$$ (i.e., have the form $(\ldots, \$,\#)$. Otherwise, 
    since $H_{\alpha, k}$ is connected, every cut separating $\$$ and $\#$ would be of 
    size strictly greater than $2k$. Since a linear ordering and its mirror image 
    have the same cutwidth, we may assume that the optimal linear arrangement 
    has the form $L=(x_{\tau(1)}, 
    x_{\tau(2)},\ldots,x_{\tau(n)},\$,\#)$ for some permutation $\tau$ of 
    $\{1,\ldots,n\}$. Let $\sigma$ be the marking sequence $(x_{\tau(1)}, 
    x_{\tau(2)},\ldots,x_{\tau(n)})$ of $\alpha$ induced by $\tau$. Suppose, for 
    contradiction, that for some~$i$, with $1 \leq i < n$, after marking 
    $x_{\tau(1)},\ldots,x_{\tau(i)}$, we have $k^\prime > k$ marked blocks. Furthermore, let $K_1 = \{x_{\tau(1)},\ldots,x_{\tau(i)}\}$ and $K_2 = \{x_{\tau(i+1)}, \ldots, x_{\tau(n)}, \$, \#\}$. For every marked block $\alpha[s..t]$ that is not a prefix or a suffix of $\alpha$, we have $\alpha[s], \alpha[t] \in K_1$ and $\alpha[s - 1], \alpha[t + 1] \in K_2$ and therefore $\{\alpha[s-1], \alpha[s]\}, \{\alpha[t], \alpha[{t+1}]\} \in \cutedges(K_1, K_2)$. Moreover, for a marked prefix $\alpha[1..s]$, we have $\alpha[1], \alpha[s] \in K_1$ and $\$, \alpha[s + 1] \in K_2$ and therefore $\{\alpha[1], \$\}, \{\alpha[s], \alpha[s+1]\} \in \cutedges(K_1, K_2)$. Analogously, the existence of a marked suffix $\alpha[t..|\alpha|]$ leads to $\{\alpha[|\alpha|], \$\}, \{\alpha[t-1], \alpha[t]\} \in \cutedges(K_1, K_2)$.
    Consequently, for each marked block, we have two unique edges in $\cutedges(K_1, K_2)$, which implies $|\cutedges(K_1, K_2)| \geq 2k' > 2k$.
    This contradicts the assumption that $L$ is a witness that 
    $H_{\alpha, k}$ has cutwidth $2k$. Thus, $\alpha$ must be $k$-local.
    \end{proof}

Next we formally state how upper complexity bounds for $\cutwidthProb$ carry over to $\locProb$ via the reduction above. In particular, we formulate this result to also cover fpt-algorithms with respect to the standard parameters $\cutwidth(G)$ and $\loc(\alpha)$. The maximum degree in a multigraph $G$ is bounded from above by $2\cutwidth(G)$, so the number of nodes $n$ and the number of edges $h$ satisfy $h \le n \cdot \cutwidth(G)$. Hence, we state the complexity in terms of $n$ and $\cutwidth(G)$ rather than with respect to $h$, which is the actual input size (assuming connected graphs). 

\begin{lemma}\label{locAlgoMin}
If $\minCutwidthProb$ (resp. $\cutwidthProb$) can be solved in $\bigo(f(\cutwidth(G),n))$ time for a multigraph $G$ with $n$ vertices, then $\minlocProb$ (resp., $\locProb$) can be solved in $\bigo(f(2\loc(\alpha),{|\Sigma|+2})\allowbreak+|\alpha|)$ time for a word $\alpha$ over an alphabet~$\Sigma$.
\end{lemma}

\begin{proof}
We only show the claim for $\minCutwidthProb$; the case of $\cutwidthProb$ follows immediately from Lemma~\ref{loc-cut}. Our goal is to compute $\loc(\alpha)$ for the word $\alpha$, i.e., the minimum $k$ such that $\alpha$ is $k$-local. By Lemma~\ref{loc-cut}, we get $\cutwidth(H_{\alpha,k})=2k$
for $k\ge \loc(\alpha)$ and $\cutwidth(H_{\alpha,k})>2k$ for $k < \loc(\alpha)$. Consider the multigraph $H_{\alpha}$ obtained by removing the vertices $\#$ and $\$$ from $H_{\alpha,i}$ (the result does not depend on $i\in \N$), and observe that $2\loc(\alpha)-4\leq \cutwidth(H_\alpha)\leq 2\loc(\alpha)$. Indeed, if $\cutwidth(H_\alpha)<2\loc(\alpha)-4$, we add the two missing nodes $\#$ and $\$$ (in this order) as a prefix to an optimal linear arrangement for  $H_{\alpha}$ and get a linear arrangement of $H_{\alpha,\loc(\alpha)-1}$ of width $2\loc(\alpha)-2$, a contradiction. 

Hence, in order to determine $\loc(\alpha)$, we proceed as follows: Compute $\ell=\cutwidth(H_{\alpha})$ and iterate over integers $k$, $\frac{\ell}{2}\leq k\leq \frac{\ell+4}{2}$, in  increasing order, checking if $\cutwidth(H_{\alpha,k})=2k$.
The first value for which this equality holds equals $\loc(\alpha)$, and the marking sequence induced by the respective linear arrangement of $H_{\alpha,k}$ is an optimal one for $\alpha$ (as proved in Lemma~\ref{loc-cut}).
\end{proof}

Next, we formally state and prove the approximation preserving properties of this reduction.

\begin{lemma}\label{loc-cut-apx}
If there is an $r(\opt,h)$-approximation algorithm for $\minCutwidthProb$ running in $\bigo(f(h))$ time for an input multigraph with $h$ edges, then there is an $(r(2\opt,|\alpha|)+\frac{1}{\opt})$-approximation algorithm for $\minlocProb$ running in $\bigo(f(|\alpha|)+|\alpha|)$ time on an input word~$\alpha$.
\end{lemma}

\begin{proof}
    As already indicated in the proof of Lemma~\ref{loc-cut}, for $k=\loc(\alpha)$, every linear arrangement for $H_{\alpha,k}$ naturally translates to a marking sequence for $\alpha$. However, in an approximate linear arrangement, the vertices  $\#$ and $\$$ do not have to be at the first (or last) positions. Still, the marking sequence corresponding to the linear arrangement $L$ can have not more than  $\frac{\cutwidth(L)}{2}+1$ marked blocks, since only  suffix and prefix can be marked blocks which correspond to only one instead of two edges in a cut in $H_{\alpha,k}$. This observation remains valid if we do not include the extra vertices  $\#$ and $\$$ in  $H_{\alpha,k}$ in the reduction. Let $H_\alpha$ be the graph obtained from $H_{\alpha,k}$ (for some $k$) by removing the extra vertices $\#$ and $\$$ (observe that this also removes the dependence on $k$). Removing vertices only decreases the cutwidth, so Lemma~\ref{loc-cut} implies that  $\cutwidth(H_{\alpha})\leq 2m^*(\alpha)$. Let $\alpha$ be an instance of $\minlocProb$ and $\mathcal A$ an $r(\opt,h)$-approximation for $\minCutwidthProb$ on multigraphs. 
    The approximation algorithm $\mathcal A$ run on $H_\alpha$ returns a linear arrangement $L=\mathcal A(H_{\alpha})$ with $\cutwidth(L)\leq r(\opt,h)\cutwidth(H_{\alpha})$. Let $\sigma$ be the marking sequence corresponding to $L$, then $R(\alpha,\sigma)=\frac{\pi_\sigma(\alpha)}{m^*(\alpha)}\leq \frac{2}{\cutwidth(H_{\alpha})}(\frac{\cutwidth(L)}{2}+1)= \frac{\cutwidth(L)}{\cutwidth(H_{\alpha})}+\frac{1}{m^*(\alpha)}=R(H_{\alpha},L)+\frac{1}{m^*(\alpha)}$.
    The performance ratio $R(H_{\alpha},L)$ is at most $r(\opt,h)$, where $h=|\alpha|$ is the number of edges in $H_{\alpha}$. For the optimum value $k=m^*(\alpha)$, the cutwidth of $H_{\alpha,k}$ is at least $2k-2$ and $\sigma$ has performance ratio at most $r(2\opt,|\alpha|)$ (with respect to the optimum value $k$ for $\minlocProb$).
    The approximation procedure  builds the graph $H_{\alpha}$ in $\bigo(|\Sigma|)$, runs $\mathcal A$ on $H_{\alpha}$ in $\bigo(f(|\alpha|))$ and translates the linear arrangement into a marking sequence $\sigma$ in $\bigo(|\Sigma|)$. This gives an $(r(2\opt,|\alpha|)+\frac{1}{\opt})$-approximation for $\minlocProb$ running time in $\bigo(f(|\alpha|)+|\alpha|)$.
    \end{proof}

\subsection{Reducing Cutwidth to Locality Number}\label{sec:cutwidthToLocality}

We next describe a reduction from $\cutwidthProb$ to $\locProb$. To this end, let $H = (V,E)$ be a connected multigraph, where $V$ is the set of nodes and $E$ the multiset of edges (for technical reasons, we assume $|V| \geq 2$). First, we construct the multigraph $H' = (V, E')$ obtained by duplicating every edge in $H$. As such, each node in $H'$ has even degree, so we can fix some Eulerian cycle $C$ (i.e., a cycle visiting each edge exactly once) in $H'$, and, moreover, $\cutwidth(H')=2\cutwidth(H)$. For each edge $e \in E'$, let $\alpha_e$ be the word over $V$ that represents a traversal of the Eulerian path $P$ obtained from $C$ by deleting $e$. It is not important on which endpoint of the deleted edge $e$ we start the traversal of the Eulerian path $P$ (note that by introducing some order on $V$, we could easily make this choice unique).

\begin{figure}[tb]
\begin{center}

\includegraphics{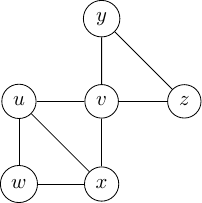}
\hspace{0.2cm}
\includegraphics{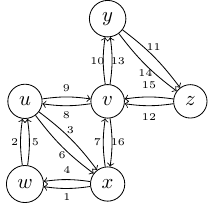}
\hspace{0.1cm}
\includegraphics{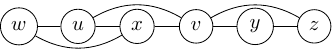}

\end{center}
\caption{A graph $H$ and its multigraph $H'$ obtained by doubling the edges; the edge labels describe a Eulerian cycle that starts and ends in $x$. Deleting the edge $(v, x)$ in this cycle yields the word $\alpha_{(v, x)} = xwuxwuxvuvyzvyzv$, which has an optimal marking sequence $(w, u, x, v, y, z)$ with marking number $3$, and, thus, induces an optimal linear arrangement of $H$ with cutwidth $3$.}
\label{fig:loc-cut}
\end{figure}

Let us again clarify this reduction with an example. Let $H$ be the graph shown on the left of Figure~\ref{fig:loc-cut} (note that this is the same graph from Figure~\ref{fig:exampleCutwidth}). By duplicating every edge in $H$, we obtain the graph $H'$ shown in the middle of Figure~\ref{fig:loc-cut} (the edge directions and labels are not an explicit part of the reduction and shall serve illustrative purposes). A possible Eulerian cylce of $H'$ that starts in vertex $x$ is  illustrated by the edge labels and the edge directions, i.\,e., the Eulerian cycle is $(x, w, u, x, w, u, x, v, u, v, y, z, v, y, z, v, x)$. Now splitting this cycle into a path by deleting its last edge $e = (v, x)$ results in a Eulerian path that corresponds to the word $\alpha_{(v, x)} = xwuxwuxvuvyzvyzv$.\par
The important property of the word $\alpha_{e}$ is that for every edge $\{x, y\}$ of $H$ (except $e$), it contains two distinct size-$2$ factors that are $xy$- or $yx$-factors (for example, the original edge $\{x, w\}$ translates into two $xw$-factors, while the original edge $\{u, v\}$ translates into a $vu$-factor and a $uv$-factor). Consider the cuts of a fixed linear arrangement of $H'$ from left to right and the marked versions of $\alpha_{e}$ with respect to the corresponding marking sequence. By construction, every boundary between a currently marked block and an adjacent unmarked block corresponds to a crossing edge of the current cut. This means that if there are $\ell$ marked blocks, then, depending on whether there is a marked prefix or suffix, the current cut must have size at least $2(\ell-1)$ and at most $2\ell$. On the other hand, every crossing edge of the current cut (except $e$, if contained in the cut) is responsible for a marked symbol next to an unmarked one. This means that if the size of the current cut is $2\ell$ (note that it must be even due to the duplication of edges), then there are $\ell$ marked blocks if no prefix or suffix is marked, there are $\ell + 1$ marked blocks if both a prefix and a suffix is marked, and if a prefix is marked but no suffix is marked (or the other way around), then in the current marked version there are $2\ell - 1$ boundaries between marked and unmarked blocks, and therefore the current cut contains $2\ell - 1$ edges different from $e$ (the ones responsible for the $2\ell - 1$ boundaries between marked and unmarked blocks), and the additional edge $e$, which is not represented by any size-$2$ factor in $\alpha_e$. Consequently, if $H'$ has a cutwidth of $2k$ (which means that $H$ has a cutwidth of $k$), then the locality number of $\alpha_e$ is either $k$ or $k + 1$. \par
If $(v_1, v_2, \ldots, v_n)$ is a linear ordering for $H$ with optimal cutwidth, then choosing any $e$ that is adjacent to $v_n$ will result in a word $\alpha_e$ for which, if marked according to $H$'s optimal linear arrangement, the situation that both a prefix and suffix is marked can only happen at the end when all symbols are marked. This means that for such a choice of $e$, we have $\loc(\alpha_e) = \cutwidth(H)$.\par
With respect to our example, we note that in fact the optimal linear arrangement shown on the right of Figure~\ref{fig:loc-cut} has a cutwidth of $3$, while the corresponding optimal marking sequence has a marking number of $3$ for the word $\alpha_{\{v, x\}}$ (note that here the edge $\{v, x\}$ is not adjacent to the first or last vertex of the linear arrangement). We shall next provide a formal proof for these intuitive observations.

\begin{lemma}\label{cut-loc}
For any edge $e$ in $E'$, the word $\alpha_e$ satisfies $\cutwidth(H) \leq \loc(\alpha_e)\leq \cutwidth(H)+1$. Moreover, there is a vertex $v \in V$ such that $\loc(\alpha_e) = \cutwidth(H)$ for every edge $e$ incident to $v$.
\end{lemma} 

\begin{proof}
Let $k = \cutwidth(H)$. Note that there is a natural bijection between the linear arrangements of $H'$ and the marking sequences of the word $\alpha_e$, since they both are essentially permutations of $\{1, 2, \ldots, n\}$, i.\,e., for a permutation $\tau$ of $\{1,2,\ldots,n\}$, we can interpret $(x_{\tau(1)},x_{\tau(2)},\ldots,x_{\tau(n)})$ both as the linear 
arrangement for $H'$ and the a marking sequence of $\alpha_e$ \emph{induced by $\tau$}.
In the following, let $\tau$ be a permutation of $\{1, 2, \ldots, n\}$, let $i \in \{1,2,\ldots, n-1\}$, $K_1 = \{x_{\tau(1)},x_{\tau(2)},\ldots,x_{\tau(i)}\}$ and $K_2 = \{x_{\tau(i+1)},\ldots,x_{\tau(n)}\}$, and let $\cutedges(K_1, K_2) = 2\ell$ (note that since every edge has been duplicated, the size of every cut of $H'$ is even).\par
Consider $\alpha_e$ after marking the letters $x_1,\ldots,x_{\tau(i)}$.
For every marked block $\alpha[s..t]$ that is not a prefix or a suffix of $\alpha$, we have $\alpha[s], \alpha[t] \in K_1$ and $\alpha[s - 1], \alpha[t + 1] \in K_2$ and therefore $\{\alpha[s-1], \alpha[s]\}, \{\alpha[t], \alpha[t+1]\} \in \cutedges(K_1, K_2)$.
Moreover, for a marked prefix $\alpha[1..s]$, we have $\alpha[s] \in K_1$ and $\alpha[s + 1] \in K_2$ and therefore $\{\alpha[s], \alpha[s+1]\} \in \cutedges(K_1, K_2)$. Analogously, the existence of a marked suffix $\alpha[t..|\alpha|]$ leads to $\{\alpha[t-1], \alpha[t]\} \in \cutedges(K_1, K_2)$.

Conversely, for every edge in $\cutedges(K_1, K_2)$,
with the exception of $e$ (if $e$ is in $\cutedges(K_1, K_2)$ at all),
there is a unique length-$2$ factor $\alpha_e[p..p+1]$ of $\alpha_e$ such that either $\alpha_e[p]$ is marked and $\alpha_e[p+1]$ is unmarked, or vice-versa. Thus, if all marked blocks are internal, i.e., no marked block is a prefix or a suffix, then there are exactly $\ell$ marked blocks. Also, if both a prefix and a suffix occurs as a marked block, then we have $\ell+1$ marked blocks. Finally, if a prefix occurs as a marked block, but no suffix, or vice-versa, then there are only $\ell$ marked blocks; note that in this case we must have $e \in \cutedges(K_1, K_2)$. Since we consider all permutations, the arguments above are sufficient to conclude that, in our setting, each $\alpha_e$ has locality number either $k$ or $k+1$.

Furthermore, consider a linear ordering $L=(x_{j_1}, \ldots, x_{j_n})$ of $H'$ 
which is optimal, i.e., $|\cutedges_L(i)|\leq 2k$. Note that if either the first 
or last letter of $\alpha_e$ is the last letter $x_{j_n}$ to be marked according 
to the marking sequence induced by the linear ordering $(x_{j_1}, \ldots, 
x_{j_n})$, the case that both a suffix and prefix of $\alpha_e$ are marked 
cannot be reached until $i = n$ and the entire word is marked. Consequently, 
this would imply that $\alpha_e$ has locality number $k$. For any permutation 
of the linear ordering $(x_{j_1}, \ldots, x_{j_n})$, this holds for $\alpha_e$ 
where $e$ is an edge adjacent to the node $x_{j_n}$, since the path $P$ obtained 
by removing such an edge $e$ from $C$ must start or end with $x_{j_n}$. 
\end{proof}

Again, we formally state and prove the approximation preserving properties of this reduction.

\begin{lemma}\label{cut-loc-apx}
If there is an $r(\opt,|\alpha|)$-approximation algorithm for $\minlocProb$ running in $\bigo(f(|\alpha|))$ time on a word $\alpha$, then there is an $r(\opt, h)$-approximation algorithm for $\minCutwidthProb$ running in $\bigo(n(f(h)+h))$ time on a multigraph with $n$ vertices and $h$ edges.
\end{lemma}

\begin{proof}
Let $G=(V,E)$ be an instance of $\minCutwidthProb$ and $\mathcal A$ an $r(\opt,|\alpha|)$-approximation algorithm for $\minlocProb$. By Lemma~\ref{cut-loc}, there exists a vertex $v\in V$ such that  $\loc(\alpha_e)=\cutwidth(G)$ for any  edge $e\in E$ adjacent to $v$.  The approximation algorithm $\mathcal A$ hence returns on input $\alpha_e$ a marking sequence $\sigma$ with $\pi_\sigma(\alpha_e)\leq r(\opt,|\alpha|)\cutwidth(G)$.\par
In the proof of Lemma~\ref{cut-loc} it is further  shown that any marking sequence  $\sigma$ for $\alpha_e$ translates to a linear arrangement $L$ for $G$ with $\cutwidth(L)\leq \pi_\sigma(\alpha_e)$. The performance ratio of this linear arrangement is $R(G, L)=\frac{\cutwidth(L)}{\cutwidth(G)}\leq \frac{\pi_\sigma(\alpha_e)}{\loc(\alpha_e)}\leq R(\alpha_e,\sigma)$.\par
 The procedure which,  for each vertex $v\in V$, constructs $\alpha_e$ for some $e\in E$ adjacent to $v$  in $\bigo(h)$, runs $\mathcal A$ in $\bigo(f(|\alpha_e|))=\bigo(f(h))$ and checks the resulting linear arrangement in $\bigo(h)$ and returns the best linear arrangement among all $v\in V$,  yields an $r(\opt,h)$-approximation for $\minCutwidthProb$ on multigraphs in $\bigo(n(f(h)+h))$.
\end{proof}

\subsection{Consequences of the Reductions} 

In the following, we discuss the lower and upper complexity bounds that we obtain from the reductions provided above. We first note that since $\cutwidthProb$ is $\NP$-complete, so is $\locProb$. In particular, note that this answers one of the main questions left open in~\cite{FSTTCS}.

\begin{theorem}\label{locNPC}
$\locProb$ is $\NP$-complete (under Turing reductions), even if every symbol has at most $3$ occurrences.  
\end{theorem}

\begin{proof}
Lemma~\ref{cut-loc} shows a polynomial time Turing reduction from $\cutwidthProb$ to $\locProb$. Indeed, given a (multi)graph $H$ we construct in linear time the multigraph $H'$ by duplicating its edges. $H'$ has an Eulerian cycle, so, using Hierholzer's algorithm, we can compute such a cycle in linear time \cite{Hierholzer1873}. Let $C$ be the computed Eulerian cycle. For each edge $e$ of $C$ construct, in linear time, the word $\alpha_e$ as described before Lemma \ref{cut-loc}. By Lemma \ref{cut-loc} we get that $\cutwidth(H)=\frac{\cutwidth(H')}{2}=\min\{\loc(\alpha_e)\mid e\mbox{ edge of } C\}$. This completes the reduction, and, thus, as $\cutwidthProb$ is $\NP$-hard (see, e.\,g.,~\cite{surveyDiaz}), we get that $\locProb$ is also $\NP$-hard.\par
In order to show that the hardness holds even if every symbol has at most $3$ occurrences, we first observe that in the reduction from Section~\ref{sec:cutwidthToLocality}, the number of occurrences of any symbol $x$ in the constructed word $\alpha_{e}$ corresponds to the degree of the vertex $x$ in the graph $H$. Hence, since $\cutwidthProb$ is $\NP$-complete already for graphs with maximum degree $3$ (see~\cite{MakedonEtAl1985}), it follows that $\locProb$ is $\NP$-complete even if every symbol has at most $3$ occurrences.\par
In order to show that $\locProb$ is in $\NP$, let $\alpha \in \Sigma^*$ be an arbitrary word and $k \in \mathbb{N}$. We can guess a marking sequence $\sigma$ for $\alpha$ in polynomial time, and then check in polynomial time whether its marking number $\pi_\sigma(\alpha)$ is less than or equal to~$k$.
\end{proof}

Next, we formally state the positive fixed-parameter tractability results that $\locProb$ inherits from $\cutwidthProb$ via the reduction from Section~\ref{sec:locToCutwidth} (note that the fixed-parameter tractability of $\locProb$ was left as open problem in~\cite{FSTTCS}).

\begin{theorem}\label{locFPTAlphabet}
$\locProb$ is fixed-parameter tractable if parameterised by $|\Sigma|$. Moreover, it can be solved in time and space $\bigo^*(2^{|\Sigma|})$, or in $\bigo^*(4^{|\Sigma|})$ time and polynomial space. 
\end{theorem}

\begin{proof}
In~\cite{BodlaenderFKKT12}, the authors present algorithms for $\cutwidthProb$ that run in $\bigo^*(2^{n})$ time and space, or in $\bigo^*(4^{n})$ time and polynomial space (where $n$ is the number of vertices), and they also work for multigraphs.\footnote{These algorithms actually support weighted graphs without any major modification and in the same complexity. In this setting, parallel edges connecting two vertices are replaced by a single ``super-edge'' whose weight is the number of parallel edges.} Hence, Lemma~\ref{locAlgoMin} implies that $\locProb$ can be solved in $\bigo^*(2^{|\Sigma|})$ time and space, or in $\bigo^*(4^{|\Sigma|})$ time and polynomial space. 
\end{proof}
\begin{theorem}\label{locFPTLocNumber}
$\locProb$ is fixed-parameter tractable if parameterised by $\loc(\alpha)$. Moreover, it can be solved with linear fpt-running-time $g(\loc(\alpha))\bigo(|\Sigma|)$.
\end{theorem}

\begin{proof}
The algorithm from~\cite{ThilikosSB05} solves $\cutwidthProb$ with linear fpt-running-time $g(\cutwidth(G)) \bigo(n)$ (where $n$ is the number of vertices). Hence, Lemma~\ref{locAlgoMin} implies that $\locProb$ can be solved with linear fpt-running-time $g(\loc(\alpha))\bigo(|\Sigma|)$. 
\end{proof}

The most natural parameters for $\locProb$ are the alphabet size $|\Sigma|$ and the standard parameter $\loc(\alpha)$ (recall that we have just seen that $\locProb$ is fixed-parameter tractable with respect to these two parameters). However, for string problems it is also common to investigate the parameterised complexity with respect to the maximum number of occurrences per symbols in the word $\alpha$ (we denote this parameter by $\maxocc{\alpha}$). Theorem~\ref{locNPC} already demonstrated that $\locProb$ is not fixed-parameter tractable with respect to $\maxocc{\alpha}$ (unless $\pclass = \NP$). However, we only know that $\locProb$ stays $\NP$-hard if $\maxocc{\alpha}$ is bounded by a constant $k$ with $k \geq 3$, and that the problem is trivial if $\maxocc{\alpha}$ is bounded by $1$ (in this case, the locality number is always $1$). The complexity of $\locProb$ is open for the case where $\maxocc{\alpha} \leq 2$.

\begin{openproblem}\label{maxxoccOpenProblem}
Can $\locProb$ be solved in polynomial time if $\maxocc{\alpha} \leq 2$?
\end{openproblem}

If, on the other hand, the maximum number of occurrences per symbol in $\alpha$ is large in terms of $\alpha$'s length, i.\,e., we have that $\maxocc{\alpha} = \Omega(\frac{|\alpha|}{\log(|\alpha|)})$, then $\locProb$ can be solved in polynomial time. Indeed, since $\maxocc{\alpha} \geq \frac{|\alpha|}{|\Sigma|}$, we can  conclude that $\frac{|\alpha|}{|\Sigma|} = \Omega(\frac{|\alpha|}{\log(|\alpha|)})$, which also means that $\log(|\alpha|) = \Omega(|\Sigma|)$. Consequently, $|\Sigma| = \bigo(\log(|\alpha|))$, which means that $\locProb$ can be solved in polynomial time by using the $\bigo^*(2^{|\Sigma|})$-time algorithm mentioned in Theorem~\ref{locFPTAlphabet}.

We conclude this section by pointing out that Lemmas~\ref{loc-cut-apx}~and~\ref{cut-loc-apx} imply some approximation results for $\minlocProb$. However, we shall discuss approximation issues in greater detail in Section~\ref{sec:approx}.

\section{Locality and Pathwidth}\label{sec:approx}

One of the main results of this section is a reduction from the problem of computing the locality number of a word $\alpha$ to the probem of computing the pathwidth of a graph. This reduction, however, does not technically provide a reduction from the decision problem $\locProb$ to $\pathwidthProb$, since the constructed graph's pathwidth ranges between $\loc(\alpha)$ and $2\loc(\alpha)$, and therefore the reduction cannot be used to solve $\minlocProb$ exactly. The main purpose of this reduction is to carry over approximation results from $\minPathwidthProb$ to $\minlocProb$ (also recall that exact and fpt-algorithms for $\minlocProb$ are obtained in Section~\ref{sec:cutwidth} via a reduction to $\minCutwidthProb$). Hence, in this section we are mainly concerned with approximation algorithms. \par
Our strongest positive result about the approximation of the locality number will be derived from the reduction mentioned above (see Section~\ref{sec:locToPW}). However, we shall first investigate in Section~\ref{sec:apprGreedy} the approximation performance of several obvious greedy strategies to compute the locality number (with ``greedy strategies'', we mean simple algorithmic strategies that build up a marking sequence from left to right by choosing the next symbol to be marked by some simple greedy rule). This is mainly motivated by two aspects. Firstly, ruling out simple strategies is a natural initial step in the search for approximation algorithms for a new problem. Secondly, due to the results of Section~\ref{sec:cutwidth}, the investigated greedy strategies for computing the locality number can also be interpreted as greedy strategies for computing the cutwidth of a graph. This may provide a new angle to approximating the cutwidth of a graph, i.e., some greedy strategies may only become apparent in the locality number point of view and are hard to see in the graph formulation of the problem. It may seem naive to expect new approximation results for cutwidth in this way, but, as mentioned in the introduction and as shall be discussed in detail in Section~\ref{sec:direct}, approximating the cutwidth via approximation of the locality number may be beneficial for cutwidth approximation (although not by using simple greedy strategies, but the algorithm that follows from the reduction to computing the pathwidth). \par
Before presenting the main results of this section, let us briefly discuss some inapproximability results for $\minlocProb$ that directly follow from the reductions of Section~\ref{sec:cutwidth} and known results about cutwidth approximation. Firstly, it is known that, assuming the Small Set Expansion Conjecture (denoted $\textsc{SSE}$; see~\cite{SSE}), there exists no constant-ratio approximation for $\minCutwidthProb$ (see~\cite{WuEtAl2014}). Consequently, approximating  $\minlocProb$ within any constant factor is also SSE-hard. In particular, we point out that stronger inapproximability results for $\minCutwidthProb$ are not known. \par

On certain graph classes, the $\textsc{SSE}$ conjecture is equivalent to the Unique Games Conjecture \cite{UCG} (see \cite{SSE,SSEfollow}), which, at its turn, was used to show that many approximation algorithms are tight (see~\cite{UCG-survey}) and is considered a major conjecture in inapproximability. However, some works seem to provide evidence that could lead to a refutation of $\textsc{SSE}$; see \cite{AroraBS10,BarakRS11,Sinop}. In this context, our negative result of Section~\ref{sec:apprGreedy} can also be interpreted as a series of unconditional results which state that multiple natural greedy strategies for computing the locality number (and their equivalents for computing the cutwidth) do not provide low-ratio approximations of $\minlocProb$ (or $\minCutwidthProb$, respectively).

\subsection{Greedy Strategies}\label{sec:apprGreedy}

Since a marking sequence is just a linear arrangement of the symbols of the input word, computing marking sequences seems to be well tailored to greedy algorithms: until all symbols are marked, we choose an unmarked symbol according to some greedy strategy and mark it. Unfortunately, we can formally show that many natural candidates for greedy strategies fail to yield promising approximation algorithms (and are therefore also not helpful for cutwidth approximation). \par
For a systematic investigation, we shall now define our \emph{basic greedy strategies}:\smallskip\\
\begin{tabular}{ll}
$\SOstrategy$ & Among all unmarked symbols, choose one with a smallest number of occurrences.\\
$\MOstrategy$ & Among all unmarked symbols, choose one with a largest number of occurrences.\\
$\SNMstrategy$ & Among all unmarked symbols, choose one that, after marking it, results in the\\
& smallest total number of marked blocks.\\
$\LRstrategy$ & Among all unmarked symbols, choose the one with the leftmost occurrence.
\end{tabular}\smallskip \\

These strategies are -- except for $\LRstrategy$ -- nondeterministic, since there are in general several valid choices of the next symbol to mark. However, we will show poor performances for these strategies independent of the nondeterministic choices (i.\,e., the approximation ratio is bad for every possible nondeterministic choices), which are stronger negative results. We make the convention that all strategies -- except, of course, $\LRstrategy$ -- can choose any symbol as the initially marked symbol, which is justified by the fact that, in terms of running-time, we could afford to try out every possible choice of the first symbol.\par
In the following, for every greedy strategy $S$ and for every word $\alpha$, let $\greedy_S(\alpha)$ be the optimal marking number over all marking sequences that can be obtained by strategy $S$. For every word $\alpha$ let $\psi_S(\alpha) = \frac{\greedy_{S}(\alpha)}{\loc(\alpha)}$, and for every $\ell \in \mathbb{N}$, let $\psi_S(\ell) =\max\left\{\psi_S(\alpha) \mid |\alpha| = \ell\right\}$; the function $\psi_S : \mathbb{N} \to \mathbb{N}$ is called the \emph{approximation performance of strategy $S$}. Our negative results will be as follows. For every strategy $S$ and for every $\ell \in \mathbb{N}$, we show that there is a constant $c$ and a length-$\ell$ word $\alpha$, such that $\loc(\alpha) = c$, while $\greedy_S(\alpha) = \Omega(\ell)$. Note that this means that the approximation performance $\psi_S(\ell)$ of strategy $S$ is linear. \par
We first investigate strategies $\SOstrategy$ and $\SNMstrategy$. For every $\ell \geq 2$, let 
\begin{equation*}
\alpha = (x_1 x_2 \ldots x_{\ell})^2 x_1 \beta_1 x_2 \beta_2 x_3 \beta_3 \ldots \beta_{\ell-1} x_{\ell}\,,
\end{equation*}
where, for every $i$, $1 \leq i \leq \ell-1$, $\beta_i = (y_{2i-1} y_{2i})^4$. \par
For example, if $\ell = 4$, then we obtain $\alpha = (x_1 x_2 x_3 x_4)^2 x_1 (y_{1} y_{2})^4 x_2 (y_{3} y_{4})^4 x_3 (y_{5} y_{6})^4 x_4$. It can be easily seen that $|\alpha| = 11\ell - 8 = \bigo(\ell)$.

\begin{lemma}\label{SOSNMstrategyLemma}
$\psi_{\SOstrategy}(\alpha) \geq \frac{\ell-1}{6}$ and $\psi_{\SNMstrategy}(\alpha) \geq \frac{\ell-1}{6}$.
\end{lemma}

\begin{proof}
We first observe that $(x_1, y_1, y_2, x_2, y_3, y_4, x_3, y_5, y_6, \ldots)$ is an optimal marking sequence which shows that $\loc(\alpha) = 6$. Next, we consider how the strategies $\SOstrategy$ and $\SNMstrategy$ can mark $\alpha$. If the first marked symbol is some $x_i$, then $\SOstrategy$ would next mark all remaining $x_{j}$, $j \neq i$, in some order, since each symbols $x_i$ has fewer occurrences than any of the symbols $y_j$, and $\SNMstrategy$ would next mark all remaining $x_{j}$, $j \neq i$, since these can be marked in such an order that per new marking we increase the number of new blocks only by one, while marking some $y_j$ would increase the number of marked blocks by three. This leads to at least $\ell$ marked blocks. If, on the other hand, some $y_{2j-1}$ or $y_{2j}$ is marked first, then $\SOstrategy$ marks some $x_i$ next and then all remaining $x_{i'}$ as before, while $\SNMstrategy$ would mark the remaining symbol of $y_{2i-1}$ or $y_{2i}$ (because this is the only choice that does not increase the number of marked blocks) and then all $x_i$ in some order that produces a minimal number of new marked blocks. This results in at least $\ell-1$ marked blocks. Thus, $\psi_{\SOstrategy}(\alpha) \geq \frac{\ell-1}{6}$ and $\psi_{\SNMstrategy}(\alpha) \geq \frac{\ell-1}{6}$. 
\end{proof}

Next, we consider the strategy $\MOstrategy$. For every $\ell \geq 2$, let 

\begin{equation*}
\gamma = x_1 x_2 \ldots x_{\ell} x_1 y_1 x_2 y_2 x_3 y_3 \ldots y_{\ell-1} x_{\ell}\,.
\end{equation*}

For example, if $\ell = 5$, then we obtain $\gamma = x_1 x_2 x_3 x_4 x_5 x_1 y_1 x_2 y_2 x_3 y_3 x_4 y_4 x_5$. It can be easily seen that $|\gamma| = 3\ell - 1 = \bigo(\ell)$.

\begin{lemma}\label{MOstrategyLemma}
$\psi_{\MOstrategy}(\gamma) \geq \frac{\ell-1}{2}$.
\end{lemma}

\begin{proof}
We first observe that $(x_1, y_1, x_2, y_2, x_3, y_3, \ldots)$ is an optimal marking sequence which shows that $\loc(\gamma) = 2$. If $\MOstrategy$ marks some $x_i$ first, then it will mark all remaining $x_j$ next, since each of these symbols have $2$ occurrences. This results in $\ell$ marked blocks. If, on the other hand, the first symbol is some $y_j$, then again the symbols $x_i$ have the most occurrences and are therefore marked next in some order. This leads to $\ell-1$ marked blocks. Thus, $\psi_{\MOstrategy}(\gamma) \geq \frac{\ell-1}{2}$.
\end{proof}

Finally, we consider the strategy $\LRstrategy$. For every even number $\ell \geq 2$, let
\begin{equation*}
\delta = x_1 x_2 \ldots x_{\ell} x_1 x_\ell x_2 x_{\ell-1} x_3 x_{\ell-2} \ldots x_{\frac{\ell}{2}} x_{\frac{\ell}{2} + 1}\,. 
\end{equation*}
For example, if $\ell = 6$, then we obtain $\delta = x_1 x_2 x_3 x_4 x_5 x_6 x_1 x_6 x_2 x_5 x_3 x_4$. It can be easily seen that $|\delta| = 2\ell = \bigo(\ell)$.

\begin{lemma}\label{LRstrategyLemma}
$\psi_{\LRstrategy}(\delta) \geq \frac{\ell}{4}$.
\end{lemma}

\begin{proof}
We first observe that $\loc(\delta) = 2$, which is witnessed by the marking sequence 
\begin{equation*}
(x_1, x_{\ell}, x_2, x_{\ell-1}, x_3, x_{\ell-2}, \ldots, x_{\frac{\ell}{2}}, x_{\frac{\ell}{2} + 1}) 
\end{equation*}
(note that this marking sequence maintains a marked prefix and one additional marked internal factor starting with $x_{\ell} x_1 x_{\ell}$, which is alternately extended to both sides). 
Now assume that the strategy $\LRstrategy$ marks some symbol $x_i$. If $i \leq \frac{\ell}{2}$, then it marks next all the symbol $x_1, \ldots, x_{i-1}, x_{i+1}, \ldots, x_{\frac{\ell}{2}}$, which results in $\frac{\ell}{2} + 1$ marked blocks. If, on the other hand, $i > \frac{\ell}{2}$, then symbols $x_1, \ldots, x_{\frac{\ell}{2}}$ are marked next, which leads to at least $\frac{\ell}{2}$ marked blocks. Thus $\psi_{\LRstrategy}(\delta) \geq \frac{\ell}{4}$.
\end{proof}

In the following, we investigate another aspect of greedy strategies. Any symbol that is marked next in a marking sequence can have \emph{isolated} occurrences (i.\,e., occurrences that are not adjacent to any marked block) and \emph{block-extending} occurrences (i.\,e., occurrences with at least one adjacent marked symbol). Each isolated occurrence results in a new marked block, while each block-extending occurrence just extends an already existing marked block, and potentially may even combine two marked blocks and therefore may decrease the overall number of marked blocks. Therefore, marking a symbol when it only has isolated occurrences causes the maximum number of marked blocks that can ever be contributed by this symbol, and therefore this seems to be the worst time to mark this symbol. Hence, in terms of a greedy strategy, it seems reasonable to only mark symbols if they also have block-extending occurrence (obviously, this is not possible for the initially marked symbol). \par
We call a marking sequence $\sigma$ for a word $\alpha$ \emph{block-extending}, if every symbol that is marked except the first one has at least one block-extending occurrence. This definition leads to the general combinatorial question of whether every word has an optimal marking sequence that is block-extending, or whether the seemingly bad choices of marking a symbol that has only isolated occurrences (and that is not the first symbol) is necessary for optimal marking sequences. We answer this question in the negative.

For every even number $\ell \geq 2$, let $\beta = x_1 y x_2 y x_3 y \ldots x_{\ell} y$.

\begin{proposition}\label{GreedyPropositionBE}
$\loc(\beta) = \frac{2}{\ell}$ and, for every block-extending marking sequence $\sigma$ for $\beta$, we have $\pi_{\sigma}(\beta) \geq \ell-1$.
\end{proposition}

\begin{proof}
For the sake of convenience, let $\ell = 2k$ for some $k \geq 1$. Let $\sigma$ be any block-extending marking sequence for $\alpha$. If $\sigma$ marks $y$ first, then we have $2k$ marked blocks and if some $x_i$, $1 \leq i \leq 2k$, is marked first, then $y$ is marked next, which leads to $2k-1$ marked blocks. Thus, $\pi_{\sigma}(\beta) \geq 2k-1$. On the other hand, we can proceed as follows. We first mark the $k$ symbols $x_2, x_3, \ldots, x_{k + 1}$, which leads to $k$ marked blocks (and which is a marking sequence that is \emph{not} block extending). Then we mark $y$, which joins all the previously marked blocks into one marked block and turns $k - 1$ occurrences of $y$ into new individual marked blocks (i.\,e., the $k - 2$ occurrences of $y$ between the symbols $x_{k + 2}, x_{k + 3}, \ldots, x_{2k}$ and the single occurrence of $y$ after $x_{2k}$). Thus, there are still $k$ marked blocks, and from now on marking the rest of the symbols only decreases the number of marked blocks. Consequently, $\loc(\beta) \leq k$. Moreover, after any marking sequence has marked symbol $y$, there are $2k$ marked occurrences of symbol $y$. If these marked occurrences form at least $k$ marked blocks, the overall marking number of the marking sequence is at least $k$. If they form at most $k-1$ marked blocks, then at least $k + 1$ of the symbols $x_i$ must be marked as well, and since these symbols were marked before marking $y$, they have formed at least $k+1$ marked blocks before marking $y$. This means that the overall marking number is at least $k + 1$. This shows that $\loc(\beta) \geq k$, and therefore $\loc(\beta) = k$.
\end{proof}

This proposition points out that even simple words can have only optimal marking sequences that are not block-extending. In terms of greedy strategies however, Proposition~\ref{GreedyPropositionBE} only shows a lower bound of roughly $2$ for the approximation ratio of any greedy algorithm that employs some block-extending greedy strategy (since the lower bound applies to \emph{every} marking sequence that is block-extending). We note that the requirement of marking in a block-extending way does not specify which one of the possible block-extending symbols should be marked; trying out all of them is obviously too costly. In order to further investigate block-extending greedy strategies, we therefore couple the block-extending requirement with other greedy strategies, e.\,g., for the strategies $S \in \{\SOstrategy, \MOstrategy, \SNMstrategy, \LRstrategy\}$ from above, we denote by $\BEstrategy-S$ the greedy strategy which only marks block-extending symbols (except the first one) and chooses among possible block-extending symbols according to strategy $S$ (note that these strategies are still non-deterministic in the sense of how $S$ is a non-deterministic strategy). More precisely, the strategies $\BEstrategy-S$ are defined by replacing ``unmarked symbols'' by ``unmarked symbols that have at least one block-extending occurrence'' in the descriptions of the basic strategies from above. \par
We observe that for $S \in \{\SOstrategy, \MOstrategy, \SNMstrategy, \LRstrategy\}$ the strategy $\BEstrategy-S$ is not covered by $S$, i.\,e., the set of marking sequences for a word that can be obtained by $S$ is not necessarily a superset of the marking sequences that satisfy $\BEstrategy-S$. For example, an unmarked symbol that has a maximum (or minimum) number of occurrences among all block-extending symbols, might not have a maximum (or minimum, respectively) number of occurrences among all unmarked symbols. Therefore, the lower bounds form Lemmas~\ref{SOSNMstrategyLemma},~\ref{MOstrategyLemma}~and~\ref{LRstrategyLemma} do not carry over automatically. Nevertheless, $\BEstrategy-S$ behaves more or less in the same way as $S$ on the witness words $\alpha$, $\gamma$ and $\delta$ defined above, which yields the following.

\begin{lemma}\label{BEstrategyLemma}
Let the words $\alpha$, $\gamma$ and $\delta$ be defined as above. Then
\begin{itemize}
\item $\psi_{\BEstrategy-\SOstrategy}(\alpha) \geq \frac{\ell-1}{6}$,
\item $\psi_{\BEstrategy-\SNMstrategy}(\alpha) \geq \frac{\ell-1}{6}$, 
\item $\psi_{\BEstrategy-\MOstrategy}(\gamma) \geq \frac{\ell-1}{6}$,
\item $\psi_{\BEstrategy-\LRstrategy}(\delta) \geq \frac{\ell}{4}$.
\end{itemize}
\end{lemma}

\begin{proof}
Optimal marking sequences for the words and the lengths of these words have been discussed above, so it only remains to show that no $\BEstrategy-S$ with $S \in \{\SOstrategy, \MOstrategy, \SNMstrategy, \LRstrategy\}$ can produce better marking sequences.\par
We first discuss $\BEstrategy-\SOstrategy$ and $\BEstrategy-\SNMstrategy$ together. If the initially marked symbol is some $x_i$, then both $\BEstrategy-\SOstrategy$ and $\BEstrategy-\SNMstrategy$ would next mark all remaining $x_{j}$, $j \neq i$. The only difference to strategies $\SOstrategy$ and $\SNMstrategy$ is that $\BEstrategy-\SOstrategy$ and $\BEstrategy-\SNMstrategy$ must mark these symbols such that always a block-extending symbol is marked next, i.\,e., we must mark in such a way there are always at most $2$ marked blocks in the prefix $(x_1 x_2 \ldots x_{\ell})^2$. This leads to at least $\ell$ marked blocks. If some $y_{2i-1}$ or $y_{2i}$ is marked first, then $\BEstrategy-\SOstrategy$ marks $x_i$ or $x_{i+1}$ next (depending on whether $y_{2i-1}$ or $y_{2i}$ is marked first) and then all remaining $x_{j}$ as before, while $\BEstrategy-\SNMstrategy$ would mark the remaining symbol of $y_{2i-1}$ or $y_{2i}$ and then all $x_j$. This results in at least $\ell-1$ marked blocks. Thus, $\psi_{\BEstrategy-\SOstrategy}(\alpha) \geq \frac{\ell-1}{6}$ and $\psi_{\BEstrategy-\SNMstrategy}(\alpha) \geq \frac{\ell-1}{6}$.\par
Next, we consider $\BEstrategy-\MOstrategy$. It can be easily seen that no matter whether we initially mark some $x_i$ or some $y_i$, just like $\MOstrategy$ the strategy $\BEstrategy-\MOstrategy$ will mark all remaining $x_i$ next (these symbols have a maximum number of occurrences and two of them must be block-extending). Thus, $\BEstrategy-\MOstrategy$ necessarily produces $\ell$ or $\ell-1$ marked blocks, and therefore $\psi_{\BEstrategy-\MOstrategy}(\gamma) \geq \frac{\ell-1}{6}$. \par
Finally, we consider $\BEstrategy-\LRstrategy$. This strategy behaves similar to $\LRstrategy$, but we have to argue a bit more carefully. Assume that $x_i$ is the first symbol marked by $\BEstrategy-\LRstrategy$. If $i \leq \frac{\ell}{2}$, then we mark next $x_{i-1}$, then $x_{i-2}$ and so on, until all $x_1, x_2, \ldots, x_i$ are marked. Then the symbols $x_{i + 1}, x_{i + 2}, \ldots, x_{\frac{\ell}{2}}$ are marked, which leads to $\frac{\ell}{2} + 1$ marked blocks. If, on the other hand, $i = \frac{\ell}{2} + j$ with $j \geq 1$, then the next marked symbol will be $x_{\frac{\ell}{2}-(j-1)}$. Then, as before, we will mark $x_{\frac{\ell}{2}-(j-1)-1}, x_{\frac{\ell}{2}-(j-1)-2}, \ldots$ until all $x_1, x_2, \ldots, x_{\frac{\ell}{2}-(j-1)}$ are marked, and then $x_{\frac{\ell}{2}-(j-1) + 1}, \ldots, x_{\frac{\ell}{2}}$ are marked. This results in $\frac{\ell}{2}$ marked blocks. Thus, $\psi_{\BEstrategy-\LRstrategy}(\delta) \geq \frac{\ell}{4}$. 
\end{proof}

If we are concerned with block-extending greedy strategies, then it is natural to choose among the block-extending symbols according to the number of their block-extending or isolated occurrences. This motivates the following two strategies: \smallskip\\

\begin{tabular}{ll}
$\BEstrategyAltOne$ & Among all block-extending symbols, choose one that has the most\\
& block-extending occurrences.\\
$\BEstrategyAltTwo$ & Among all block-extending symbols, choose one \\
& for which $\frac{\text{\#block-extending occ.}}{\text{\#occ.}}$ is maximal.
\end{tabular}\smallskip\\

Unfortunately, the witness word $\alpha$ also shows poor approximation ratios for these strategies.

\begin{lemma}\label{BEAltOnestrategyProposition}
$\psi_{\BEstrategyAltOne}(\alpha) \geq \frac{\ell-1}{6}$ and $\psi_{\BEstrategyAltTwo}(\alpha) \geq \frac{\ell-1}{6}$.
\end{lemma}

\begin{proof}
Again, we first note that the optimal marking sequence for $\alpha$ and the length $|\alpha|$ have already been discussed above. If we first mark a symbol $x_i$, then, among all symbols extending a marked block, i.\,e., symbols $x_{i-1}$, $x_{i + 1}$, $y_{2i-1}$ and $y_{2{i+1}}$, the symbols $x_{i-1}$ and $x_{i+1}$ each have $3$ occurrences in total, two of which are block-extending, whereas the symbols $y_{2i-1}$ and $y_{2{i+1}}$ each have $4$ occurrences, only one of which is block-extending. Consequently, both $\BEstrategyAltOne$ and $\BEstrategyAltTwo$ chose either $x_{i-1}$ or $x_{i+1}$ next. This situation does not change until all $x_i$ are marked, which leads to $\ell$ marked blocks. If, on the other hand, some $y_{2i-1}$ or $y_{2i}$ is marked first, then we mark next the remaining symbol $y_{2i-1}$ or $y_{2i}$ such that $\beta_i$ is completely marked (this is due to the fact that this symbol is the only block-extending one). Next, $x_{i}$ and $x_{i+1}$ are marked, in some order (again, this is enforced by the fact that we can only mark block-extending symbols), which brings us back to the situation described above which leads to the marking of all remaining $x_j$, leading to $\ell-1$ marked blocks. Consequently, $\psi_{\BEstrategyAltOne}(\alpha) \geq \frac{\ell-1}{6}$ and $\psi_{\BEstrategyAltTwo}(\alpha) \geq \frac{\ell-1}{6}$.
\end{proof}

\subsection{Reducing Locality Number to Pathwidth}\label{sec:locToPW}

In the following, we obtain an approximation algorithm for the locality number by reducing it to the problem of computing the pathwidth of a graph. To this end, we first describe another way of how a word can be represented by a graph. Recall that the reduction to cutwidth from Section~\ref{sec:cutwidth} also transforms words into graphs. The main difference is that the reduction from Section~\ref{sec:cutwidth} turns every symbol from the alphabet into an individual vertex of the graph (thus, producing a graph with $\bigo(|\Sigma|)$ vertices), while the reduction to pathwidth will use a vertex per position of the word $\alpha$, i.\,e., $|\alpha|$ individual vertices. In the reduction from Section~\ref{sec:cutwidth} the information of the actual occurrences of the symbols in the word is encoded  by the edges (in particular, the length $|\alpha|$ is represented by the number of edges), while in the following reduction the alphabet is encoded by connecting the vertices that correspond to positions of the same symbol to cliques in the graph (in particular, the number of edges may range between $|\alpha|$ and $|\alpha|^2$). We proceed with a formal definition and an example.\par
For a word $\alpha$, the graph $\alphagraph = (V_{\alpha}, E_{\alpha})$ is defined by $V_{\alpha} = \{1, 2, \ldots, |\alpha|\}$ and $E_{\alpha} = \{\{i, i+1\} \mid 1 \leq i \leq |\alpha|-1\} \cup \{\{i, j\} \mid \{i, j\} \subseteq \pset_{x}(\alpha)$ for some $x \in \alphabet(\alpha)\}$. Intuitively, $\alphagraph$ is a length-$|\alpha|$ path (due to the edges $\{\{i, i+1\} \mid 1 \leq i \leq |\alpha|-1\}$), and, additionally, we add edges such that every set $\pset_x(\alpha)$, $x \in \alphabet(\alpha)$, is a clique.\par
Let us discuss a brief example. Let $\alpha = \tc \ta \tb \ta \tc \ta \tb \ta \tc$ be a word of length $9$ over alphabet $\{\ta, \tb, \tc\}$. The graph $G_{\alpha}$ has therefore vertices $\{1, 2, \ldots, 9\}$, which are connected into a path from vertex $1$ to vertex $9$. In addition, $\pset_{\ta}(\alpha) = \{2, 4, 6, 8\}$ forms a clique, $\pset_{\tb}(\alpha) = \{3, 7\}$ forms a clique, and $\pset_{\tb}(\alpha) = \{1, 5, 9\}$ forms a clique. See Figure~\ref{fig:localityPathwidthReduction} for an illustration. 

\begin{figure}
\centering

\includegraphics{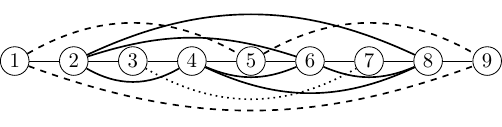}

\caption{The graph $G_{\alpha}$ for $\alpha = \tc \ta \tb \ta \tc \ta \tb \ta \tc$; the three cliques are drawn with different edge-types.}
\label{fig:localityPathwidthReduction}
\end{figure}

We use $\alphagraph$ as a unique graph representation for words and whenever we talk about a path decomposition for $\alpha$, we actually refer to a path decomposition of $\alphagraph$. Recall that we consider path-decompositions as certain marking schemes, which we called pd-marking schemes (see Section~\ref{sec:graphTerm} and Figure~\ref{fig:examplePathwidth}). Since $\alphagraph$ has the positions of $\alpha$ as its vertices, the pd-marking scheme behind a path decomposition (and its respective terminology) directly translates to a marking scheme of the positions of $\alpha$.\par
A very similar reduction has been used in~\cite{rei:patIaC} in order to prove that certain structural restrictions of patterns with variables lead to polynomial-time cases of the corresponding matching problem. The reduction from~\cite{rei:patIaC} is more general in the sense that it does not require the vertices of $\pset_{x}(\alpha)$ to be a clique, but only requires that these vertices form a connected component (if we do not count the ``path-edges'' $\{\{i, i+1\} \mid 1 \leq i \leq |\alpha|-1\}$).\par
The main property of this reduction is that the pathwidth of $\alphagraph$ ranges between $\loc(\alpha)$ and $2\loc(\alpha)$. 

\begin{lemma}\label{pathwidthLemma}
Let $\alpha$ be a word with $|\alpha| \geq 2$. Then $\loc(\alpha) \leq \pathwidth(\alphagraph) \leq 2\loc(\alpha)$. Further, given a path decomposition $Q$ for $\alphagraph$,  a marking sequence $\sigma$ for $\alpha$ with $\pi_\sigma(\alpha)\leq \width(Q)$ can be constructed in time $\bigo(|\alpha|)$.
\end{lemma}

We defer the somewhat technical proof of Lemma~\ref{pathwidthLemma} to the end of this section, and first discuss the consequences and some further properties of the reduction.\par

A rather simple observation is that the statement of Lemma~\ref{pathwidthLemma} is in fact not true for words $\alpha$ of size $1$, since then $\loc(\alpha) = 1$ and $\pathwidth(\alphagraph) = 0$. \par
Intuitively speaking, every marked block in an optimal marking sequence for $\alpha$ accounts for one unit of the quantity $\loc(\alpha)$, while in an optimal path decomposition of $\alphagraph$, any marked block is represented by two $\act$ vertices (i.\,e., vertices that are in the current bag, see the terminology introduced in Section~\ref{sec:graphTerm}). This explains why $\pathwidth(\alphagraph)$ can be twice as large as $\loc(\alpha)$; on the other hand, that $\pathwidth(\alphagraph)$ can be strictly smaller than $2\loc(\alpha)$ is due to the fact that every marked block of size $1$ is actually represented by only $1$ $\act$ vertex, instead of two. We can formally show that there are rather simple example words $\alpha$ and $\beta$ that reach the extremes of $2 \loc(\alpha) = \pathwidth(\alphagraph)$ and $\loc(\beta) = \pathwidth(\betagraph)$, i.e., the bounds of \ref{pathwidthLemma} are tight. The proof of the following proposition also serves as an introduction to path-decompositions for the graph representation $\alphagraph$ of words (and our use of the terminology explained in Section~\ref{sec:graphTerm}), and therefore as a preparation for the proof of Lemma~\ref{pathwidthLemma}. 

\begin{proposition}\label{tightnessProposition}
Let $\alpha = (x_1 x_2 \ldots x_n x_{n-1} \ldots x_2)^k x_1$ with $n \geq 3$, and let $\beta = (x_1 x_2)^k$. Then we have $\loc(\alpha) = k$ and $\pathwidth(\alphagraph) = 2k$, and $\loc(\beta) = \pathwidth(\betagraph) = k$.
\end{proposition}

\begin{proof}
We start with proving the first statement and first observe that $\loc(\alpha) \leq k$ due to the marking sequence $x_n, x_{n-1}, \ldots, x_1$. In order to show $\pathwidth(\alphagraph) \ge 2k$, we first observe that, for every $i \in \{2, \ldots, n-1\}$, $\pset_{x_i}(\alpha)$ is a clique of size $2k$ in $\alphagraph$, which implies that every path-decomposition $Q$ (interpreted as a pd-marking scheme) for $\alphagraph$ reaches a step where all $2k$ vertices of $\pset_{x_i}(\alpha)$ are $\act$. Now let $Q$ be a path-decomposition for $\alphagraph$ (interpreted as a pd-marking scheme), let $i \in \{2, \ldots, n-1\}$ be such that all $\pset_{x_i}(\alpha)$ are first set to $\act$, i.\,e., when all vertices $\pset_{x_i}(\alpha)$ are $\act$ for the first time, then in every $\pset_{x_j}(\alpha)$, $j \in \{2, \ldots, n-1\} \setminus \{i\}$, there is at least one $\open$ vertex (in particular, no vertex from any $\pset_{x_j}$, $2 \leq j \leq n-1$, is $\closed$). Moreover, in the following we consider the earliest point of $Q$, where all $\pset_{x_i}(\alpha)$ are $\act$. \par
If, at this point, there is some additional $\act$ vertex, then there are $2k+1$ $\act$ vertices; thus, in the following we assume that there are no other $\act$ vertices. If there is also no $\closed$ vertex, then all other vertices are $\open$, which means that every vertex from $\pset_{x_i}(\alpha)$ has at least one adjacent $\open$ vertex and therefore we have to set an $\open$ vertex to $\act$, before we can set a vertex from $\pset_{x_i}(\alpha)$ to $\closed$; this leads to at least $2k + 1$ $\act$ vertices. It remains to consider the case where there is some $\closed$ vertex $j$. This means that all vertices of $\pset_{\alpha[j]}(\alpha)$ are $\closed$, which implies that $j \in \pset_{x_1}(\alpha) \cup \pset_{x_n}(\alpha)$. We first consider the case $j \in \pset_{x_1}(\alpha)$. Since every vertex from $\pset_{x_2}(\alpha)$ is adjacent to some vertex from $\pset_{x_1}(\alpha)$, we can conclude that all vertices from $\pset_{x_2}(\alpha)$ are $\act$, i.\,e., $i = 2$. The assumption $j \in \pset_{x_n}(\alpha)$ analogously leads to the situation that $i = n-1$. Consequently, all $2k$ vertices from $\pset_{x_i}(\alpha)$ are $\act$, either $\pset_{x_1}(\alpha)$ are all $\closed$ or $\pset_{x_n}(\alpha)$ are all $\closed$, and all other vertices are $\open$. In both of these cases, every vertex from $\pset_{x_i}(\alpha)$ has at least one adjacent $\open$ vertex, which, as before, means that we have to set an $\open$ vertex to $\act$, before we can set a vertex from $\pset_{x_i}(\alpha)$ to $\closed$; this, as well, leads to at least $2k + 1$ $\act$ vertices. Consequently, $\width(Q) \geq 2k$, and, with Lemma~\ref{pathwidthLemma}, we can conclude $\pathwidth(\alphagraph) = 2k$.\par
With respect to the second statement, we first note that any marking sequence for $\beta$ leads to $k$ marked blocks, which implies $\loc(\beta) = k$. Moreover, a pd-marking scheme $Q$ with $\width(Q) = k$ can be easily constructed as follows. First, we set all positions of $\pset_{x_1}(\beta)$ to $\act$. Then we set position $2$ to $\act$, position $1$ to $\closed$, position $4$ to $\act$, position $3$ to $\closed$ and so on, until all positions of $\pset_{x_2}(\beta)$ are $\act$ and all positions of $\pset_{x_1}(\beta)$ are $\closed$. Finally, the positions of $\pset_{x_2}(\beta)$ are set to $\closed$. There are at most $k+1$ positions $\act$ at the same time; thus, $\width(Q) = k$ and therefore $\pathwidth(\betagraph) \leq k$. Together with Lemma~\ref{pathwidthLemma}, this implies $\loc(\beta) = \pathwidth(\betagraph) = k$.
\end{proof}

As explained at the beginning of this section, the construction of a graph $\alphagraph$ from a word $\alpha$ does not reduce the decision problem $\locProb$ to $\pathwidthProb$ (since $\pathwidth(\alphagraph)$ lies between $\loc(\alpha)$ and $2\loc(\alpha)$); its main purpose is to obtain approximation results, which is formally stated by the next lemma.

\begin{lemma}\label{pathwidthLemma-apx}
If $\minPathwidthProb$ admits an $\bigo(f(n))$-time $r(\opt,n)$-approximation algorithm, then $\minlocProb$ admits an $\bigo(f(|\alpha|)+|\alpha|^2)$-time $2r(2\opt,|\alpha|)$-approximation algorithm.
\end{lemma}

\begin{proof}
Let $\alpha$ be an instance of $\minlocProb$ and $\mathcal A$ an $r(\opt,n)$-approximation for $\minPathwidthProb$. By Lemma~\ref{pathwidthLemma}, it follows that $\pathwidth(\alphagraph)\leq 2m^*(\alpha)$.\par
Let $Q$ be the path decomposition computed by $\mathcal A$ on $G_\alpha$ and $\sigma$ be the corresponding marking sequence constructed with Lemma~\ref{pathwidthLemma}. With the inequality $m^*(\alpha)\geq \frac{1}{2}\pathwidth(\alphagraph)$, the performance ratio of $\sigma$ can be bounded by $R(\alpha,\sigma)=\frac{\pi_\sigma(\alpha)}{m^*(\alpha)}\leq \frac{2}{\pathwidth(\alphagraph)}\width(Q)\leq 2R(\alphagraph,Q)$.
With $R(\alphagraph,Q)\leq r(\cutwidth(\alphagraph),n)$ from the approximation ratio of $\mathcal A$, $n=|\alpha|$ from the construction of $\alphagraph$, and $\pathwidth(\alphagraph)\leq 2m^*(\alpha)$ from Lemma~\ref{pathwidthLemma}, the claimed bound of $2r(2\opt,|\alpha|)$ on the approximation ratio follows. The  approximation procedure to compute $\sigma$, creates $\alphagraph$ in $\bigo(|\alpha|^2)$, runs $\mathcal A$ in $\bigo(f(|\alpha|))$ and translates the path-decomposition $Q$ into $\sigma$ in $\bigo(|\alpha|)$, which takes an overall running time in $\bigo(f(|\alpha|)+|\alpha|^2)$.
\end{proof}

Consequently, approximation algorithms for $\minPathwidthProb$ carry over to $\minlocProb$. For example, the $\bigo(\sqrt{\log(\opt)} \log(n))$-approximation algorithm for $\minPathwidthProb$ from~\cite{FeigeEtAl2008} implies the following.

\begin{theorem}\label{thm:approxLoc}
There is an $\bigo(\sqrt{\log(\opt)} \log(n))$-approximation algorithm for $\minlocProb$.
\end{theorem}

Another consequence that is worth mentioning is due to the fact that an optimal path decomposition can be computed faster than $\bigo^*(2^n)$. More precisely, it is shown in~\cite{SuchanVillanger2009} that for computing path decompositions, there is an exact algorithm with running time $\bigo^*((1.9657)^n)$, and even an additive approximation algorithm with running time $\bigo^*((1.89)^n)$. Consequently, there is a $2$-approximation algorithm for $\minlocProb$ with running time $\bigo^*((1.9657)^n)$ and an asymptotic $2$-approximation algorithm with running time $\bigo^*((1.89)^n)$ for $\minlocProb$.\par
Many existing algorithms constructing path decompositions are of theoretical interest only, and this disadvantage carries over to the possible algorithms computing the locality number or cutwidth (see Section~\ref{sec:direct}) based on them. However, the reduction of \ref{pathwidthLemma} is also applicable in a purely practical scenario, since any kind of practical algorithm constructing path decompositions can be used  to compute marking sequences (the additional tasks of building $\alphagraph$ and the translation of a path decomposition for it back to a marking sequence are computationally simple). This observation is particularly interesting since developing practical algorithms constructing tree and path decompositions of small width is a vibrant research area. See, e.g., the work~\cite{CoudertEtAl2016} and the references therein for practical algorithms constructing path decompositions; also note that designing exact and heuristic algorithms for constructing tree decompositions was part of the ``PACE 2017 Parameterized Algorithms and Computational Experiments Challenge''~\cite{DellEtAl2017}.\par
As mentioned several times already, our reductions to and from the problem of computing the locality number also establish the locality number for words as a (somewhat unexpected) link between the graph parameters cutwidth and pathwidth. We shall discuss in more detail in Section~\ref{sec:direct} the consequences of this connection. Next, we conclude this section by providing a formal proof of Lemma~\ref{pathwidthLemma}, which is the main result of this section.

\subsection{Proof of Lemma~\ref{pathwidthLemma}}\label{sec:mainProof}

In order to prove Lemma~\ref{pathwidthLemma}, we shall prove the two claims $\pathwidth(\alphagraph) \leq 2\loc(\alpha)$ and $\loc(\alpha) \leq \pathwidth(\alphagraph)$ separately. Recall that for any word $\alpha$, by $G_{\alpha}$ we denote the graph constructed as described in Section~\ref{sec:locToPW}. \par
We first prove $\pathwidth(\alphagraph) \leq 2\loc(\alpha)$. Intuitively speaking, we will translate the stages of a marking sequence $\sigma$ for $\alpha$ into steps of a pd-marking scheme for $G_{\alpha}$ in a natural way: each marked block $\alpha[s..t]$ is represented by letting the border positions $s$ and $t$ be $\act$, the internal position $s+1, s+2, \ldots, t-1$ $\closed$, and all other positions $\open$. In particular, this means that each stage of the marking sequence with $k$ marked blocks is represented by at most $2k$ $\act$ positions in the corresponding step of the pd-marking scheme (note that marked blocks of size $1$ are represented by only one $\act$ position). The difficulty will be to show that in the process of transforming one such step of the pd-marking scheme into the next one, we do not produce more than $2\pi_{\sigma}(\alpha) + 1$ $\act$ positions. This is non-trivial, since due to the cover-property of the pd-marking scheme, we must first set \emph{all} positions to $\act$ that correspond to occurrences of the next symbol to be marked by $\sigma$ before we can set them from $\act$ to $\closed$.

\begin{lemma}\label{mainPatwidthLemmaOne}
Let $\alpha$ be a word. Then $\pathwidth(\alphagraph) \leq 2\loc(\alpha)$.
\end{lemma}

\begin{proof}
We first observe that there is a natural correspondence between any marked version of $\alpha$ and a step of a marking scheme of $G_{\alpha}$ (recall the terminology introduced in the last paragraph of Section~\ref{sec:graphTerm}). More precisely, every marked block $\alpha[s..t]$ can be represented by having the \emph{border positions} $s$ and $t$ of $G_{\alpha}$ marked as $\act$, and all \emph{internal positions} $j$ with $s+1 \leq j \leq t-1$ marked as $\closed$. All other positions that are unmarked in $\alpha$ are $\open$ in $G_{\alpha}$. Note that $s = t$ means that the marked block is represented by only one $\act$ position and no $\closed$ positions, and that $t = s + 1$ means that the marked block is represented by two $\act$ positions and no $\closed$ positions. Hence, each of $\alpha$'s marked blocks of size $1$ is represented by only one $\act$ position, while each marked block of size at least $2$ is represented by two $\act$ positions.\par
In this way, a marking sequence $\sigma = (x_1, x_2, \ldots, x_m)$ for $\alpha$ with $\pi_{\sigma}(\alpha) = k$ translates into steps $p_1, p_2, \ldots, p_m$ (i.\,e., $p_i$ represents stage $i$ of $\sigma$ as described above) of a marking scheme for $G_{\alpha}$. By our observation from above, it is obvious that at each step $p_i$ there are at most $2k_i$ $\act$ vertices, where $k_i$ is the number of marked blocks at stage $i$ of $\sigma$. It now remains to show how the steps $p_1, p_2, \ldots, p_m$ representing $\sigma$'s stages can be transformed into a pd-marking scheme of $G_{\alpha}$. More precisely, we have to describe how we can obtain step $p_{i+1}$ from step $p_i$, how we can reach step $p_1$ from $Q$'s initial step (i.\,e., where all positions are $\open$), and how to transform step $p_m$ into the final step of $Q$ where all positions are $\closed$. Moreover, this should be done in such a way that the marking scheme is a pd-marking scheme, and such that at most $2k + 1$ positions are $\act$ in each step.\par
We can reach step $p_1$ from $Q$'s initial step by just setting all positions of $\pset_{x_1}(\alpha)$ to $\act$ (note that $|\pset_{x_1}(\alpha)| = k_1 \leq k$), and the final step of $Q$ can be obtained from step $p_m$ by setting the only $\act$ positions $1$ and $|\alpha|$ to $\closed$ (note that at stage $m$ of $\sigma$ the whole word is one marked block, so, by definition of step $p_m$, only positions $1$ and $|\alpha|$ are $\act$, while all other positions are $\closed$). Obviously, the maximum number of $\act$ positions for these steps is bounded by $\max\{k_1, 2\} \leq k$. \par
Let $s$ be arbitrary with $1 \leq s \leq m - 1$. We now describe how step $p_{s + 1}$ can be obtained from step $p_s$. 
Intuitively speaking, we first set all positions from $\pset_{x_{s+1}}(\alpha)$ that extend already marked blocks to $\act$ (note that this includes positions that join two already marked blocks). We have to set each of these positions to $\act$ one after the other, and whenever some $\act$ position becomes an internal position of a marked block, then it must be set to $\closed$ so that we do not get too many $\act$ positions. However, we only set internal $\act$ positions to $\closed$ if they are not from $\pset_{x_{s+1}}(\alpha)$, since due to the fact that $\pset_{x_{s+1}}(\alpha)$ is a clique in $G_{\alpha}$, we must reach a point where all positions from $\pset_{x_{s+1}}(\alpha)$ are $\act$  at the same time. After this is done, we mark the remaining positions $\pset_{x_{s+1}}(\alpha)$ that create new marked blocks of size $1$. Let us now formally describe this marking scheme. \par
We call every $j \in \pset_{x_{s+1}}(\alpha)$ \emph{extending}, if marking position $j$ extends an already marked block, and we call it \emph{isolated}, if marking position $j$ creates a new marked block of size $1$. First, we set all extending positions $j \in \pset_{x_{s+1}}(\alpha)$ to $\act$ (in some order), but every time we do this, we perform the following update operation before setting the next position to $\act$: every $\act$ position from $\{1, 2, \ldots, |\alpha|\} \setminus \pset_{x_{s+1}}(\alpha)$ that is an internal position of some marked block is set to $\closed$. As soon as all extending positions are $\act$, we set all isolated $j \in \pset_{x_{s+1}}(\alpha)$ to $\act$.\par
We have now reached the following situation (which we denote by step $p'_s$):
\begin{itemize} 
\item All positions of $\pset_{x_{s+1}}(\alpha)$ are $\act$.
\item All border positions of marked blocks of stage $s+1$ of $\sigma$ are $\act$.
\item All internal positions of marked blocks of stage $s+1$ of $\sigma$ are $\closed$, except possibly some of the positions from $\pset_{x_{s+1}}(\alpha)$ that have now become internal positions of marked blocks.
\end{itemize} 
We can now transform step $p'_s$ into step $p_{s + 1}$ by setting all $\act$ positions from $\pset_{x_{s+1}}(\alpha)$ to $\closed$ that are internal positions of marked blocks. This only decreases the number of $\act$ positions. \par

\begin{figure}[tb]
\begin{center}

\includegraphics{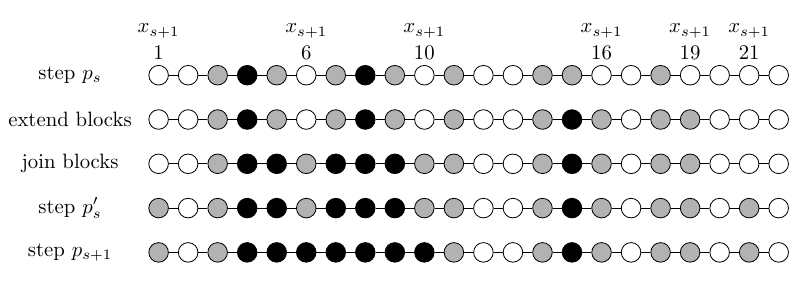}

\end{center}
\caption{An illustration of how step $p_s$ is transformed into step $p_{s + 1}$ (for simplicity, the edges that connect the sets $\pset_{x_i}(\alpha)$ into cliques are omitted). As in Figure~\ref{fig:examplePathwidth}, white vertices are $\open$, grey vertices are $\act$, and black vertices are $\closed$. The upper most graph represents step $p_s$. As indicated above, we have $\pset_{x_i}(\alpha) = \{1, 6, 10, 16, 19, 21\}$; positions $6, 10, 16, 19$ are extending and $1, 21$ are isolated. If we first set the extending positions to $\act$ that do not join marked blocks, i.\,e., positions $16$ and $19$, then we obtain the situation represented by the second graph. Note that after setting $16$ to $\act$, we have to immediately set $15$ to $\closed$, whereas position $19$ does not trigger such an action. Next, we set all remaining extending positions $6$ and $10$ to $\act$, which yields the third graph. Immediately after setting $6$ to $\act$, we have to set both $5$ and $7$ to $\closed$, and immediately after setting $10$ to $\act$, we have to set $9$ to $\closed$. Now it only remains to set the isolated positions $1$ and $21$ to $\act$, which yields the second to last graph corresponding to step $p'_s$. Finally, in order to reach step $p_{s + 1}$, we set all $\act$ positions from $\pset_{x_i}(\alpha)$ to $\closed$ that are internal positions, which sets $6$ and $10$ to $\closed$.}
\label{fig:pspsplusoneIllustration}
\end{figure}

This completes the definition of the marking scheme. Figure~\ref{fig:pspsplusoneIllustration} contains an example of how step $p_{s + 1}$ is obtained from step $p_s$. In this example, we first set extending positions to $\act$ that do not join marked blocks, and then we set the remaining extending positions to $\act$. This is done for illustrational reasons (recall that we have not restricted the order in which we set extending positions to $\act$). \par
Next, we observe that this marking scheme is a pd-marking scheme. To this end, we observe that every edge $\{j, j+1\}$ with $1 \leq j \leq |\alpha| - 1$ is covered, since for every edge $\{j, j+1\}$, we either set $j$ from $\open$ to $\act$ while $j+1$ is already $\act$, or the other way around. Moreover, all positions from $\pset_{x_j}(\alpha)$ are $\act$ at step $p'_j$; thus, the cover property is also satisfied with respect to these edges.\par
Finally, we have to show that in this pd-marking scheme, the maximum number of $\act$ positions is bounded by $2k+1$. This is obviously true at step $p_1$. Now let $s$ with $1 \leq s \leq |\alpha| - 1$ be arbitrary. Since the total number of $\act$ positions at step $p_{s}$ and $p_{s + 1}$ are bounded by $2k$, we only have to show that the maximum number of $\act$ positions in the marking scheme transforming $p_{s}$ into $p_{s + 1}$ is bounded by $2k + 1$. Let us assume that at stage $s$ and $s + 1$ of $\sigma$, there are $k_s$ ($k_{s + 1}$, respectively) marked blocks, and exactly $k_{s, 1}$ ($k_{s + 1, 1}$, respectively) blocks have size $1$; note that this means that at step $p_s$ there are $k_{s, 1} + 2 (k_s - k_{s, 1})$ $\act$ positions. \par
In the first phase of the marking scheme, i.\,e., the phase where we only set extending positions to $\act$, the following different situations can arise, whenever we set some position $j$ to $\act$ (see Figure~\ref{fig:pspsplusoneIllustration} for an illustration):
\begin{enumerate} 
\item\label{extB1} $j$ extends a marked block of size $1$, but does not join two blocks: The number of $\act$ positions increases by $1$.\\ This is due to the fact that by setting $j$ to $\act$, we do not create any internal $\act$ positions that could be set to $\closed$.
\item $j$ extends a marked block of size at least $2$, but does not join two blocks: The number of $\act$ positions increases by $1$ and then decreases by $1$. \\
Assume that the block of size $2$ is extended to the right. Then $j-1$ must be $\act$ and, since the block has size at least $2$, $j-2$ cannot be $\open$. Moreover, since $j-1$ is a neighbour of $j$ it cannot be an element from $\pset_{x_{s+1}}(\alpha)$. This means that $j-1$ is set to $\closed$. 
\item $j$ joins two marked blocks of size $2$: the number of $\act$ positions increases by $1$ and then decreases by $2$.\\
We can argue similarly as in the previous case. Positions $j - 1$ and $j + 1$ must be $\act$ and, since the blocks have size at least $2$, neither $j-2$ nor $j + 2$ can be $\open$. Moreover, since both $j-1$ and $j+1$ are neighbours of $j$ they cannot be elements from $\pset_{x_{s+1}}(\alpha)$. This means that $j-1$ and $j + 1$ are set to $\closed$. 
\item $j$ joins a marked block of size $1$ and a block of size $2$: the number of $\act$ positions increases by $1$ and then decreases by $1$.\\
Without loss of generality, assume that the block of size $2$ is to the left of the block of size $1$. Then $j-1$ must be $\act$ and, since the block has size at least $2$, $j-2$ cannot be $\open$. Moreover, since $j-1$ is a neighbour of $j$ it cannot be an element from $\pset_{x_{s+1}}(\alpha)$. This means that $j-1$ is set to $\closed$.
\item\label{joinB1} $j$ joins two blocks of size $1$: the number of $\act$ positions increases by $1$.\\
This is due to the fact that by setting $j$ to $\act$, we do not create any internal $\act$ positions that could be set to $\closed$.
\end{enumerate} 

We note that only operations of Type~\ref{extB1}~and~\ref{joinB1} increase the overall number of $\act$ positions. In the worst case, we apply all these operations first, before performing the other operations that potentially decrease the number of $\act$ positions. Let us define $\ell_s$ to be the number of $\act$ positions at step $p_s$, and $k_{s, 1}$ to be the number of marked blocks of size $1$ at stage $s$. Since any original marked block of size $1$ can be responsible for at most one operation of Type~\ref{extB1}~or~\ref{joinB1} (after such an operation, the block is not of size $1$ anymore), the maximum number of $\act$ positions after the first phase of the marking is at most $\ell_s + k_{s, 1} + 1$ (we have to count ``$+ 1$'' since also in the operations that do not increase the overall number of $\act$ positions, we always have to first set a position to $\act$ and then, in a new step of the marking scheme, we can set another position to $\closed$). Since at step $p_s$ every marked block of size at least $2$ is represented by two $\act$ positions, and every marked block of size $1$ is represented by only one $\act$ position, we have $\ell_s = 2(k_{s} - k_{s, 1}) + k_{s, 1} = 2k_{s} - k_{s, 1}$, and therefore $\ell_s + k_{s, 1} + 1 = 2k_{s} - k_{s, 1} + k_{s, 1} + 1 = 2k_{s} + 1 \leq 2k + 1$. This means that the number of $\act$ positions is bounded by $2k + 1$ until we reach the situation where we first set an isolated position to $\act$. \par
When we start setting isolated positions to $\act$, we increase the number of $\act$ positions until we have reached step $p'_{s}$. Hence, we have to bound the total number of $\act$ positions at step $p'_{s}$. \par
For the following reasoning, let us assume that in going from stage $s$ to stage $s + 1$ in $\sigma$, we mark all occurrences of $x_{s+1}$ one after the other (instead of all of them in parallel). Each of these individual markings can either create a new marked block of size $1$, or join an existing marked block with another existing marked block, or just extend a marked block (possibly of size $1$) without joining any marked blocks. Let us assume that creating a marked block of size $1$ happens $q$ times, and joining two marked blocks happens $t$ times (how often we extend marked blocks without joining is not important).\par
We first count the number of $\act$ positions at step $p'_{s}$ that correspond to border positions of marked blocks. For each marked block of size at least $2$ there are $2$ such $\act$ positions, while for each block of size $1$ there is only $1$ such $\act$ position. Consequently, there are $2(k_{s+1} - k_{s+1, 1}) + k_{s+1, 1} = 2k_{s+1} - k_{s+1, 1}$ such border positions, where $k_{s+1, 1}$ is the number of marked blocks of size $1$. Since $q \leq k_{s+1, 1}$, we have $2k_{s+1} - k_{s+1, 1} \leq 2k_{s+1} - q$ border positions. In addition to these border positions, we also have a number of $\act$ positions that are internal positions of marked blocks. However, each such internal $\act$ position results from joining two blocks, which means that we have $r$ such positions. Hence, we have at most $2k_{s+1} - q + r$ $\act$ positions. Since stage $s+1$ of $\sigma$ is obtained from stage $s$ by creating $q$ new blocks and joining $r$ marked blocks (and extending some blocks), we have $k_{s+1} = k_s + q - r$. Consequently, $2k_{s+1} - q + r = k_{s+1} + k_{s} + q - r - q + r = k_{s+1} + k_s \leq 2k$. This means that the maximum number of $\act$ positions in the marking scheme that transforms step $p_s$ into step $p_{s+1}$ is bounded by $2k + 1$.\par
Consequently, we have described a pd-marking scheme for $G_{\alpha}$ that has always at most $2k + 1$ $\act$ positions, which means that $\pathwidth(G_{\alpha}) \leq 2k$.
\end{proof}

Next, we take care of the other inequality $\loc(\alpha) \leq \pathwidth(\alphagraph)$. On an intuitive level, the proof will proceed as follows. Any pd-marking scheme $Q$ for $\alphagraph$ induces a linear order on $\{x_1, x_2, \ldots, x_m\}$ (and therefore a marking sequence $\sigma$), since it is forced to go through individual steps where all positions of the cliques $\pset_{x_i}(\alpha)$ with $1 \leq i \leq m$ are $\act$ at the same time. It is our goal to prove that there are at least $\pi_{\sigma}(\alpha) + 1$ $\active$ position in the pd-marking scheme $Q$, since this implies that $\loc(\alpha) \leq \pi_{\sigma}(\alpha) \leq \width(Q)$.\par
The marking sequence $q$ has a stage $s$ in which the maximum of $\pi_{\sigma}(\alpha)$ marked blocks is reached for the first time. In the corresponding step $p_s$ of the pd-marking scheme (i.\,e., the step where the positions of $\pi_{\sigma}(\alpha)$ are all $\act$ for the first time), we obviously cannot assume that the marked blocks are represented in the way of the proof of Lemma~\ref{mainPatwidthLemmaOne} (i.\,e., border positions are $\act$, internal positions are $\closed$, and all other positions are $\open$). However, by carefully analysing step $p_s$, we can identify at least $\pi_{\sigma}(\alpha)$ $\act$ positions (for this, we need the property that $Q$ is a pd-marking scheme for $\alphagraph$ and that a maximum number of marked blocks is reached at stage $s$ of $\sigma$). If now there is one additional $\act$ position, then there are $\pi_{\sigma}(\alpha) + 1$ $\act$ positions and we are done. This is the easy part of the proof, and the the more difficult part is the case where we do not have an additional $\act$ position, i.\,e., the identified $\pi_{\sigma}(\alpha)$ $\act$ positions are the only $\act$ positions. This property, however, can be shown to impose some structural constraints with respect to step $p_s$ of the pd-marking sequence and stage $s$ of the marking sequence. In particular, by some more technical combinatorial observations and exhaustive case distinctions, we are able to prove that we will necessarily get $\pi_{\sigma}(\alpha) + 1$ $\act$ positions in the next step of the pd-marking sequence, or we can prove that the marking sequence $\sigma$ can be changed into a better marking sequence $\sigma'$ with $\pi_{\sigma'}(\alpha) = \pi_{\sigma}(\alpha) - 1$. This latter property means that we have $\loc(\alpha) \leq \pi_{\sigma}(\alpha) - 1 \leq \width(Q)$.

\begin{lemma}
Let $\alpha$ be a word with $|\alpha| \geq 2$. Then, given any path-decomposition $Q$ for $\alphagraph$, a marking sequence $\sigma$ for $\alpha$ with $\loc(\sigma)\leq \width(Q)$ can be constructed in $\bigo(|\alpha|)$. In particular, $\loc(\alpha) \leq \pathwidth(\alphagraph)$.
\end{lemma}

\begin{proof}
Let $Q = (B_0, B_1, B_2, \ldots, B_{2|\alpha|})$ be an arbitrary nice path-decomposition for $\alphagraph$, which we interpret as a pd-marking scheme (see the last paragraph of Section~\ref{sec:graphTerm}). For every $i$, $1 \leq i \leq m$, $\pset_{x_i}(\alpha)$ is a clique in $G_{\alpha}$; thus, there must be a step $p_i$ of this pd-marking sequence in which all positions of $\pset_{x_i}(\alpha)$ are $\act$ for the first time. Without loss of generality, we assume that $p_1 < p_2 < \ldots < p_m$. Next, let $\sigma = (x_1, x_2, \ldots, x_m)$ be the marking sequence for $\alpha$ that marks the symbols $x_1, x_2, \ldots, x_m$ in the same order as their occurrences $\pset_{x_i}(\alpha)$ are all together $\act$ for the first time in the pd-marking scheme. Thus $\sigma$ can be constructed from $Q$ in time $\bigo(|\alpha|)$.  We now prove that for $k = \pi_{\sigma}(\alpha)$ one of the following cases holds:
\begin{itemize}
\item There is a step of $Q$ with at least $k + 1$ $\act$ positions. In this case, we have $\loc(\alpha) \leq k \leq \width(Q)$.
\item There is a step of $Q$ with at least $k$ $\act$ positions and a marking sequence $\sigma'$ with $\pi_{\sigma'}(\alpha) = k-1$. In this case, we have $\loc(\alpha) \leq k - 1 = \width(Q)$.
\end{itemize}
Since the path decomposition is arbitrary, this implies that $\loc(\alpha) \leq \width(Q)$ for every path decomposition $Q$, and therefore also $\loc(\alpha) \leq \pathwidth(\alphagraph)$.\par
Let $s$, $1 \leq s \leq m$, be chosen such that the maximum number of marked blocks is reached for the first time at stage $s$ of $\sigma$, i.\,e., after marking symbol $x_s$, we obtain $k$ marked blocks for the first time. As defined above, $p_s$ is the step of $Q$ where all positions of $\pset_{x_s}(\alpha)$ are $\act$ for the first time. We now represent the pd-marking scheme at step $p_s$ and the marked version of $\alpha$ at stage $s$ of $\sigma$ as a single marked word $\widehat{\alpha}$ over the alphabet $\{\opensym, \actsym, \closesym\}$. More precisely, for every $i$ with $1 \leq i \leq |\alpha|$, $\widehat{\alpha}[i] = \opensym$ if position $i$ is $\open$ at step $p_s$, $\widehat{\alpha}[i] = \actsym$ if position $i$ is $\act$ at step $p_s$ and $\widehat{\alpha}[i] = \closesym$ if position $i$ is $\closed$ at step $p_s$ of the pd-marking scheme. This fully describes step $p_s$ of the pd-marking scheme. Moreover, we represent the marked version of $\alpha$ at stage $s$ of $\sigma$ by marking the symbols of $\widehat{\alpha}$ in the same way, i.\,e., for every $i$ with $1 \leq i \leq |\alpha|$, the symbol $\widehat{\alpha}[i]$ is marked if $\alpha[i]$ is marked at stage $s$ of $\sigma$ (i.\,e., $\alpha[i] \in \{x_1, x_2, \ldots, x_{s}\}$), and $\widehat{\alpha}[i]$ is unmarked otherwise. We shall also consider $\widehat{\alpha}$'s factorisation according to its marked and unmarked blocks, i.\,e., we consider the factorisation
\begin{equation*}
\widehat{\alpha} = \beta_0 \mu_1 \beta_1 \mu_2 \ldots \mu_k \beta_k\,,
\end{equation*}
where the factors $\beta_i$, $0 \leq i \leq k$, correspond to the unmarked blocks of $\widehat{\alpha}$, and $\mu_{i}$, $1 \leq i \leq k$, correspond to the marked blocks of $\widehat{\alpha}$. Next, we establish some simple properties of $\widehat{\alpha}$ and its factorisation. \par

\begin{enumerate}
\item\label{simpleObsOne} If $\widehat{\alpha}[i] = \closesym$ for some $i$, $1 \leq i \leq |\alpha|$, then position $i$ is $\closed$ at step $p_s$ of the pd-marking scheme, which means that the situation where all positions of $\pset_{\alpha[i]}(\alpha)$ have been $\act$ at the same time must have occurred already. This means that $\alpha[i] \in \{x_1, x_2, \ldots, x_s\}$. Hence, position $i$ is marked and therefore in some marked block. Consequently, all of $\widehat{\alpha}$'s occurrences of $\closesym$ occur in marked blocks $\mu_j$, i.\,e., $\beta_0, \beta_k \in \{\actsym, \opensym\}^*$, $\beta_{i} \in \{\actsym, \opensym\}^+$ and $\mu_{i} \in \{\actsym, \closesym\}^+$ for every $i$ with $1 \leq i \leq k$ (note that the factors $\mu_i$ cannot contain occurrences of $\opensym$ by definition of $\sigma$ (i.\,e., an $x_i$ is marked by $\sigma$ not before all positions $\pset_{x_i}(\alpha)$ are $\act$)).
\item\label{simpleObsTwo} For every $i$, $1 \leq i \leq k-1$, if the last symbol of the marked block $\mu_i$ is $\closesym$, then the first symbol of the unmarked block $\beta_i$ must be $\actsym$ (or $\beta_i = \varepsilon$, which can happen for $i = k$), since otherwise $\beta_i$'s first symbol must be $\opensym$ (see Point~\ref{simpleObsOne}), which leads to the contradiction that at step $p_s$ of the pd-marking scheme there is a $\closed$ position adjacent to an $\open$ position (this violates the cover property of path decompositions). Analogously, it follows that if the first symbol of the marked block $\mu_{i}$ is $\closesym$, then the first symbol of the unmarked block $\beta_{i - 1}$ must be $\actsym$ (or $\beta_{i - 1} = \varepsilon$, which can happen for $i = 1$). 
\item\label{simpleObsThree} For every $i$ with $1 \leq i \leq |\alpha|$ such that $\alpha[i] = x_s$, position $i$ must be contained in some marked block of $\widehat{\alpha}$ (since $x_s$ is marked at stage $s$ of $\sigma$), and (by definition of $\sigma$) position $i$ must be $\act$ at step $p_s$ of the pd-marking scheme, i.\,e., $\widehat{\alpha}[i] = \actsym$. Moreover, there must be at least one such $i$ with $1 \leq i \leq |\alpha|$ such that $\alpha[i] = x_s$.
\end{enumerate}

With Point~\ref{simpleObsThree}, there is some $j$ with $1 \leq j \leq |\alpha|$ such that $\alpha[j] = x_s$, and some $r$ with $1 \leq r \leq k$ such that $j$ is a position of the marked block $\mu_{r}$ of $\widehat{\alpha}$. Next, for every $i$ with $1 \leq i \leq k$, we define a position $t_i$ with $\widehat{\alpha}[t_i] = \actsym$ that either lies in $t_i$, or is the first position of $\beta_i$, or the last position of $\beta_{i-1}$. First, we set $t_r = j$. For every $i$, $1 \leq i < r$, we let $t_i$ be some position of $\mu_i$ that is an occurrence of $\actsym$ if one exists. If, on the other hand, $\mu_i$ has no occurrences of $\actsym$, then, due to Point~\ref{simpleObsOne}, $\mu_i$'s last symbol is $\closesym$, and, by Point~\ref{simpleObsTwo}, this means that $\beta_i$'s first symbol is $\actsym$, so we let $t_i$ be $\beta_i$'s first position. Analogously, for every $i$ with $r < i \leq k$, we let $t_i$ be some position of $\mu_i$ that is an occurrence of $\actsym$ if one exists, and if $\mu_i$ has no occurrences of $\actsym$, then its first symbol is $\closesym$, which means that $\beta_{i-1}$'s last symbol is $\actsym$, so we let $t_i$ be $\beta_{i-1}$'s last position.\par
Since every $t_i$ with $1 \leq i < r$ is in $\mu_i \beta_i[1]$, every $t_i$ with $r < i \leq k$ is in $\beta_{i-1}[|\beta_i|] \mu_i$, and $t_r$ is in $\mu_r$, these positions $t_i$ are in fact $k$ distinct positions that are $\act$ at step $p_s$ of the pd-marking scheme.\par
Now, if there is at least one additional $\act$ position, then there are at least $k+1$ $\act$ positions at step $p_s$, i.\,e., we have arrived at the first of the two cases mentioned above. In order to conclude the proof, we assume now that the $\act$ positions $t_i$, $1 \leq i \leq k$, are the only $\act$ positions at step $p_s$ of the pd-marking scheme, and we will show that this either leads again to the first case, but with respect to some other step of the pd-marking scheme, or to the second case mentioned above, i.\,e., there is a marking sequence $\sigma'$ with $\pi_{\sigma'}(\alpha) = k-1$. \par
First, we divide $\widehat{\alpha}$ into the part left of $\mu_r$, the factor $\mu_r$ and the part right of $\mu_r$, i.\,e., we consider the factorisation $\widehat{\alpha} = \widehat{\alpha}_{1} \mu_r \widehat{\alpha}_{2}$, where we call $\widehat{\alpha}_{1} = \beta_0 \mu_1 \beta_1 \mu_2 \ldots \beta_{r-1}$ the \emph{left side} and we call $\widehat{\alpha}_{2} = \beta_{r} \mu_{r+1} \beta_{r+1} \ldots \mu_k \beta_k$ the \emph{right side}. By our assumption that the positions $t_i$ with $1 \leq i \leq k$ are the only occurrences of $\actsym$ in $\widehat{\alpha}$, and the Points~\ref{simpleObsOne}~to~\ref{simpleObsThree} from above, we can conclude several facts about the form of the left and the right side. In the following, we shall only analyse the left side; all the following arguments apply analogously to the right side as well. \par
Each position $t_{\ell}$ with $1 \leq \ell < r$ is either inside $\mu_{\ell}$, or it is the leftmost position of $\beta_{\ell}$. If it is the leftmost position of $\beta_{\ell}$, then there is no occurrence of $\actsym$ in $\mu_{\ell}$, which means that the leftmost symbol in $\mu_{\ell}$ must be $\closesym$, and therefore, the rightmost symbol of $\beta_{\ell - 1}$ must be an $\actsym$ (note that $\beta_{\ell - 1} \in \{\actsym, \opensym\}^+$ (see Point~\ref{simpleObsOne}) and it is not possible that a symbols $\opensym$ occurs next to a symbol $\closesym$). However, this rightmost position must be the position $t_{\ell - 1}$ and is therefore also the leftmost position of $\beta_{\ell - 1}$. In particular, this means that $\beta_{\ell-1} = \actsym$ and $\mu_{\ell-1} \in \{\closesym\}^+$. This inductively proceeds to the left and therefore also means that $\beta_0 = \varepsilon$. Therefore, if $\ell$ with $1 \leq \ell < r$ is maximal such that $t_{\ell}$ is not inside $\mu_{\ell}$, then we have the following situation:

\begin{equation*}
\widehat{\alpha}_1 = \mu_1 \underbrace{\actsym}_{\beta_1} \mu_2 \underbrace{\actsym}_{\beta_2} \ldots \mu_{\ell-1} \underbrace{\actsym}_{\beta_{\ell-1}} \mu_{\ell} \underbrace{\actsym \opensym^{g_1}}_{\beta_{\ell}} \mu_{\ell + 1} \beta_{\ell + 1} \mu_{\ell + 2} \ldots \beta_{r-2} \mu_{r-1} \beta_{r-1}  \,,
\end{equation*}
where $\mu_i \in \{\closesym\}^+$ for every $i$ with $1 \leq i \leq \ell$, and $g_1 \geq 0$. Moreover, since $\ell$ is maximal, all $\mu_{i}$ with $\ell + 1 \leq i \leq r-1$ contain an $\act$ position, which means that $\beta_{i} \in \{\opensym\}^+$ for every $i$ with $\ell + 1 \leq i \leq r - 1$. However, this directly implies that $\mu_{i} = \actsym$ with $\ell + 2 \leq i \leq r-1$ and $\mu_{\ell+1} = \closesym^{g_2} \actsym$ for some $g_2 \geq 0$ with the property that at most one of $g_1$ and $g_2$ can be positive. More precisely, we have the following situation:
\begin{equation}\label{mainEqu}
\widehat{\alpha}_1 = \mu_1 \underbrace{\actsym}_{\beta_1} \mu_2 \underbrace{\actsym}_{\beta_2} \ldots \mu_{\ell-1} \underbrace{\actsym}_{\beta_{\ell-1}} \mu_{\ell} \underbrace{\actsym \opensym^{g_1}}_{\beta_{\ell}} \underbrace{\closesym^{g_2} \actsym}_{\mu_{\ell + 1}} \beta_{\ell + 1} \underbrace{\actsym}_{\mu_{\ell + 2}} \ldots \beta_{r-2} \underbrace{\actsym}_{\mu_{r-1}} \beta_{r-1}  \,,\tag{$\dagger$}
\end{equation}
where $\beta_i \in \{\opensym\}^+$ for every $i$ with $\ell + 1 \leq i \leq r - 1$, $\mu_i \in \{\closesym\}^+$ for every $i$ with $1 \leq i \leq \ell$, and $g_1, g_2 \geq 0$ with $0 \in \{g_1, g_2\}$.\par
If, on the other hand, no such $\ell$, $1 \leq \ell < r$, exists, then all the $\act$ positions $t_i$ are in $\mu_i$ for every $i$ with $1 \leq i \leq r-1$. In particular, this means that $\beta_i \in \{\opensym\}^+$ for every $i$ with $1 \leq i \leq r-1$, which forces all $\mu_i$, $2 \leq i \leq r-1$, to start and end with an $\act$ position, while $\mu_1$ must have a rightmost $\act$ position, but a leftmost $\act$ position only if $\beta_0 \neq \varepsilon$. Thus,
\begin{equation*}\label{mainEquNoEll}
\widehat{\alpha}_1 = \beta_0 \underbrace{\closesym^{g} \actsym}_{\mu_{1}} \beta_1 \underbrace{\actsym}_{\mu_{2}} \ldots \beta_{r-2} \underbrace{\actsym}_{\mu_{r-1}} \beta_{r-1}\,,\tag{$\star$}
\end{equation*}
where $\beta_i \in \{\opensym\}^+$  for every $i$ with $1 \leq i \leq r-1$, $g \geq 0$, and $g > 0$ implies $\beta_0 = \varepsilon$.\par
As mentioned above, these observations also apply to the right side $\widehat{\alpha}_2$ in an analogous way. \par
Next, we turn to the marked block $\mu_r$ in between the left and right side of $\widehat{\alpha}$. Since $\mu_r$ contains exactly one occurrence of $\actsym$, we can factorise it into $\mu_r = \nu_1 \actsym \nu_2$, where $\nu_1, \nu_2 \in \{\closesym\}^*$. We now consider four individual cases that arise from whether or not $\nu_1$ or $\nu_2$ are empty. For each of these cases, we can use the observations from above in order to further determine the structure of the left or right side of $\widehat{\alpha}$. 

\begin{enumerate}[start=1,label={\bfseries Claim (\arabic*)},align=left]
\item \label{CaseLeftNonempty} If $\nu_1 \neq \varepsilon$, then we have 
\begin{equation*}
\widehat{\alpha}_1 = \mu_1 \actsym \mu_2 \actsym \ldots \mu_{r-1} \actsym\,.
\end{equation*}
\emph{Proof}: If $\nu_1 \neq \varepsilon$, then $\nu_1[1] = \closesym$, which implies that $\beta_{r-1}[|\beta_{r-1}|] = \actsym$ and therefore $\beta_{r-1} = \actsym$. This means that for $\ell = r-1$, we have the case described in~(\ref{mainEqu}), i.\,e., where $t_{\ell}$ is the rightmost $\act$ position that is not in $\mu_{\ell}$ and, since $\beta_{\ell} = \actsym$, we also have the case $g_1 = 0$. This directly implies the statement claimed above. \hfill $\square$
\item\label{CaseRightNonempty} If $\nu_2 \neq \varepsilon$, then we have 
\begin{equation*}
\widehat{\alpha}_2 = \actsym \mu_{r+1} \actsym \mu_{r+2} \ldots \actsym \mu_k\,.
\end{equation*}
\emph{Proof}: Analogous to~\ref{CaseLeftNonempty}. \hfill $\square$
\item\label{CaseLeftEmpty} If $\nu_1 = \varepsilon$, then we have one of the following two cases:
\begin{enumerate}
\item\label{CaseLeftEmptySubOne} For some $\ell$ with $1 \leq \ell \leq r-1$, 
\begin{equation*}
\widehat{\alpha}_1 = \mu_1 \underbrace{\actsym}_{\beta_1} \mu_2 \underbrace{\actsym}_{\beta_2} \ldots \mu_{\ell-1} \underbrace{\actsym}_{\beta_{\ell-1}} \mu_{\ell} \underbrace{\actsym \opensym^{g_1}}_{\beta_{\ell}} \underbrace{\closesym^{g_2} \actsym}_{\mu_{\ell + 1}} \beta_{\ell + 1} \underbrace{\actsym}_{\mu_{\ell + 2}} \ldots \beta_{r-2} \underbrace{\actsym}_{\mu_{r-1}} \beta_{r-1}  \,,
\end{equation*}
where $\beta_i \in \{\opensym\}^+$ for every $i$ with $\ell + 1 \leq i \leq r - 1$, $\mu_i \in \{\closesym\}^+$ for every $i$ with $1 \leq i \leq \ell$, and $g_1, g_2 \geq 0$ with $0 \in \{g_1, g_2\}$. Note that this is exactly the case described in Equation~\ref{mainEqu}.
\item\label{CaseLeftEmptySubTwo} $\widehat{\alpha}_1 = \beta_0 \underbrace{\closesym^{g} \actsym}_{\mu_{1}} \beta_1 \underbrace{\actsym}_{\mu_{2}} \ldots \beta_{r-2} \underbrace{\actsym}_{\mu_{r-1}} \beta_{r-1}$,\\
where $\beta_i \in \{\opensym\}^+$ for every $i$ with $1 \leq i \leq r-1$, $g \geq 0$, and $g > 0$ implies $\beta_0 = \varepsilon$. Note that this is exactly the case described in Equation~\ref{mainEquNoEll}.
\end{enumerate}
\emph{Proof}: If there is some $\ell'$, $1 \leq \ell' \leq r-1$, such that $t_{\ell'}$ is not in $\mu_{\ell'}$, then we can consider a maximal $\ell$ with this property and can conclude that we have the case described in~(\ref{mainEqu}), which is exactly the statement of~\ref{CaseLeftEmptySubOne}. If, on the other hand, no such $\ell'$ exists, then we have the case described in~(\ref{mainEquNoEll}), which is exactly the statement of~\ref{CaseLeftEmptySubTwo}.\hfill $\square$
\item\label{CaseRightEmpty} If $\nu_2 = \varepsilon$, then we have one of the following two cases:
\begin{enumerate}
\item\label{CaseRightEmptySubOne} For some $\ell$ with $r \leq \ell \leq k$, 
\begin{equation*}
\widehat{\alpha}_2 = \beta_{r} \underbrace{\actsym}_{\mu_{r+1}} \beta_{r+1} \underbrace{\actsym}_{\mu_{r+2}} \ldots \beta_{\ell-1} \underbrace{\actsym \closesym^{g_1}}_{\mu_{\ell}} \underbrace{\opensym^{g_2} \actsym}_{\beta_{\ell}} \mu_{\ell+1} \underbrace{\actsym}_{\beta_{\ell + 1}} \mu_{\ell+2} \underbrace{\actsym}_{\beta_{\ell + 2}} \ldots \underbrace{\actsym}_{\beta_{k-1}} \mu_{k}\,,
\end{equation*}
where $\beta_i \in \{\opensym\}^+$ for every $i$ with $r  \leq i \leq \ell - 1$, $\mu_i \in \{\closesym\}^+$ for every $i$ with $\ell + 1 \leq i \leq k$, and $g_1, g_2 \geq 0$ with $0 \in \{g_1, g_2\}$.
\item\label{CaseRightEmptySubTwo} $\widehat{\alpha}_2 = \beta_{r} \underbrace{\actsym}_{\mu_{r+1}} \beta_{r+1} \underbrace{\actsym}_{\mu_{r+2}} \ldots \underbrace{\actsym \closesym^{g}}_{\mu_{1}} \beta_k$,\\
where $\beta_i \in \{\opensym\}^+$ for every $i$ with $r \leq i \leq k - 1$, $g \geq 0$, and $g > 0$ implies $\beta_k = \varepsilon$.\par
\end{enumerate}
\emph{Proof}: Analogous to~\ref{CaseLeftEmpty}. \hfill $\square$
\end{enumerate}

\noindent In order to conclude the proof, we need some preliminary observations and definitions. 

\begin{quote}
\textbf{Observation ($\ast$)}: Let $i$ be arbitrary with $1 \leq i \leq |\alpha|$. Position $i$ is marked at stage $s$ of $\sigma$ (and therefore in some marked block of $\widehat{\alpha}$) if and only if $\alpha[i] \in \{x_1, x_2, \ldots, x_{s}\}$. If $\alpha[i] = x_s$, then position $i$ must be $\act$ at step $p_s$ (and therefore $\widehat{\alpha}[i] = \actsym$). If $\alpha[i] \in \{x_1, x_2, \ldots, x_{s-1}\}$, then position $i$ must be $\act$ or $\closed$ at step $p_s$ (and therefore $\widehat{\alpha}[i] \in \{\actsym, \closesym\}$). 
\end{quote}

Let $i$ with $1 \leq i \leq |\alpha|$ be some position that is $\act$ at step $p_s$ of the pd-marking scheme. We say that $i$ is \emph{blocked} if it is unmarked in $\widehat{\alpha}$ or if $\opensym \in \{\widehat{\alpha}[i-1], \widehat{\alpha}[i+1]\}$. The idea of this definition is that if $i$ is blocked, then it cannot be set from $\act$ to $\closed$ in the next step of the pd-marking scheme. Indeed, if $i$ is not marked, then $\alpha[i] \in \{x_{s+1}, x_{s+2}, \ldots, x_m\}$ (see Observation~($\ast$)), and, by assumption, for symbols $x \in \{x_{s + 1}, x_{s + 2}, \ldots, x_m\}$ we have not yet reached the situation that all positions from $\pset_{x}(\alpha)$ are $\act$ at the same time); thus, none of the corresponding positions can be set to $\closed$ in the next step. Moreover, if $\opensym \in \{\widehat{\alpha}[i-1], \widehat{\alpha}[i+1]\}$, then the $\act$ position $i$ is adjacent to an $\open$ position and therefore cannot be set to $\closed$ in the next step.\par

We next show the following claim: 

\begin{quote}
\textbf{Claim ($\ast \ast$)}: Every position on the left side and on the right side that is $\act$ at step $p_s$ of the pd-marking scheme is blocked. 
\end{quote}

\emph{Proof}: Due to~\ref{CaseLeftNonempty}~to~\ref{CaseRightEmpty}, we know that every position $i$ that is $\act$ in step $p_s$ of the pd-marking scheme and corresponds to a marked occurrence of $\actsym$ on the left side, satisfies $\widehat{\alpha}[i+1] = \opensym$. Moreover, every position $i$ that is $\act$ in step $p_s$ of the pd-marking scheme and corresponds to a marked occurrence of $\actsym$ on the right side, satisfies $\widehat{\alpha}[i-1] = \opensym$. Consequently, every position on the left and right side that is $\act$ at step $p_s$ of the pd-marking scheme is blocked. \qed\medskip

We are now ready to finally conclude this proof. To this end, we analyse all possible cases of how $\mu_r$ may look like, and we show that for each case we either obtain a contradiction, or there is a step of $Q$ with at least $k + 1$ $\act$ positions, or there is a marking sequence $\sigma'$ with $\pi_{\sigma'}(\alpha) = k-1$ (see the two cases mentioned at the beginning of this proof). \par
First, we note that if $t_r$ is blocked, then, due to Claim~$(\ast \ast)$, every position that is $\act$ at step $p_s$ of the pd-marking scheme is blocked. Thus, an $\open$ position will be set to $\act$ in the next step of the pd-marking scheme, which means that there are $k + 1$ $\act$ positions in the next step of the pd-marking scheme.\par
If $t_r$ is not blocked, then, since $t_r$ is marked in $\widehat{\alpha}$, we must have $\widehat{\alpha}[t_r - 1] \neq \opensym$ and $\widehat{\alpha}[t_r + 1] \neq \opensym$. We consider four cases that arise from whether $\nu_1$ or $\nu_2$ are empty.

\begin{itemize}
\item $\nu_1 = \varepsilon$ and $\nu_2 \neq \varepsilon$: Since $\nu_1 = \varepsilon$, position $t_r$ is a right neighbour of $\beta_{r-1}$. If the last symbol of $\beta_{r-1}$ is an occurrence of $\opensym$, then $t_r$ is blocked, which is a contradiction to our assumption that $t_r$ is not blocked. Therefore, the last symbol of $\beta_{r-1}$ is an occurrence of $\actsym$, which, according to~\ref{CaseLeftEmpty}, is only possible if we have the situation described in~\ref{CaseLeftEmptySubOne} with $\ell = r - 1$ and $g_1 = 0$. Indeed, if~\ref{CaseLeftEmptySubTwo} applies or if~\ref{CaseLeftEmptySubOne} applies with $\ell < r - 1$, then $\beta_{r-1} \in \{\opensym\}^+$, and if~\ref{CaseLeftEmptySubOne} applies with $g_1 > 0$, then $\beta_{r-1} = \actsym \opensym^{g_1}$.\par
\ref{CaseRightNonempty} implies that there is no marked position in the right side of $\widehat{\alpha}$ that is also $\act$ at step $p_s$ of the pd-marking scheme. Furthermore, since we have~\ref{CaseLeftEmptySubOne} with $\ell = r - 1$, there is also no marked position in the left side of $\widehat{\alpha}$ that is also $\act$ at step $p_s$ of the pd-marking scheme. Consequently, Observation ($\ast$) implies that $t_r$ is the only position with $\alpha[i] = x_s$. Since $x_s$ is marked in stage $s$ of $\sigma$, $\nu_2$ corresponds to a non-empty factor of $\alpha$ that is marked at stage $s-1$ of $\sigma$, and position $t_r - 1$ is not marked in stage $s - 1$ of $\sigma$, we know that marking $x_s$ in stage $s$ just extends the marked block that corresponds to $\nu_2$. Hence, at stage $s-1$ there were $k$ marked blocks, which is a contradiction to the assumption that $s$ is the first stage of $\sigma$ where we reach the maximum of $k$ marked blocks. 
\item $\nu_1 \neq \varepsilon$ and $\nu_2 = \varepsilon$: This case leads to a contradiction analogously to the previous case.
\item $\nu_1 \neq \varepsilon$ and $\nu_2 \neq \varepsilon$: Due to~\ref{CaseLeftNonempty}~and~\ref{CaseRightNonempty}, there is no marked position in the right or left side of $\widehat{\alpha}$ that is also $\act$ at step $p_s$ of the pd-marking scheme. Hence, Observation ($\ast$) again implies that $t_r$ is the only position with $\alpha[i] = x_s$. Since $x_s$ is marked in stage $s$ of $\sigma$, and since both $\nu_1$ and $\nu_2$ correspond to non-empty factors of $\alpha$ that are marked at stage $s-1$ of $\sigma$, we know that marking $x_s$ in stage $s$ joins two marked blocks and changes no other marked block. Hence, at stage $s-1$ there were $k+1$ marked blocks, which is a contradiction.
\item $\nu_1 = \varepsilon$ and $\nu_2 = \varepsilon$: Since $t_r$ is not blocked and $\beta_{r-1} \widehat{\alpha}[t_r] \beta_{r}$ is a factor of $\widehat{\alpha}$, we know that neither the last symbol of $\beta_{r-1}$ nor the first symbol of $\beta_{r}$ can be occurrences of $\open$. Just like in the first and second case,~\ref{CaseLeftEmpty}~and~\ref{CaseRightEmpty} imply that this is only possible if with respect to the left side, we have the situation described in~\ref{CaseLeftEmptySubOne} with $\ell = r - 1$ and $g_1 = 0$, and with respect to the right side, we have the situation described in~\ref{CaseRightEmptySubOne} with $\ell = r$ and $g_2 = 0$. This means that we have the following situation 

\begin{equation*}
\widehat{\alpha} = \mu_1 \underbrace{\actsym}_{\beta_1} \mu_2 \underbrace{\actsym}_{\beta_2} \ldots \mu_{r-1} \underbrace{\actsym}_{\beta_{r-1}} \underbrace{\actsym}_{\mu_r} \underbrace{\actsym}_{\beta_{r}} \mu_{r+1} \underbrace{\actsym}_{\beta_{r+1}} \ldots \mu_{k-1} \underbrace{\actsym}_{\beta_{k-1}} \mu_k\,.
\end{equation*}

Hence, $t_r$ is the only $\act$ position that is marked in $\widehat{\alpha}$, which means that $t_r$ is the only position with $\alpha[i] = x_s$. In particular, $t_r$ is the only position that is unmarked in stage $s-1$ and marked in stage $s$ of $\sigma$. This implies that in stage $s-1$ of $\sigma$ there are exactly $k-1$ marked blocks (i.\,e., the factors $\mu_j$ with $1 \leq j \leq k$ and $j \neq r$) and, by our assumption that stage $s$ is the first stage with $k$ marked blocks, we also know that in stages $1, 2, \ldots, s-1$ the maximum number of marked blocks is $k-1$. Moreover, at stage $s-1$, every unmarked position except $t_r$ (i.\,e., all the positions corresponding to the occurrences of $\actsym$, except the occurrence $\widehat{\alpha}[t_r] = \actsym$) is a neighbour of some marked block. Consequently, we can change $\sigma$ into a marking sequence $\sigma'$ as follows. The marking sequence $\sigma'$ simulates $\sigma$ up to stage $s-1$. As observed above, so far the maximum number of marked blocks is $k-1$. Then, instead of marking $x_s$, $\sigma'$ marks all other unmarked symbols in some order. In each of the corresponding stages of the marking sequence, marking the next symbol leaves the number of marked blocks unchanged, or decreases it (this can be easily seen by consulting the factorisation illustrated above). Finally, symbol $x_s$ is marked as the last symbol. Thus, $\sigma'$ is a marking sequence for $\alpha$ with $\pi_{\sigma'}(\alpha) = k-1$.\qedhere
\end{itemize}

\end{proof}

\section{A New Relationship Between Pathwidth and Cutwidth}\label{sec:direct}

We observe that the reduction from $\minCutwidthProb$ to $\minlocProb$ from Section~\ref{sec:locToCutwidth} combined with the reduction from $\minlocProb$ to $\minPathwidthProb$ from Section~\ref{sec:locToPW} gives a reduction from $\minCutwidthProb$ to $\minPathwidthProb$. Moreover, this reduction is approximation preserving; thus, it carries over approximations for $\minPathwidthProb$ (e.\,g.,~\cite{FeigeEtAl2008,GroenlandEtAl2023}) to $\minCutwidthProb$, and yields new results for $\minCutwidthProb$.

Pathwidth and cutwidth are classical graph parameters that play an important role for graph algorithms, independent from our application for computing the locality number. Therefore, it is the main purpose of this section to translate the reduction from $\minCutwidthProb$ to $\minPathwidthProb$ that takes $\minlocProb$ as an intermediate step into a direct reduction from $\minCutwidthProb$ to $\minPathwidthProb$. Such a reduction is of course implicitly hidden in the reductions of Sections~\ref{sec:locToCutwidth} and \ref{sec:locToPW}, but we believe that explaining the connection in a more explicit way will be helpful for researchers that are mainly interested in the graph parameters cutwidth and pathwidth.

The relationship between cutwidth and pathwidth revealed by this direct reduction is best illustrated via a third graph parameter that we call \emph{second order cutwidth}. To the best of our knowledge, this parameter has not explicitly been studied before. \par

Let $L = (v_1, v_2, \ldots, v_n)$ be a linear arrangement of a (multi)graph $G = (V, E)$. For the classical cutwidth, we consider the maximum number of edges that span over a gap between a vertex $v_i$ and $v_{i+1}$. For the second order cutwidth on the other hand, we consider the maximum number of edges that span over a vertex $v_i$ or that are adjacent to $v_i$, i.\,e., all edges $\{v_k, v_\ell\}$ with $k \leq i \leq \ell$ (note that this is equivalent to considering the maximum number of edges that span over the gap between $v_{i-1}$ and $v_{i}$ or over the gap between $v_i$ and $v_{i+1}$). Let us now formally define the second order cutwidth.\par
For every $i \in \{1, 2, \ldots, n\}$, we define $\Gamma_{L, i} = \{\{v_k, v_\ell\}\in E \mid k \leq i \leq \ell\}$. The \emph{second order cutwidth} of the linear arrangement $L$ is defined by $\socutwidth(L) = \max\{|\Gamma_{L, i}| \mid 1 \leq i \leq n\}$, and the second order cutwidth of $G$ is defined by $\socutwidth(G) = \min\{\socutwidth(L) \mid L \text{ is a linear arrangement for $G$}\}$. Since $\Gamma_{L, i} = \cutedges_L(i-1) \cup \cutedges_L(i)$ for every $i \in \{1, 2, \ldots, n\}$, we can also define the second order cutwidth in terms of the sets $\cutedges_L$, i.\,e., $\socutwidth(L) = \max\{|\cutedges_L(i-1) \cup \cutedges_L(i)| \mid 1 \leq i \leq n\}$.\par
For example, the linear arrangement $L$ on the top of Figure~\ref{fig:exampleCutwidth} has a second order cutwidth of $6$, which, e.\,g., is witnessed by $\Gamma_{L, 3}$, since $\Gamma_{L, 3} = \{\{u, w\}, \{u, x\}, \{v, x\}, \{v, y\}, \{v, z\}, \{w, x\}\}$. Note that $\Gamma_{L, 3} = \cutedges_L(2) \cup \cutedges_L(3)$. On the other hand, the linear arrangement $L'$ on the bottom has a second order cutwidth of $4$ (witnessed by $|\Gamma_{L', 2}| = 4$).\par

To understand the relationship between this parameter and cutwidth, we first show the following.

\begin{lemma}\label{cw_to_cw2}
Let $G=(V,E)$ be a connected graph with at least three vertices, then $\cutwidth(G)+1\leq\socutwidth(G)\leq 2 \cutwidth(G)$. 
Further, given a linear arrangement $L$ for $G$, a linear arrangement $L'$ for $G$ with $\socutwidth(L)-1\geq \cutwidth(L')$ can be computed in $\bigo(|E|)$. 
\end{lemma}

\begin{proof}
Simply by definition, we see that
\begin{align*} 
\socutwidth(L) &= \max\{\Gamma_{L, i} \mid 1 \leq i \leq n\} = \max\{|\cutedges_L(i-1) \cup \cutedges_L(i)| \mid 1 \leq i \leq n\}\\
&\leq \max\{|\cutedges_L(i-1)| + |\cutedges_L(i)| \mid 1 \leq i \leq n\}\leq 2\cutwidth(L)\,, 
\end{align*} 
for any linear arrangement $L$ for $G$. This directly gives the second inequality.

For the first inequality, observe that by definition $\cutwidth(L)\leq \socutwidth(L)$ for any linear arrangement $L$. Let $L = (v_1, v_2, \ldots, v_n)$ be a linear arrangement of minimum second order cutwidth, and assume $\cutwidth(L)=\socutwidth(L)$ (otherwise $L' := L$ shows the claim). We show how to construct a linear arrangement $L'$ of strictly smaller cutwidth than $L$ in polynomial time, by iterative rearrangements. To this end, we define $I_{max}(L)=\{v_t \mid |\cutedges_L(t)|= \socutwidth(G)\}$ and let $I_{max}^1(L)$ be the subset of $I_{max}(L)$ of degree one vertices, i.\,e., $I_{max}(L) = \{v_t \in I_{max}(L) \mid |N(v_t)| = 1\}$.
We now stepwise rearrange $L$ and use $I_{max}$ and $I^1_{max}$ to track progress, never increasing the second order cutwidth of the arrangement. Our goal is a linear arrangement $L'$ with $\socutwidth(L')=\socutwidth(G)$ where $I_{max}(L')=\emptyset$; note that this implies the desired $\cutwidth(L')\leq \socutwidth(G)-1$. Until we have reached this goal, we have a linear arrangement $L$  with $\cutwidth(L)=\socutwidth(L)$ implying $\cutedges_L(t) = \Gamma_{L, t}$ for each $t\in I_{max}(L)$. Since $\Gamma_{L, t} = \{\{v_k, v_\ell\}\in E \mid k \leq t \leq \ell\}$ and $\cutedges_L(t) = \{\{v_k, v_\ell\}\in E \mid k \leq t < \ell\}$, this also means that $v_t$ has no neighbour in $\{v_1,\dots, v_{t-1}\}$. Hence, $I_{max}(L)$ is an independent set. 

We change $L$ into a linear arrangement $L'$ as follows. Let $v_t \in I_{max}(L)$ be arbitrary, and let $v_\ell$ be the neighbour of $v_t$ in $G$ with the smallest index in $L$. We move $v_t$ directly to the right of the node $v_\ell$, i.\,e., we  define $L' := (v_1, \ldots, v_{t-1}, v_{t+1}, \ldots, v_{\ell}, v_t, v_{\ell+1}, \ldots, v_n)$. We also define $I_{max}(L')$ and $I_{max}^1(L')$ analogously as for $L$. By this definition, $v_t$ has position $\ell$ and $v_{\ell}$ has position $\ell-1$ in the new linear arrangement $L'$. 

By the choice of $\ell$, the cut of $v_{\ell}$ with respect to $L$ is the same as the cut of $v_t$ with respect to $L'$, i.\,e., $\cutedges_{L'}(\ell) = \cutedges_{L}(\ell)$. Among all other vertices, only $v_\ell$ can have a larger cut in $L'$ compared to $L$, since all other cuts either stay the same (the ones strictly to the right of $v_t$ in $L'$), or do not contain the edges of $v_t$ anymore (the ones strictly to the left of $v_{\ell}$ in $L'$). The only difference between the cuts $\cutedges_{L}(\ell)$ and $\cutedges_{L'}(\ell -1)$ are edges involving $v_t$. More precisely, all edges $\{v_t, u\}$ with $u \in N(v_t) \setminus \{v_{\ell}\}$ are in $\cutedges_{L}(\ell)$ (since the vertices $N(v_t) \setminus \{v_{\ell}\}$ are to the right of $v_{\ell}$ in $L$), while $\{v_t, v_{\ell}\}$ is not in $\cutedges_{L}(\ell)$ (since $v_t$ is to the left of $v_{\ell}$ in $L$). In $L'$, we have the opposite situation, i.\,e., none of the edges $\{v_t, u\}$ with $u \in N(v_t) \setminus \{v_{\ell}\}$ is in $\cutedges_{L'}(\ell-1)$, while $\{v_t, v_{\ell}\}$ is in $\cutedges_{L'}(\ell-1)$.

Summarizing, we see that all first or second order cuts for vertices $V\setminus\{v_t,v_\ell\}$ can only decrease from $L$ to $L'$. In particular -- which we will denote as property $(\dagger)$ for future reference -- we observe that $\cutedges_{L'}(\ell) = \cutedges_{L}(\ell)$, and $\cutedges_{L'}(\ell-1)=\cutedges_{L}(\ell)\cup\{v_t, v_{\ell}\}\setminus\{\{v_t,v_u\mid u \in N(v_t) \setminus \{v_{\ell}\}\}$.

Towards proving $\socutwidth(L')=\socutwidth(G)$ we only have to check the second order cuts of $v_\ell$ and $v_t$. Since $\{v_t, v_{\ell}\}\in \Gamma_{L,\ell}$, ($\dagger$) implies $\Gamma_{L',\ell-1}\subseteq \Gamma_{L,\ell}$.  Further, $\Gamma_{L',\ell}=\cutedges_{L'}(\ell-1)\cup \cutedges_{L'}(\ell)=\cutedges_{L}(\ell)\cup\{\{v_t, v_{\ell}\}\}\subseteq \Gamma_{L,\ell}$. Thus $\socutwidth(L)=\socutwidth(G)$ shows the claimed second order width of $L'$.

For the progress of reducing $I_{max}(L')$, first note that $v_\ell \notin I_{max}(L)$; recall that $v_{t} \in I_{max}(L)$ and $I_{max}(L)$ is an independent set. This implies that $|\cutedges_{L'}(\ell)| = |\cutedges_{L}(\ell)| \leq \socutwidth(G)-1$, so $v_t \notin I_{max}(L')$. Further, for all vertices except $v_\ell$, the cut values did not increase from $L$ to $L'$, so $I_{max}(L') \subseteq (I_{max}(L) \setminus \{v_t\}) \cup \{v_{\ell}\}$ and $I_{max}^1(L') \subseteq (I_{max}^1(L) \setminus \{v_t\}) \cup \{v_{\ell}\}$. \par
We now consider two cases depending on whether or not $|\cutedges_{L'}(\ell-1)|$ is strictly smaller than $|\cutedges_{L}(\ell)| + 1$.

\medskip

\noindent\textbf{Case 1}, $|\cutedges_{L'}(\ell-1)| < |\cutedges_{L}(\ell)| + 1$: Since $|\cutedges_{L}(\ell)| \leq \socutwidth(G) - 1$ (as observed above), this means that $|\cutedges_{L'}(\ell-1)| < \socutwidth(G)$. Hence, $v_{\ell} \notin I_{max}(L')$, which implies that $|I_{max}(L')| \leq |I_{max}(L)| - 1$.

\medskip

\noindent\textbf{Case 2}, $|\cutedges_{L'}(\ell-1)| = |\cutedges_{L}(\ell)| + 1$: By  ($\dagger$) this is only possible if $N(v_t) \setminus \{v_{\ell}\}=\emptyset$, which means that $v_t \in I_{max}^1(L)$. Since $G$ is a connected graph with at least three vertices, $v_\ell$ has at least one neighbour other than $v_t$, which means that $v_{\ell} \notin I_{max}^1(L')$. Consequently, $|I_{max}^1(L')| \leq |I_{max}^1(L)| - 1$.

\medskip

Thus, we have created a linear arrangement $L'$ with $\socutwidth(L')=\socutwidth(G)$ such that either $|I_{max}(L')|<|I_{max}(L)|$ or $|I_{max}(L')|=|I_{max}(L)|$ and $|I^1_{max}(L')|<|I^1_{max}(L)|$. Iterative application of this rearrangement converges to the desired linear arrangement $L'$ with $\socutwidth(L')=\socutwidth(G)$ and $I_{max}(L')=\emptyset$.

Computing $\mathcal C_L(i)$ for each $1\leq i\leq n$ and with this also the set $I_{max}(L)$ can be done in $\bigo(|E|)$. With this information, we only have to check the neighbourhood of the vertex that is moved for each rearrangement. Further, we only move vertices in $I_{max}$, so each vertex is moved at most once. Thus, all rearrangements needed to create $L'$ from $L$ can be done in $\bigo(|E|)$.
\end{proof}

Armed with this relationship, we now give the direct reduction from problem $\minCutwidthProb$ to problem $\minPathwidthProb$ that will in fact satisfy $\socutwidth(G)-1\leq  \pathwidth(G')\leq \socutwidth(G)$,  effectively linking cutwidth and pathwidth with the help of Lemma~\ref{cw_to_cw2}. This reduction is surprisingly simple. Let $G = (V, E)$ be a simple graph.\footnote{We discuss the case of multi-graphs later on.} We translate every original node $u \in V$ into a clique $\mathcal{K}(u) = \{u_v \mid v \in N(u)\}$, and we translate every original edge $\{u, v\}$ into an edge $\{u_v, v_u\}$. Hence, we replace each vertex $u$ by a clique of size $|N(u)|$, and these cliques are connected according to the original graph. More formally, $G$ is transformed into the graph $G' = (V', E')$ with $V' = \bigcup_{u \in V} \mathcal{K}(u)$ and $E' =\{\{u_v,v_u\}\mid \{u,v\}\in E\}\cup\{\{v_u,v_w\}\mid u, w \in N(v), u \neq w\}$. See Figure~\ref{fig:directReduction} for an illustration.\par

\begin{figure}
\begin{center}

\includegraphics{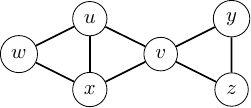}
\hspace{1cm}
\includegraphics{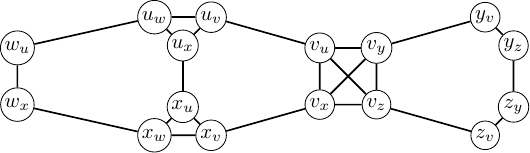}

\end{center}

\caption{A graph $G$ (left side) and the corresponding graph $G'$ obtained by the reduction (right side). Note that vertex $v$ of degree $4$ becomes a clique $\{v_u, v_y, v_x, v_z\}$, where $v_u$ is connected to $u_v$ (which is a vertex of the $3$-clique representing degree-$3$ vertex $u$), $v_y$ is connected to $y_v$ and so on.}
\label{fig:directReduction}
\end{figure}

We proceed by showing the second inequality for our goal $\socutwidth(G)-1\leq  \pathwidth(G')\leq \socutwidth(G)$.

\begin{lemma}\label{cw_to_pw1}
Let $G$ be a graph with at least one edge, then $\pathwidth(G')\leq \socutwidth(G)$.
\end{lemma}

\begin{proof}
Consider a graph $G=(V,E)$ and let $L=(v_1,\dots, v_n)$ be an optimal linear arrangement for the second order cutwidth of $G$. Every node $u_w$ of $G'$ is either a \emph{left node} or a \emph{right node} according to whether $u$ occurs to the left or to the right of $w$ in the linear arrangement $L$. More formally, a vertex $u_w$ of $G'$ is a \emph{left node}, if $u = v_j$ and $w = v_\ell$ with $j < \ell$; the term \emph{right node} is defined analogously. Obviously, $u_w$ is a right node if and only if $w_u$ is a left node.\par

To prove the claimed bound on the pathwidth of $G'$, we construct a path decomposition for $G'$ of width at most $\socutwidth(G)$ in the form of a pd-marking scheme.

Intuitively, the pd-marking scheme is as follows. Let us first recall that, for every $i \in \{1, 2, \ldots, n\}$, $\Gamma_{L, i} = \{\{v_k, v_\ell\}\in E \mid k \leq i \leq \ell\}$, and that $\socutwidth(G)$ is the maximum over the cardinalities of these sets. For every $i = 1, 2, \ldots, n$, we produce a step $i$ of the marking scheme, where every edge $\{u, v\} \in \Gamma_{L, i}$ is represented by having the right node of $\{u_v, v_u\} \in E'$ set to $\act$ (and these are the only $\act$ vertices) and the left node set to $\closed$. Moreover, for all edges $\{u, v\} \in E$ such that $u, v \in \{v_1, v_2, \ldots, v_{i-1}\}$, both $u_v$ and $v_u$ are $\closed$, and all other vertices are $\open$. This means that the number of $\act$ vertices corresponds to $|\Gamma_{L, i}|$. In order to conclude the prove, it suffices to show that (1) we can obtain step $i+1$ from step $i$ with at most $|\Gamma_{L, i+1}| + 1$ vertices being $\act$ at the same time, and (2) that the cover property is satisfied. Both claims follow from how we obtain step $i+1$ from $i$: Let $v := v_{i+1}$. We first set all $\act$ right nodes from $\mathcal{K}(v_{i-1})$ to $\closed$, then we set all $\open$ \emph{left} nodes $v_u \in \mathcal{K}(v)$ to $\act$ (now all vertices from $\mathcal{K}(v)$ are $\act$ at the same time as required by the cover property), then for every such left node $v_u \in \mathcal{K}(V)$, we set the right node $u_v$ to $\act$ and then $v_u$ to $\closed$ (this is done one by one, to get at most one $|\Gamma_{L, i + 1}| + 1$ $\act$ vertices). Now we have reached step $i+1$. Let us now define the pd-marking scheme more formally and prove its correctness.

For every $v = v_1, v_2, \ldots, v_n$, we perform the following steps.
\begin{itemize}
\item Step $1(v)$: Set all $\open$ left nodes from $\mathcal{K}(v)$ to $\act$.
\item Step $2(v)$: For every left node $v_u \in \mathcal{K}(v)$, set the right node $u_v \in \mathcal{K}(u)$ from $\open$ to $\act$, and then set $v_u$ from $\act$ to $\closed$.
\item Step $3(v)$: Set all $\act$ right nodes from $\mathcal{K}(v)$ to $\closed$.
\end{itemize}

Note that in our intuitive explanation above, Step $3(v)$ is the first operation that we have to do in order to get from step $i$ to step $i + 1$. More precisely, after finishing Step $2(v)$ we have reached the situation that has been called step $i$ above. It is simpler to state the Steps $1(v)$, $2(v)$ and $3(v)$ in this way, since these are exactly the operations that are with respect to the clique $\mathcal{K}(v_i)$.

By induction, it can be easily seen that after Step $1(v)$ all vertices from $\mathcal{K}(v)$ are $\act$ (i.\,e., before Step $1(v)$, only the left nodes are still $\open$), after Step $2(v)$ all right nodes from $\mathcal{K}(v)$ are still $\act$, but all left nodes from $\mathcal{K}(v)$ are $\closed$, and after Step $3(v)$ all vertices of $\mathcal{K}(v)$ are $\closed$.\par
Let us now prove that the marking scheme described above is a valid pd-marking scheme. Our considerations already show that we set every vertex $v \in V'$ from $\open$ to $\act$ and then from $\act$ to $\closed$. Now let $\{p_q, r_s\}$ be an arbitrary edge of $G'$. If $p_q, r_s \in \mathcal{K}(v)$ for some $v \in V$, then both $p_q$ and $r_s$ are $\act$ after Step $1(v)$. If there is no $v \in V$ with $p_q, r_s \in \mathcal{K}(v)$, then, by definition of $G'$, $p_q = u_w$ and $r_s = w_u$ with $\{u, w\} \in E$. Let us assume that $u_w$ is a left node, which means that $w_u$ is a right node. After Step $1(u)$, the vertex $u_w$ is $\act$. Then, in Step $2(u)$, we set $w_u \in \mathcal{K}(w)$ to $\act$, before setting $u_w$ to $\closed$. Thus, both $u_w$ and $w_u$ are $\act$ at the same time. The case where $w_u$ is a left node and $u_w$ is a right node can be handled analogously. Consequently, the above defined marking scheme is a valid pd-marking scheme, i.\,e., it describes a valid path decomposition of $G'$.\par
It remains to estimate the width of the path decomposition, i.\,e., the maximal number of vertices that are active at the same time. We formulate the invariant: for every $i \in \{1, 2, \ldots, n\}$, as soon as Step $2(v_i)$ is done, every $\act$ vertex $u_w$ is a right node such that $\{u, w\} \in \Gamma_{L, i}$.\par
First, we note that the invariant holds after Step $2(v_1)$, since then the set of $\act$ vertices is $\{w_{v_1} \mid w \in N(v_1)\}$ (which are all right vertices) and $\Gamma_{L, 1} = \{\{v_1, w\} \mid w \in N(v_1)\}$. Let us now assume that the invariant holds for some $i \in \{1, 2, \ldots, n-1\}$. \par
Let $u_w$ be a vertex that is $\act$ after Step $2(v_{i+1})$. If $u_w$ was already $\act$ immediately after Step $2(v_{i})$, then $u_w$ is a right node and $\{u, w\} \in \Gamma_{L, i}$. Moreover, $u_w \notin \mathcal{K}(v_i)$, since then it would have been set to $\closed$ in Step $3(v_{i})$. This means that $\{u, w\} \in \Gamma_{L, i + 1}$. If, on the other hand, $u_w$ was not already $\act$ immediately after Step $2(v_{i})$, then it has been set to $\act$ in Step $2(v_{i+1})$ (note that all vertices set to $\act$ in Step $1(v_{i+1})$ are set to $\closed$ in Step $2(v_{i+1})$). This means that it is a right node and that $\{u, w\}$ is in $\Gamma_{L, i + 1}$. Hence, as soon as Step $2(v_{i+1})$ is done, every $\act$ vertex $u_w$ is a right node such that $\{u, w\} \in \Gamma_{L, {i+1}}$.\par
By induction, this proves the invariant.\par
Let us now estimate the maximum number of $\act$ vertices in the entire pd-marking scheme. For every $i \in \{1, 2, \ldots, n\}$, the number of $\act$ vertices after Step $2(v_i)$ is bounded by $|\Gamma_{L, i}|$ (due to the invariant). It can be easily seen that carrying out Steps $3(v_{i})$, $1(v_{i+1})$ and $2(v_{i+1})$ produces a maximum of $|\Gamma_{L, i + 1}| + 1$ $\act$ vertices. More precisely, after setting some $\act$ right nodes from $\mathcal{K}(v_i)$ to $\closed$ in Step $3(v_i)$, we set a number $p$ of left nodes to $\act$ in Step $1(v_{i+1})$, which are then all set to $\closed$ in Step $2(v_{i+1})$, and instead we set $p$ right nodes to $\act$ (which then each account for one of the $\act$ vertices immediately after Step $2(v_{i+1})$). However, a left node $u_w$ is set to $\closed$ \emph{immediately} after the corresponding right node $w_u$ is set to $\act$; thus, we only need one additional $\act$ vertex, i.\,e., both $u_w$ and $w_u$ are $\act$ at the same time. Consequently, the maximum number of $\act$ vertices of the entire pd-marking scheme is $\max\{|\Gamma_{L, i}| \mid 1 \leq i \leq n\} + 1$, which means that its width equals $\socutwidth(L)$. 
\end{proof}

We now give the second part of the relationship between pathwidth and the second order cutwidth. Note that we also state this result constructively, to later use it for transferring approximations.

\begin{lemma}\label{cw_to_pw2}
Let $G=(V,E)$ be a graph with at least one edge, then $\pathwidth(G')\geq \socutwidth(G)-1$. Further, given a path decomposition $Q$ for $G'$, a linear arrangement $L$ for $G$ with $\socutwidth(L) \leq \width(Q) + 1$ can be constructed in $\bigo(|Q|)$. 
\end{lemma}

\begin{proof}
Let $Q$ be a path decomposition for $G'$, which we consider in the form of a pd-marking scheme with $p$ steps. For every $v \in V$, let $\phi(v) \in \{1, 2, \ldots, p\}$ be minimal such that all vertices from $\mathcal{K}(v)$ are $\act$ at step $\phi(v)$ (since $\mathcal{K}(v)$ is a clique, such a $\phi(v)$ must exist). Note that $\phi(v) \neq \phi(v')$ for every $v, v' \in V$ with $v \neq v'$ (since from one step to the next at most one vertex is changed). Let $L = (v_1, v_2, \ldots, v_n)$ be the linear arrangement of $G$ induced by the indices $\phi(v)$, i.\,e., for every $i, j \in \{1, 2, \ldots, n\}$, $\phi(v_i) < \phi(v_j)$ if and only if $i < j$. We note that $L$ can be created from $Q$ in time $\bigo(|Q|)$.

We will show that $\width(Q) \geq \socutwidth(L) - 1$ (since $Q$ is an arbitrary path decomposition, this proves the statement of the lemma). To this end, we will show that for every $i \in \{1,\dots,n\}$ and every edge $\{u,w\} \in \Gamma_{L, i} = \{\{v_k,v_\ell\}\in E \mid k \leq i \leq \ell\}$, the vertex $u_w$ or $w_u$ is $\act$ at step $\phi(v_i)$ of $Q$. Since the number of $\act$ vertices at any step of $Q$ is bounded by $\width(Q) + 1$, this implies that $|\Gamma_{L, i}| \leq \width(Q) + 1$, which directly implies that $\socutwidth(L) \leq \width(Q) + 1$.

Let $i \in \{1,\dots,n\}$ and let $\{u, w\} \in \Gamma_{L, i}$. For every $x \in V'$, let $I_{x} \subseteq \{1, 2, \ldots, p\}$ be the set of all steps of $Q$ in which $x$ is $\act$. By definition, the sets $I_{x}$ are intervals over $\{1, 2, \ldots, p\}$, and we know that $\phi(u) < \phi(v_i) < \phi(w)$. Since $I_{u_w}$ contains $\phi(u)$, $I_{w_u}$ contains $\phi(w)$, and $I_{u_w} \cap I_{w_u} \neq \emptyset$ (since there has to be a step where both $u_w$ and $w_u$ are $\act$), we conclude that $\{\phi(u), \phi(u) + 1, \ldots, \phi(w)\} \subseteq I_{u_w} \cup I_{w_u}$. Thus, we also have that $\phi(v_i) \in I_{u_w} \cup I_{w_u}$. This means that $u_w$ or $w_u$ is $\act$ at step $\phi(v_i)$ of $Q$.
 \end{proof}

As we want to transfer approximations for $\minPathwidthProb$ to $\minCutwidthProb$ also for multigraphs, we briefly explain how all results of this section easily generalize to this setting. The relationship of second order cutwidth and cutwidth remains exactly the same; the proof of Lemma~\ref{cw_to_cw2} generalizes to multigraphs with the only adjustment that $I_{max}^1(L)$ is defined as the set of vertices in $I_{max}(L)$ that have exactly one neighbour (that can be connected by multiple edges, so in this sense not of degree one).

For the connection to pathwidth, we can extend the reduction described before Lemma~\ref{cw_to_pw1} to multigraphs in a straightforward way. Let $G$ be a multigraph and let $\{u, v\}$ be an edge of $G$ with some multiplicity $t$ (i.\,e., in $G$ there are~$t$ parallel edges going from $u$ to $v$). While in the simple graph case an edge $\{u, v\}$ of $G$ was translated into the single edge $\{u_v, v_u\}$ of $G'$, we will now use $t$ simple edges $\{u^i_v, v^i_u\}$, $1\leq i \leq t$, in order to represent the multiplicity $t$ of the multi-edge of $G$. Hence, we represent a single vertex $v$ of $G$ with $N(v) = \{u_1, u_2, \ldots, u_k\}$ by several vertices $\mathcal{K}(v) = \{v^1_{u_1}, \ldots, v^{t_1}_{u_1}, v^1_{u_2}, \ldots, v^{t_2}_{u_2}, \ldots, v^1_{u_k}, \ldots, v^{t_k}_{u_k}\}$, where the $t_1, t_2, \ldots, t_k$ are the multiplicities of the edges between $v$ and its neighbours $u_1, u_2, \ldots, u_k$. Analogously to the simple graph case, we connect all the vertices of $\mathcal{K}(v)$ into a clique.

We can now prove Lemma~\ref{cw_to_pw1} for the case of multi-graphs in a similar way as for simple graphs. For a given multi-graph $G$, we apply the reduction from above, which yields a \emph{simple} graph $G'$ (recall that the multiplicities are represented by individual vertices in the cliques $\mathcal{K}(v)$ with $v \in V$). Then, we fix again an optimal linear arrangement $L=(v_1,\dots, v_n)$ for the second order cutwidth of $G$. We call a node $u^i_w$ a \emph{left node} if $u$ occurs to the left of $w$ with respect to $L$, and \emph{right nodes} are defined analogously. Now, we can define a pd-marking scheme in the same way as in the proof of Lemma~\ref{cw_to_pw1}, i.\,e., for every $v = v_1, v_2, \ldots, v_n$, we perform the Steps $1(v)$, $2(v)$ and $3(v)$. Since also in the adapted reduction we have the cliques $\mathcal{K}(v)$, both Step $1(v)$ (i.\,e., set all $\open$ left nodes from $\mathcal{K}(v)$ to $\act$) and Step $3(v)$ (i.\,e., set all $\act$ right nodes from $\mathcal{K}(v)$ to $\closed$) apply verbatim in the same way, while Step $2(v)$ reads as follows: For every left node $v^i_u \in \mathcal{K}(v)$, set the right node $u^i_v \in \mathcal{K}(u)$ from $\open$ to $\act$, and then set $v^i_u$ from $\act$ to $\closed$. The proof that this gives a path decomposition of width at most $\socutwidth(G)$ then is completely analogous. \par
The same holds for the proof of Lemma~\ref{cw_to_pw2}. 

We are now ready to state the main result of this section, i.\,e., how pathwidth approximation carries over to cutwidth approximation.

\begin{lemma}\label{cw_to_pw-apx}
If there is an $r(\opt,|V|)$-approximation algorithm for $\minPathwidthProb$ with running-time $\bigo(f(|V|))$, then there is also an $2r(2\opt,h)$-approximation algorithm for $\minCutwidthProb$ on multigraphs with running time  $\bigo(f(h)+h^2+n)$, where $n$ is the number of vertices and $h$ is the number of edges.
\end{lemma}

\begin{proof}
Let $G=(V,E)$ be an instance of $\minCutwidthProb$ and let $\mathcal A$ be an $r(\pathwidth(G'),|V|)$-approximation for $\minPathwidthProb$  Lemma~\ref{cw_to_pw1} combined with Lemma~\ref{cw_to_cw2} shows that $\pathwidth(G')\leq \socutwidth(G)\leq 2\cutwidth(G)$. Further,  Lemma~\ref{cw_to_pw2} shows that any path-decomposition $P$ of width $k$ for $G'$ can be translated into a linear arrangement $L$ for $G$ with $\socutwidth(L)\leq k+1$ in $\bigo(|P|)$.  By  Lemma~\ref{cw_to_cw2}, we can compute then from $L$ a linear arrangement $L'$ with $\cutwidth(L')\leq \socutwidth(L)-1\leq k$  in $\bigo(h)$.

The relative error of $L'$ can thus be bounded by $R(G,L)=\frac{\cutwidth(L')}{\cutwidth(G)}\leq \frac{2\pathwidth(P)}{\pathwidth(G')}=2R(G',P)$.
The algorithm which builds $G'$ from $G$ in $\bigo(n+h)$, runs $\mathcal A$ on $G'$ in $\bigo(f(h))$ and creates a linear arrangement $L'$ in $\bigo(h+|P|)$ has a performance ratio $2r(\pathwidth(G'),|V|)\leq 2r(2\cutwidth(G),h)$ and an overall running time in $\bigo(f(h)+h)$ (note that $\bigo(|P|)\subseteq \bigo(f(h))$, since $\mathcal A$ builds $P$ 
 in $\bigo(f(|V(G')|))=\bigo(f(h))$).
\end{proof}

For example, if we apply this lemma with respect to the $\bigo(\log n\sqrt{\log opt})$-approximation algorithm of~\cite{FeigeEtAl2008}, we obtain an $\bigo(\sqrt{\log(\opt)} \log(h))$-approximation algorithm for $\minCutwidthProb$ on multigraphs with $h$ edges, and if we apply it with respect to the $\bigo(\mathsf{tw} \sqrt{\log \mathsf{tw}})$-approximation algorithm of~\cite{GroenlandEtAl2023}, we obtain an $\bigo(\sqrt{\log(\opt)} \opt)$-approximation algorithm. Note that the second result holds since an $\bigo(\mathsf{tw} \sqrt{\log \mathsf{tw}})$-approximation algorithm for pathwidth is also an $\bigo(\opt \sqrt{\log \opt})$-approximation algorithm for pathwidth. Unfortunately, our reduction blows up the treewidth, so it does not give a translation to a ratio that depends only on the treewidth.

To the best knowledge of the authors, these are new approximations ratios for cutwidth that have not previously been reported in the literature, and that are better or incomparable to existing ones. Hence, let us state this result more prominently.

\begin{corollary}\label{thm:approxMultiCut}
There is a (polynomial-time) $\bigo(\sqrt{\log(\opt)} \log(h))$-approximation algorithm and an $\bigo(\sqrt{\log(\opt)} \opt)$-approximation algorithm for $\minCutwidthProb$ on multigraphs with $h$ edges.
\end{corollary}

\section{Conclusions}\label{sec:conc}

In this work, we have answered several open questions about the string parameter of the locality number. Our main tool was to relate the locality number to the graph parameters cutwidth and pathwidth via suitable reductions. As an additional result, our reductions also pointed out an interesting relationship between these classical graph parameters and the locality number for strings, with implications for approximating these parameters.\par
While our focus is on theoretical results in form of lower and upper complexity bounds, we stress here that the reductions may also be of practical interest, since they allow to transform any practical pathwidth or cutwidth algorithm into a practical algorithm for computing the locality number (or to transform a practical pathwidth algorithm into a practical algorithm for computing the cutwidth). This seems particularly interesting, since, as pointed out at the end of Section~\ref{sec:locToPW}, practical algorithms for constructing path decompositions of small width is a vibrant research area of practical algorithm engineering.


\appendix

\section{Additional Word-Combinatorial Considerations}
\label{combwo}

In this section, we give the details that have been omitted in Section~\ref{sec:ExamplesWordComb}.

\subsection{The Locality of the Zimin Words}

\begin{lemma}
$\loc(Z_i)=\frac{|Z_i|+1}{4}=2^{i-2}$ for $i\in\MN_{\geq 2}$.
\end{lemma}

\begin{proof}
Clearly, $x_1$ and $x_1x_2x_1$ are $1$-local.
Consider a fixed $i\in\MN$ and the marking sequence $(x_2,x_1,y_1,y_2,\dots,y_{i-2})$ for $i\geq 3$ and $\{y_1,\dots,y_{i-2}\}=\{x_3,\dots,x_i\}$. Notice
that for all $j\in\MN$, $x_j$ occurs $2^{i-j}$ times in $Z_i$. Thus by marking $x_2$, there are 
$2^{i-2}$ marked blocks. Since all occurrences of $x_1$ are adjacent to occurrences of $x_2$, 
marking $x_1$ does not change the number of marked blocks. As marking the remaining variables only leads to 
the merging of some pairs of consecutive blocks into one, we never have more than $2^{i-2}$ marked blocks. 

In the following we will show the converse. More precisely, we show that if a sequence is optimal for $Z_i$ then it starts with $x_2,x_1$. Let us note first that, for $2\leq p<r$, between two consecutive occurrences of $x_r$ in $Z_i$ there is one occurrence of $x_p$. More precisely, each occurrence of a variable $x_p$, with $p\geq 2$, is directly between two occurrences of $x_1$. Also, notice that $x_j$ has $2^{i-j}$ occurrences in $Z_i$.  Now, if $x_1$ is marked before $x_2$, because $Z_i$ starts with $x_1x_2$ and ends with $x_2x_1$, it is immediate that after the marking of $x_1$ we will have at least $2^{i-2}+1$ marked blocks in the word (separated by the $2^{i-2}$ unmarked occurrences of $x_2$). This is, thus, a marking sequence that is not optimal. So $x_2$ is marked before $x_1$ in an optimal sequence. Assume that there exists $x_j$, with $j>2$, which is also marked before $x_1$ in an optimal sequence. Let $w$ be a word such that $Z_i=x_1wx_1$. There are $2^{i-1}-2$ occurrences of $x_1$ in $w$, and $w$ starts with $x_2x_1$ and ends with $x_1x_2$. As each two consecutive (marked) occurrences of the letters $x_2$ and $x_j$ are separated by unmarked occurrences of $x_1$ we have that, just before marking $x_1$, there are at least $\min\{2^{i-1}-1,2^{i-2}+2^{i-j}\}$ marked blocks in $w$ (and the same number in $Z_i$). This again shows that this is not an optimal marking sequence. So, before $x_1$ is marked, only $x_2$ should be marked. This concludes the proof of our claim, and of the proposition.
\end{proof}

\subsection{The Locality of (Condensed) Palindromes and Repetitions.}

We use the following notation. Given a marking sequence $\sigma$, let $\sigma^R$ be the marking sequence obtained by reversing $\sigma$ (i.e. $\sigma^R(i) = \sigma(|X|-i+1)$ for $1 \leq i \leq |X|$). 

By $\loc(\condensed(w))=\loc(w)$, it is enough to show our results for condensed words. Since there are no condensed palindromes of even length, only palindromes of odd length are of interest when determining the locality number. A word $w$ is called strictly $k$-local if for every optimal marking sequence of $w$ there is a stage when exactly $k$ factors are marked. For a letter $a\in \alphabet(w)$, we denote by $|w|_a$ the number of occurrences of $a$ in $w$.  For simplicity of notations, let $[n]:=\{1,2,\ldots,n\}$.

Let $w_i\in(X\cup\overline{X})^{\ast}$ be the marked version of $w$ at stage $i\in[|\alphabet(w)|]$ for a given marking sequence $\sigma$.

\begin{lemma}\label{lem01}
Define the morphism $f:X\cup\overline{X}\rightarrow\{0,1\}$ by 
\[
f(x)=\begin{cases}
0&\mbox{if }x\in X,\\
1&\mbox{if }x\in\overline{X}.
\end{cases}
\]
If $w$ is a palindrome and $\sigma$ a marking sequence for $w$ then $f(w_i)$ is a palindrome for 
all $i\in[|\alphabet(w)|]$.
\end{lemma}

\begin{proof}
Let $w=ux u^R$ be a palindrome with $u\in X^{\ast}$ and $x\in X\cup\overline{X}$ and 
$|w|=n\in\N$. Moreover let $\sigma$ be a marking sequence for $w$ and $i\in[|\alphabet(w)|]$. Since $w$ 
is a palindrome, $w[j]=w[n-j]$. This implies $w_i[j],w_i[n-j]$ are both either in $X$ or in 
$\overline{X}$. Thus either are both mapped to $0$ or to $1$. Consequently $f(w_i)$ is a 
palindrome.
\end{proof}

\medskip

Recall the definition of {\em border priority markable} from \cite{FSTTCS}. A strictly $k$-local word $w=avb\in XX^{\ast}X$ is called border priority markable if there exists a marking sequence $\sigma$ of $w$ such that 
in every stage $i\in [|\alpha(w)|]$ of $\sigma$ where $k$ blocks are marked, $a$ and $b$ are marked as well. Analogously right-border 
priority markable and left-border priority markable are defined: A strictly $k$-local word 
$w=avb\in XX^{\ast}X$ is called right-border priority markable (rbpm) if  if there exists a marking sequence $\sigma$ of $w$ such that 
in every stage $i\in[|\alpha(w)|]$ of $\sigma$ where $k$ blocks are marked, $b$ is marked as well - respectively, for 
left-border priority markable, $a$ is marked as well.

\begin{remark}
If $w\in X^{\ast}$ is right-border priority markable, then $u^R$ is left-border priority markable.
\end{remark}

\begin{lemma}
Let $w=u\ta u^R$ be an odd-length condensed palindrome with $u\in X^{\ast}$ and $\ta\in X$. Let $u$ be 
strictly $k$-local witnessed by the marking sequence $\sigma$.
\begin{itemize}
\item If $u$ is rbpm then $\loc(w)=2k-1$,
\item if $u$ is not rbpm and $\ta\not\in\alphabet(u)$ then $\loc(w)=2k$,
\item if $u$ is not rbpm and $\ta\in\alphabet(u)$ and for all optimal marking 
sequences for $u$ there exists 
a stage $i\in[|\alphabet(u)|]$ such that $\ta$ is marked, $k$ blocks are marked,  and $u[|u|]$ is 
unmarked then $\loc(w)=2k+1$, and
\item else $\loc(w)=2k$.
\end{itemize}
\end{lemma}

\begin{proof}
Let $\sigma$ be an optimal marking sequence of $u$. If $\ta\in\alphabet(u)$ then $\sigma$ is a marking 
sequence for $w$. Marking $w$ w.r.t. $\sigma$ leads to $\pi_{\sigma}(w)\leq 2k+1$ since there are at 
most maximal $k$ blocks marked each in $u$ and $u^R$, and additionally the single $\ta$ in the 
middle. If $\ta\not\in\alphabet(u)$ then $\sigma'=\sigma\cup\{(|u|+1,\ta\}$ is a marking sequence for $w$ 
with $\pi_{\sigma'}(w)\leq 2k$, since by marking w.r.t. $\sigma$ maximal $k$ blocks are marked by 
$\sigma$ each in $u$ and $u^R$ and afterwards on marking $a$ two blocks are joined. Thus in any 
case $\loc(w)\leq 2k+1$.

\noindent \textbf{case 1.} Consider $u$ to be rbpm. Thus in every stage $i\in[|\alphabet(u)|]$ where $k$ blocks are marked, 
$u[|u|]$ is marked. This implies that $\pi_{\sigma}(w)\leq 2k-1$ or $\pi_{\sigma'}(w)\leq 2k-1$ 
with $\sigma'$ defined as above.\\
\textbf{Supposition}: $\loc(w)=:\ell<2k-1$\\
Let $\mu$ be an optimal marking sequence for $w$. Then $\mu$ is also a marking sequence for $u$ and 
thus $\pi_{\mu}(u)\geq k$. By $\loc(u)=k$ there exists a stage $i\in[|\alphabet(w)|]$ of $\mu$ such that 
$k$ blocks are marked in $u$, or more precisely $|\condensed(f(u_i))|_1=k$. On the other hand 
$|\condensed(f(w_i))|_1\leq\ell$. Since $u$ is rbpm $u[|u|]$ is marked. If $x$ is not marked,  
$|\condensed(f(u_i))|_1\leq\frac{\ell}{2}<\frac{2k-1}{2}=k-\frac{1}{2}$.  If $x$ is marked, 
$|\condensed(f(u_i))|_1\leq\frac{\ell-1}{2}<\frac{2k-2}{2}=k-1$.
This is in both cases a contradiction to $|\condensed(f(u_i))|_1=k$.

\noindent \textbf{case 2.}  Consider now that $u$ is not rbpm. Thus there exists a stage $i\in[|\alphabet(u)|]$ in which $k$ blocks 
are marked but $u[|u|]$ is unmarked. If $\ta$ is not in $\alphabet(u)$ marking $\ta$ before stage $i$ 
leads to $2k+1$ blocks for the largest such $i$. Considering $\sigma'$ then at the beginning $u$ 
and $u^R$ are completely marked and in the end two blocks are joined by marking $\ta$. This leads 
to $\loc(w)\leq 2k$.\\
\textbf{Supposition}: $\loc(w)<2k$\\
As described, $\ta$ needs to be marked after the last stage where in $u$ $k$ blocks are marked 
without $u[|u|]$ being marked. But this sums up to $k$ blocks marked in $u$ and $k$ blocks marked 
in $u^R$, hence overall $2k$ blocks. This concludes the case $\ta\not\in\alphabet(u)$.\\
Consider $\ta\in\alphabet(u)$ and assume that $\ta$ is marked by $\sigma$ when  $k$ blocks are marked in 
$u$ and $u[|u|]$ is unmarked. Thus 
$\pi_{\sigma}(w)=2k+1$.\\
\textbf{Supposition}: $\loc(w)=:\ell<2k+1$\\
Let $\mu$ be an optimal marking sequence for $w$. \\
\textbf{Additional supposition}: $\mu$ not optimal for $u$\\
Then there exists a stage $i\in[\alphabet(w)]$ such that $|\condensed(f(u_i))|_1=k+1$. If $\ta$ is unmarked 
in this stage, $|\condensed(f(w_i))|_1=2k+2>\ell$ which contradicts the first supposition. If $\ta$ is 
marked in this stage $|\condensed(f(w_i))|_1=2k+1$ which  contradicts the first supposition.\\
Thus, $\mu$ is optimal for $u$. By assumption there exists a stage $i\in[|\alphabet(u)|]$ such that 
$\ta$ is marked, $k$ blocks are marked, and $u[|u|]$ is unmarked.
This implies since $\condensed(f(w_i))$ is a palindrome that at most $\frac{\ell-1}{2}$ blocks are 
marked in $u$. Thus, $k\leq\frac{\ell-1}{2}<\frac{2k+1-1}{2}=k$.

\noindent \textbf{case 3.} In the remaining case $u$ is not rbpm, $\ta\in\alphabet(u)$, and there exists an optimal marking 
sequence for $u$ such that in every stage $\ta$ is unmarked or less than $k$ blocks are marked or 
$u[|u|]$ is marked. Let $\sigma$ be such a marking sequence. Then $\pi_{\sigma}(w)=2k$.\\
\textbf{Supposition}: $\loc(w)=:\ell<2k$\\
Let $\mu$ be an optimal marking 
sequence for $w$.
Since $u$ is not rbpm there exists a stage $i\in[|\alphabet(u)|]$ such that $|\condensed(f(u_i))|_1=k$ and 
$u[|u|]$ is unmarked. If $\ta$ were unmarked in stage $i$, 
$k=|\condensed(f(u_i))|_1\leq\frac{\ell}{2}<k$ and if $\ta$ were marked in stage $i$, 
$k=|\condensed(f(u_i))|_1\leq\frac{\ell-1}{2}<\frac{2k-1}{2}=k-\frac{1}{2}$. Thus $2k+1\leq \ell<2k$ 
would hold.
\end{proof}

\begin{lemma}
Let $w=u^i$ be the $i$-times repetition for $u\in X^{\ast}$ and $i\in\N$. If 
$u$ is strictly $k$-local then 
\[
\loc(w)=
\begin{cases}
ik-i+1,&\mbox{ if } u \mbox{ is bpm},\\
ik,&\mbox{otherwise.}
\end{cases}
\]
\end{lemma}

\begin{proof}
Let $\sigma$ be a marking sequence with $\pi_{\sigma}=\loc(u)=k$. Since 
$\alphabet(u)=\alphabet(u^i)$ for all $i\in\N$, $\sigma$ is also a marking 
sequence for $w$. If $u$ is not bpm, there exists a stage during the marking in 
which $k$ blocks are marked by $\sigma$ and at least one of $u[1]$ or $u[|u|]$ is unmarked. 
Thus marking $w$ according to the sequence $\sigma$ leads to $\pi_{\sigma}(w)=ik$.  
If $u$ is bpm, in any stage in which $k$ blocks are marked, $u[1]$ and $u[|u|]$ are marked and thus in 
$w$, while being marked according to $\sigma$, the last marked block of an occurrence of $u$ and the 
first marked block of the next occurrence of $u$ coincide, as soon as the prefix of length $|u|$ of $w$ contains $k$ marked blocks.
So, we get $\pi_{\sigma}(w)=ik-i+1$. 

For proving $\loc(w)=ik$ or $\loc(w)=ik-i+1$ respectively, consider firstly $i=2$.  
Assume first that $w$ is bpm. Suppose $\loc(w)=\ell< 2k-1$. Let $\sigma'$ be the marking sequence witnessing $\loc(w)=\ell$. Since 
$u$ is strictly $k$-local, there exists a stage in marking $w$ by $\sigma'$ in 
which $u$ has $k$ marked blocks. The second $u$ has exactly as many marked blocks as the first one, so also $k$. In the best case, in $w$ the last marked block of the first $u$ and the first marked block of the second $u$ are connected. 
Anyway, the number of marked blocks of $w$ is, in that case, exactly $2k-1$. A contradiction to the assumption $\loc(w)=\ell< 2k-1$. 
If $u$ is not bpm, then, once again, there exists a stage in marking $w$ by $\sigma'$ in 
which $u$ has $k$ marked blocks. The second $u$ has also exactly $k$ marked block. But, in this case, in $w$ the last marked block of the first $u$ and the first marked block of the second $u$ do not touch (as either the last letter of $u$ or its first letter are not marked). So $w$ has $2k$ marked blocks, a contradiction. 

This reasoning can be trivially extended for $i>2$. 
\end{proof}

\end{document}